\documentclass[12pt]{article}
\usepackage[a4paper, left=1in, right=1in, top=1in, bottom=1in]{geometry}
\usepackage{amsmath, amsthm, amssymb, amsfonts}
\usepackage{tabulary}
\usepackage{bbm}
\usepackage{multirow}

\usepackage{booktabs}
\usepackage{xcolor}
\usepackage{mathtools}
\usepackage{upgreek}
\usepackage{listings}
\usepackage{natbib}
\usepackage{array}
\usepackage{enumitem}

\newtheorem{property}{Properties}
\newtheorem{theorem}{Theorem}[section]

\newtheorem{prop}{Proposition}[section]
\newtheorem{cor}{Corollary}[section]

\usepackage{algorithm} 
\usepackage{algpseudocode} 
\algtext*{EndFor}  
\usepackage{caption}
\usepackage{subcaption}

\def\P{{\mathbb P}}
\def\1{{\mathcal 1}}

\def\E{{\mathbb E}}

\usepackage{hyperref}
\hypersetup{citecolor=blue}
\hypersetup{
    colorlinks=true,
    linkcolor=blue,
    filecolor=magenta,      
    urlcolor=cyan,
    pdftitle={Overleaf Example},
    pdfpagemode=FullScreen,
    }

\providecommand{\keywords}[1]{
  \small	
  \textbf{\textit{Keywords---}} #1
}

\title{\textbf{On improving the estimates of the sampling variances via Global-Local priors in Small Area Estimation }}

\author{Sirapat Watakajaturaphon and Jairo Fúquene-Patiño\thanks{Correspondence should be addressed to jafuquenepatino@ucdavis.edu.}\\
Department of Statistics, University of California, Davis,
CA, United States}

\date{\today}

\begin{document}
    \maketitle
    
    \begin{abstract}

The Fay-Herriot (FH) model is widely used in official statistics to produce reliable estimates for domains with small sample sizes. In the classical FH model, the sampling variances are treated as known, even though they are typically estimated from the data. In practice, these variance estimates can be highly variable. To address this issue, practitioners often use Generalized Variance Functions (GVFs) to borrow strength across areas and stabilize estimation. In this work, we propose a new Bayesian model to improve the posterior estimation of the sampling variances in Small Area Estimation (SAE).   Our proposed model incorporates Global-Local (GL) priors to improve the level of shrinkage toward the 
posterior estimates obtained with the GVF function.  We study the theoretical properties of the proposed model and develop adaptive Markov chain Monte Carlo (MCMC) algorithms to address computational challenges arising from  conditional distributions involving Gamma functions. The performance of the proposed Bayesian model is investigated through simulation studies and compared with competing approaches. Finally, we implement our proposal in two real applications by estimating the Corn production from the U.S. Department of Agriculture and the Prevalence of the Educational Attainment Index at the municipality levels in Colombia.

    \end{abstract}


\section{Introduction}\label{sec:introduction}
Due to limited sample sizes at subnational levels (e.g., counties, municipalities, and provinces) or within subpopulations (e.g., gender and age groups), the classical Fay-Herriot model has been  widely used to obtain reliable small area estimates. Consider $m$ small areas, and let $y_i$ be a direct survey estimator for the small area mean parameter $\theta_i$ in the $i$-th area for $i=1,...,m$. The Fay-Herriot (FH) model \citep{1979_Fay&Herriot} assumes

\begin{equation}\label{FH_model}
 y_i = \theta_i + e_i ~\text{ and } ~\theta_i =  \boldsymbol{x}_i^\top \boldsymbol{\beta} + u_i ~\text{ for }~ i = 1, \ldots, m,
\end{equation}

where $e_1,..., e_m$ are independent sampling errors with $e_i \sim N(0, \sigma_i^2)$ and the sampling variances $\sigma_i^2$ are assumed to be known. The area-level random effects 
$u_1,..., u_m$ are independent with $u_i \sim N(0, \sigma_u^2)$, and the vectors $\boldsymbol{e} = (e_1,..., e_m)^\top$ and $\boldsymbol{u} = (u_1,..., u_m)^\top$ are independent. 
In the FH model, $\boldsymbol{x}_i=(x_{i1},...,x_{ip})^\top\in \mathbb{R}^p$ is a vector of covariates and $\boldsymbol{\beta}=(\beta_1,...,\beta_p)^\top \in\mathbb R^p$ is the corresponding regression coefficient vector. The covariates $x_{i1},...,x_{ip}$ are assumed to be known and measured without error and are typically obtained from administrative sources or census data. In practice, the sampling variances $\sigma_i^2$ in the FH model are usually estimated from the survey data. 
In this paper, we denote Direct Variance Estimates (DVEs) of the sampling variances as $D_i,~i=1,...,m$.

Under a Bayesian framework of the FH model the posterior mean of $\theta_i$ is a weighted average of the direct estimator $y_i$ and the synthetic regression estimate $\boldsymbol{x}_i^\top \boldsymbol{\beta}$ as follows:
\begin{equation}\label{eq:FH_mean}
\E(\theta_i \mid y_i, D_i, \boldsymbol{\beta}, \sigma_u^2 ) = (1 - \gamma_{i}) y_i + \gamma_{i}  \boldsymbol{x}_i^\top \boldsymbol{\beta}= y_i - \gamma_{i} (y_i -  \boldsymbol{x}_i^\top \boldsymbol{\beta}),
\end{equation}

where $\gamma_{i} = D_i/(D_i + \sigma_u^2)$. As pointed out by \cite{2009_Maples}, the weight $\gamma_{i}$ uses $D_i$ and therefore the precision in the estimation of the sampling variances affects the estimation of the parameter of interest $\theta_i$. However, as noted by many researchers \citep[e.g.,][]{2009_Maples,2021_Morales,2023_You&Hidiroglou}, the DVEs are highly variable in domains with small sample sizes. 
To address this issue, practitioners often use Generalized Variance Functions (GVFs) proposed by \cite{2007_Wolter} to borrow strength across areas and stabilize estimation.
The GVF method has been widely studied and applied in SAE, particularly for estimating the sampling variances when  the DVEs are unstable;
see for instance, \cite{2012_Kubacki&Jedrzej},  \cite{2018_McIllece},  \cite{2019_Zhang}, and  \cite{2023_You&Hidiroglou}. 
In this work, we consider the GVF model proposed by  \cite{2007_Wolter} given by a log-linear regression model as follows:
\begin{align}\label{GVF_model}
    \log(D_i)=\boldsymbol{z}_i^\top \boldsymbol{\eta }+\varepsilon_i,\quad i=1,...,m,
\end{align}

where $\varepsilon_1,...,\varepsilon_m$ are independent with $\varepsilon_i\sim N(0,\sigma_\varepsilon^2)$, $\boldsymbol{z}_i=(z_{i1},...,z_{iq})^\top\in\mathbb{R}^q$ is a vector of covariates, and $\boldsymbol{\eta}=(\eta_1,...,\eta_q)^\top\in\mathbb{R}^q$ is the corresponding coefficient vector. Estimates of the sampling variances  are obtained by taking the exponential of the fitted values obtained with the GVF model (\ref{GVF_model}). For brevity, we denote the 
exponential of the mean response, $\exp(\boldsymbol{z}_i^\top \boldsymbol{\eta })$, in (\ref{GVF_model}) by Exp-GVF. As noted by \cite{2009_Maples} and \cite{2021_Morales}, the main goal of using the GVF method in SAE is to smooth out the DVEs and reduce their uncertainty.  To avoid \emph{over}-smoothing, suitable covariates should be included in the GVF function. Importantly, in SAE typically the DVEs, $D_i$, are associated with the sample sizes, $n_i$.
For instance, \cite{2023_You&Hidiroglou} proposed the use of the logarithm of the sample sizes with $\boldsymbol{z}_i=(1,\log(n_i))^\top$ in the GVF model (\ref{GVF_model}) to improve the estimates of the sampling variances in SAE. This choice of $\boldsymbol{z}_i$ was also considered by \cite{2021_You} to model the sampling variances through a log-linear model where $\log(\sigma_i^2)\sim N(\eta_0+\eta_1\log(n_i),\sigma_\varepsilon^2)$. We note that the log-linear model for $\sigma_i^2$ in \cite{2021_You} is an extension of the model proposed by \cite{2009_Souza} where $\log(\sigma_i^2)$ is treated as random and $\log(n_i)$ is used in place of $1/n_i$. Besides the sample size $n_i$, \cite{2009_Maples} studied the use of other potential covariates in the GVF function (\ref{GVF_model}) closely related to $\theta_i$ and the specific survey design.

\cite{2006_You&Chapman} proposed a modified FH model which includes an additional model for $D_i$ into the FH model (\ref{FH_model}). Specifically, the DVEs follow a Gamma distribution given by
\begin{align}\label{gamma_model_Di}
    D_i\mid\sigma_i^2,\nu_i \sim{\rm Ga}\left(\frac{\nu_i}{2},~\frac{\nu_i}{2\sigma_i^2}\right),\quad i=1,...,m,
\end{align}

where ${\rm Ga}(a,b)$ denotes a Gamma distribution with shape parameter $a$ and rate parameter $b$ and where $\nu_i=n_i-1$.
The modified FH model defined by (\ref{FH_model}) and (\ref{gamma_model_Di}) has been widely used in SAE and different alternatives to estimate its model parameters have been proposed.

For example, under a frequentist framework, \cite{2002_Rivest&Vandal} and \cite{2003_Wang&Fuller} estimated the model parameters in the modified FH model using Empirical Best Linear Unbiased Prediction (EBLUP). Empirical and hierarchical Bayes estimators were obtained by \cite{2009_Maples}, \cite{2012_Dass}, and \cite{2014_Maiti}, and by \cite{2006_You&Chapman} and \cite{2017_Sukasawa}, respectively.

Under a Bayesian framework, several Inverse-Gamma priors for $\sigma_i^2$ have been proposed. For instance, \cite{2006_You&Chapman} suggested the use of $\sigma_i^2\sim{\rm IG}(a_i,b_i)$ where ${\rm IG}(a_i,b_i)$ denotes an Inverse-Gamma distribution with shape parameter $a_i$ and scale parameter $b_i$. In \cite{2006_You&Chapman}
 the values of $a_i$ and $b_i$ are chosen to be very small  to reflect vague knowledge on $\sigma_i^2$. Similarly, \cite{2021_You} recommended the use of a non-informative Inverse-Gamma prior for $\sigma_i^2$ where the hyperparameter values of $a_i$ and $b_i$ are set to 0.0001. 
\cite{2009_Maples} proposed incorporating the GVF method in the prior distribution of $\sigma_i^2$ to smooth out the variability of the DVEs in areas where the precision of those estimates is lacking. To this end, \cite{2009_Maples} proposed the use of an Inverse-Gamma prior for $\sigma_i^2$ where the precision is given by $\alpha$ and the mean is an Exp-GVF, i.e., $\sigma_i^2\sim{\rm IG}(\alpha+1,\alpha\exp(\boldsymbol{z}_i^\top\boldsymbol{\eta}))$. Importantly, by using the prior in \cite{2009_Maples} the model can shrink the DVE toward the Exp-GVF.


Other prior choices for $\sigma_i^2$ include the proposals of \cite{2014_Maiti}  and \cite{2017_Sukasawa}. Specifically, \cite{2017_Sukasawa} considered two 
alternatives of Inverse-Gamma priors for $\sigma_i^2$. In the first alternative, the shape and scale parameters are $a_i$ and $b_i\gamma$ where $a_i=2$ and $b_i=1/n_i$ respectively. The second proposal in \cite{2017_Sukasawa} considers that the Exp-GVF is included in the scale of the Inverse Gamma prior, i.e., $\sigma_i^2\sim{\rm IG}(a_i,b_i\gamma\exp(\boldsymbol{z}_i^\top\boldsymbol{\eta}))$. The proposed models in \cite{2017_Sukasawa} are an extension of those in \cite{2014_Maiti}. However, fully Bayesian inference is used  in \cite{2017_Sukasawa}. 
The results in \cite{2017_Sukasawa} suggested that their proposed priors for $\sigma_i^2$ can improve the posterior estimates of $\theta_i$ and $\sigma_i^2$ compared to those obtained with the prior proposed by \cite{2006_You&Chapman}. 


An important feature of the proposed models in \cite{2009_Maples} and \cite{2017_Sukasawa} is that the posterior mean of $\sigma_i^2$  is a shrinkage estimator obtained as a weighted average of $D_i$ and the prior mean of $\sigma_i^2$.  
In practice, when the DVE is unreliable, shrinkage toward the prior mean should be produced. However, while the modified FH models in \cite{2009_Maples} and \cite{2017_Sukasawa} yield shrinkage estimators for $\sigma_i^2$, they 
are not designed to capture heterogeneity across areas and to allow shrinkage toward a regression fit obtained via the GVF method in a data-driven manner. 
To improve the models' ability to adapt to heterogeneity and sparsity patterns, both Global and Local parameters should be incorporated in the prior of $\sigma_i^2$. 



In the Bayesian literature, Global-Local (GL) shrinkage priors have been extensively studied. Particularly, under Normal prior distributions, several local shrinkage priors for the variance (or scale) have been proposed. Those local priors include the heavy-tailed exponential prior distribution of \cite{2008_Park&Casella} which leads to the Laplace prior, and the heavy-tailed polynomial Beta Prime prior distribution which leads to the Horseshoe prior proposed by \cite{2010_Carvalho}. Although both exponential and polynomial priors provide attractive shrinkage properties, the Horseshoe prior is better suited for sparse problems. 
In the SAE context, the use of GL priors was originally proposed by \cite{2018_TangGhosh}. \cite{2018_TangGhosh} showed both theoretically and empirically that  GL priors for random effects in the FH model can improve the precision of small area estimates in SAE.

Recently, \cite{2024_Hamura} developed a new Bayesian model using GL shrinkage priors for positive-valued parameters 
under Gamma sampling models. \cite{2024_Hamura} proposed including both global ($\alpha$) and local ($\omega_i$) parameters in an Inverse-Gamma distribution to address sparsity in positive-valued data. Specifically, \cite{2024_Hamura} considered 
that the data is Gamma distributed with shape $\delta_i$ and rate $\delta_i/\varrho_i$
and proposed an Inverse-Gamma  prior for $\varrho_i$ with GL parameters, i.e., $\varrho_i\sim{\rm IG}(\alpha\omega_i+1,\alpha\omega_i\lambda)$, where 
the prior mean of $\rho_i$ is $\lambda$. Under Gamma models, the GL priors induce ``sparsity" in the sense that most of the parameters $\varrho_i$ are concentrated around the grand mean $\lambda$ with only a small fraction deviating from it. \cite{2024_Hamura} showed two fundamental  properties under GL priors: i.) tail-robustness for large 
signals and ii.) Kullback-Leibler super-efficiency under sparse settings. Due to the limited number of studies on GL priors 
under Gamma models, the theoretical findings obtained by \cite{2024_Hamura} are illuminating in the Bayesian literature.


In this paper, we propose a new Bayesian model by considering  the modified FH model and GL priors to improve the estimation of the sampling variances in SAE. Specifically, the proposed model uses GL parameters in the prior of $\sigma_i^2$  
to improve the level of shrinkage toward the Exp-GVF.
In our proposal, the considered covariates in the GVF model are $\log(n_i)$ and other auxiliary variables obtained from administrative sources or population censuses. For the local parameter, we study polynomial-tailed and exponential-tailed priors. The contributions of this work are threefold. First, theoretically,  we examine properties of the marginal prior distributions of $\sigma_i^2$ under both local polynomial-tailed and exponential-tailed priors to assess which prior is more suitable for our model.
Furthermore, we study the posterior mean of $\sigma_i^2$ and posterior concentration of the posterior shrinkage factor under the proposed model. 
Second, computationally, we propose adaptive Markov chain Monte Carlo (MCMC) algorithms to address mixing and convergence issues
arising from conditional distributions involving Gamma functions. Third, practically, we demonstrate the usefulness of the proposed model through simulated data and two real SAE applications.

The work is structured as follows. In Section \ref{sec:proposed_model}, we introduce the proposed GL shrinkage prior for the sampling variances in SAE.
We present theoretical properties of the marginal priors of $\sigma_i^2$ 
and a theoretical result for the posterior concentration of the shrinkage factor. MCMC algorithms for the proposed and existing Bayesian SAE models are given in Section \ref{sec:computation}. 
Also, we provide a theoretical guarantee that
the posterior distribution of the proposed model is proper under some mild conditions. 
In Section \ref{sec:simulation}, we consider simulation studies to compare the performance of the proposed model with the other existing models in the SAE literature. In Section \ref{sec:application}, we implement our proposed model in the famous  Corn data set studied previously by researchers, e.g., \cite{2006_You&Chapman} and \cite{2017_Sukasawa}. Also, we consider a real SAE  application that, to the best of our knowledge, has not previously been explored in the SAE literature. To this end, we apply our proposal to estimate the Prevalence of the Educational Attainment Index--At Least High School (PEAI--AHS) at the municipality levels in Colombia. Section \ref{sec:conclusions} contains the conclusions of our work. Proofs, additional algorithms, and further details of the applications are provided in the supplementary material.

\section{Proposed model}\label{sec:proposed_model}

\subsection{Global-Local shrinkage prior for the sampling variances}\label{subsec:GL_shrinkage_prior}

In this work, we propose the use of an Inverse-Gamma prior for $\sigma_i^2$ with GL parameters in the modified FH model defined by (\ref{FH_model}) and (\ref{gamma_model_Di}) as follows:

\begin{align}\label{proposed_GL_shrinkage_IG_prior}
    \sigma_i^2\mid \omega_i,\alpha,\boldsymbol{z}_i^\top\boldsymbol{\eta}&\sim{\rm IG}(\alpha \omega_i+1,\alpha \omega_i\exp(\boldsymbol{z}_i^{\top}\boldsymbol{\eta})),\quad i=1,...,m,
\end{align}

where $\alpha$ and $\omega_i$ are global and local parameters, respectively, and $\boldsymbol{z}_i$ is a vector of covariates used in the  Exp-GVF. 
The GL shrinkage prior (\ref{proposed_GL_shrinkage_IG_prior}) assumes that the sampling variances are centered around the Exp-GVF, i.e., $\E(\sigma_i^2)=\exp(\boldsymbol{z}_i^{\top}\boldsymbol{\eta})$. Since ${\rm Var}(\sigma_i^2)=(\exp(\boldsymbol{z}_i^{\top}\boldsymbol{\eta}))^2/(\alpha\omega_i-1)$ if $\alpha\omega_i>1$, both $\alpha$ and $\omega_i$ affect the variance of $\sigma_i^2$  and therefore the shrinkage toward $\exp(\boldsymbol{z}_i^{\top}\boldsymbol{\eta})$.
Larger values of both parameters imply stronger shrinkage around the Exp-GVF.
In particular, the global parameter $\alpha$ controls the overall shrinkage whereas the local parameter $\omega_i$ 
determines the area-specific shrinkage toward the Exp-GVF.

As mentioned, \cite{2024_Hamura} considered GL priors  to improve robustness in Gamma models. Although our work shares some similarities with the work of \cite{2024_Hamura}, the two approaches are substantially different. In \cite{2024_Hamura} the proposed GL prior aims to produce shrinkage around the grand mean mainly because their work was not developed within a SAE framework but instead under Gamma models. Specifically, for Gamma models in \cite{2024_Hamura} sparsity is induced around the grand mean whereas in this work sparsity is produced around the Exp-GVF under the SAE context. More importantly, our theoretical findings are related to the development of shrinkage estimators with GL priors for $\sigma_i^2$ in SAE while in \cite{2024_Hamura} the robustness of GL priors for Gamma models was studied in detail.
To the best of our knowledge, the use of the GL prior given by (\ref{proposed_GL_shrinkage_IG_prior}) for the sampling variances under the modified FH model has not previously been studied in the context of SAE. 
Although our theoretical findings and proposed methodology are general and can accommodate other types of covariates, we focus on the logarithm of the sample size.
Specifically, following \cite{2023_You&Hidiroglou} in this work we consider  $\boldsymbol{z}_i=(1,~\log(n_i))^\top$ in the Exp-GVF. Furthermore, in Section \ref{subsec:PEAI-AHS}, we demonstrate the use of an additional auxiliary variable together with $\log(n_i)$ in a real SAE problem.

\subsection{Marginal prior for $\sigma_i^2$ and shrinkage factors under the proposed SAE model}\label{subsec:GL_shrinkage_prior_2}

\cite{2024_Hamura} explored the behavior of the marginal prior for $\sigma_i^2$ according to the local parameter $\omega_i$ under Gamma models. 
To this end, \cite{2024_Hamura} stated three desirable properties for the marginal prior distribution of $\sigma_i^2$. We adapt those properties considered by \cite{2024_Hamura} to our proposed framework and present them in Properties \ref{propertiesH}. 
The first property (P1) is that the marginal prior $p(\sigma_i^2\mid\alpha,\boldsymbol{z}_i^\top\boldsymbol{\eta})$ of $\sigma_i^2$ \textit{induces non-shrinkage at the origin}.
The second and third properties, (P2) and (P3), are that $p(\sigma_i^2\mid\alpha,\boldsymbol{z}_i^\top\boldsymbol{\eta})$ \textit{has a heavy right tail} and \textit{a spike at the Exp-GVF}, respectively.

\begin{property}[] \label{propertiesH}
Suppose that the GL prior $\pi(\sigma_i^2\mid\omega_i,\alpha,\boldsymbol{z}_i^\top\boldsymbol{\eta})$ in (\ref{proposed_GL_shrinkage_IG_prior}) is assumed for $\sigma_i^2$ and that a proper prior $\pi(\omega_i)$ for the local parameter is given by $\pi(\omega_i)>0$ and $\int_0^\infty\pi(\omega_i)d\omega_i<\infty$. Let $p(\sigma_i^2\mid \alpha,\boldsymbol{z}_i^\top\boldsymbol{\eta})=\int_0^\infty\pi(\sigma_i^2\mid\omega_i,\alpha,\boldsymbol{z}_i^\top\boldsymbol{\eta})\pi(\omega_i)~d\omega_i$ denote the marginal prior of $\sigma_i^2$. In the context of shrinkage priors under the Gamma sampling model given by (\ref{gamma_model_Di}), we define the following three desirable properties for $p(\sigma_i^2\mid \alpha,\boldsymbol{z}_i^\top\boldsymbol{\eta})$:

\begin{enumerate}[
  labelwidth=1cm, labelsep=0.5em, leftmargin=!]
    \item[{\rm(P1)}] Non-shrinkage at the origin.
    
    As $\sigma_i^2\to0$, the marginal prior of $\sigma_i^2$ decays to zero, i.e.,
    \begin{align*}
        p(\sigma_i^2\mid\alpha,\boldsymbol{z}_i^\top\boldsymbol{\eta})\to0.
    \end{align*}
    
    \item[{\rm (P2)}] A heavy right tail. 
    
    As $\sigma_i^2\to\infty$, the marginal prior of $\sigma_i^2$ decays to zero at a polynomial rate, i.e.,
    \begin{align*}
        p(\sigma_i^2\mid\alpha,\boldsymbol{z}_i^\top\boldsymbol{\eta})\approx(\sigma_i^2)^{-2}\to0.
    \end{align*}
    
    \item[{\rm (P3)}] A spike at the Exp-GVF (i.e., the prior mean of $\sigma_i^2$). 
    
    As $\sigma_i^2\to\exp(\textbf{z}_i^\top\boldsymbol{\eta})$, the marginal prior of $\sigma_i^2$ puts an infinite probability mass, i.e.,
    \begin{align*}
        p(\sigma_i^2\mid\alpha,\boldsymbol{z}_i^\top\boldsymbol{\eta})\to\infty.
    \end{align*}
\end{enumerate}
\end{property}

\subsubsection{Prior choices for the local parameter}

\label{prior_choices}

We evaluate two popular prior distributions in the Bayesian literature for the local parameter $\omega_i$. The first one is the Beta Prime (BP) prior with a polynomial tail and the second one is the Gamma prior  with an exponential tail which includes the exponential prior as a special case. 
We defer discussion of the prior for the global parameter  $\alpha$ to Section \ref{subsec:computation_proposed_model}, focusing in this section on  $\omega_i$.

The BP prior has been studied by several researchers. Three important examples in the Bayesian literature include a particular case of the BP prior called the Horseshoe prior proposed by \cite{2010_Carvalho}, the BP robust prior in \cite{2017_Perez}, and the BP local prior for sparse Bayesian inference under Gamma models proposed by \cite{2024_Hamura}. 
The BP prior with hyperparameters $a,b>0$ is given by

\begin{align}\label{BP_prior}
\pi^{\rm BP}(\omega_i)\propto(\omega_i)^{a-1}(1+\omega_i)^{-(a+b)},\quad \omega_i>0.
\end{align}

Several researchers have also considered the Gamma prior. For instance,  \cite{2010_Griffin&Brown} proposed the use of the Gamma prior for the variance parameters of regression coefficients in Normal linear regression models. The Gamma prior with hyperparameters $c,d>0$ is given by

\begin{align}\label{Gamma_prior}
\pi^{\rm Ga}(\omega_i)\propto(\omega_i)^{c-1}\exp(-d\omega_i),\quad \omega_i>0.       
\end{align}

Proposition 2.1 of \cite{2024_Hamura} shows the behavior of the marginal prior distribution of $\sigma_i^2$ when the BP prior (\ref{BP_prior}) is considered for the local parameter. \cite{2024_Hamura} used the results of their Proposition 2.1 to determine the hyperparameters, $a$ and $b$, of the BP prior so that all three properties of the marginal prior similarly described in Properties \ref{propertiesH} are satisfied.
In Corollary \ref{cor}, we extend the results in Proposition 2.1 of \cite{2024_Hamura} to the context of our proposed SAE model. Parts (i) and (iii) in Corollary \ref{cor} show that the properties (P1) and (P3) in Properties \ref{propertiesH} are achieved when the hyperparameters of the BP local prior are $a>1$ and $b\leq1/2$ respectively. According to the formulation of Part (ii) in Corollary \ref{cor}, the marginal prior of $\sigma_i^2$ under the BP local prior exhibits a right-tail behavior proportional to $(\sigma_i^2)^{-2}$ up to a logarithmic factor $(\log\sigma_i^2)^{-(a+1)}$. 
Thus, the distribution  $p^{\rm BP}(\sigma_i^2\mid\alpha,\boldsymbol{z}_i^\top\boldsymbol{\eta})$  is heavy-tailed
and hence the property (P2) is always achieved under the BP prior regardless of the hyperparameter values. However, since a smaller value of $a$ corresponds to a heavier right tail, the hyperparameter $a$ should be chosen as small as possible. 
In summary, similar to \cite{2024_Hamura}, we also suggest using the hyperparameters $a=2$ and $b=1/2$ for the BP local prior in our proposed model as it satisfies all three desirable properties in Properties \ref{propertiesH}. In Section \ref{shrinkage_factors}, we also discuss the prior elicitation of the hyperparameter $a$ in the BP prior based on the concentration of the posterior distribution of the posterior shrinkage factor.

\begin{cor}\label{cor}
Suppose that the GL prior in (\ref{proposed_GL_shrinkage_IG_prior}) is assumed for $\sigma_i^2$ and that the BP prior  with hyperparameters $a$ and $b$ in (\ref{BP_prior}) is assigned to the local parameter $\omega_i$. Let $p^{\rm BP}(\sigma_i^2\mid\alpha,\boldsymbol{z}_i^\top\boldsymbol{\eta})$ denote the marginal prior of $\sigma_i^2$ under the BP prior. Then it satisfies the following properties:

\begin{enumerate}
    \item[{\rm (i)}] As $\sigma_i^2\to0$,
    \begin{align*}
        p^{\rm BP}(\sigma_i^2\mid\alpha,\boldsymbol{z}_i^\top\boldsymbol{\eta})\propto\frac{(\sigma_i^2)^{a-1}}{(\alpha\exp(\boldsymbol{z}_i^\top\boldsymbol{\eta}))^a}.
    \end{align*}
    
    \item[{\rm (ii)}] As $\sigma_i^2\to\infty$,
    \begin{align*}
        p^{\rm BP}(\sigma_i^2\mid\alpha,\boldsymbol{z}_i^\top\boldsymbol{\eta})\propto\alpha\exp(\boldsymbol{z}_i^\top\boldsymbol{\eta})(\sigma_i^2)^{-2}\left(\frac{1}{\log(\sigma_i^2/\exp(\boldsymbol{z}_i^\top\boldsymbol{\eta}))^\alpha}\right)^{a+1}\to0.
    \end{align*}
    
    \item[{\rm (iii)}] As $\sigma_i^2\to\exp(\boldsymbol{z}_i^\top\boldsymbol{\eta})$,
    $$
    p^{\rm BP}(\sigma_i^2\mid\alpha,\boldsymbol{z}_i^\top\boldsymbol{\eta})\to\begin{cases}
    C\sqrt{\alpha}/\exp(\boldsymbol{z}_i^\top\boldsymbol{\eta})<\infty & \text{if}~b>1/2,\\
        \infty & \text{if}~b\leq1/2
    \end{cases}
    $$

    where $C$ is some finite positive constant depending on the values of $a$ and $b$.
\end{enumerate}
\end{cor}

\begin{proof}
See Appendix \ref{proof_Corollary_2_1}.
\end{proof}

We next investigate whether the Gamma local prior also satisfies all the properties listed in Properties  \ref{propertiesH}, potentially making it another suitable prior for $\omega_i$.
Proposition \ref{proposition} demonstrates the behavior of the marginal prior of $\sigma_i^2$ when the Gamma prior (\ref{Gamma_prior}) is used for the local parameter. 
Similar to Corollary \ref{cor}, Part (i) in our Proposition \ref{proposition} shows that the property (P1) is satisfied when the shape of the Gamma prior is $c>1$. Thus, Part (i)  implies that the exponential prior, which is a special case of the Gamma prior when $c=1$, always fails to achieve the property (P1). Similar to Corollary \ref{cor}, Part (ii) in Proposition \ref{proposition} also indicates that the right tail of $p^{\rm Ga}(\sigma_i^2\mid\alpha,\boldsymbol{z}_i^\top\boldsymbol{\eta})$ is sufficiently heavy. Hence, the property (P2) is always satisfied under the Gamma local prior with a small value of the hyperparameter $c$. 
Nevertheless, according to Part (iii) in Proposition \ref{proposition}, the property  (P3) which  is arguably the most important property is impossible to achieve under the Gamma prior. 
Unlike the BP prior, the Gamma prior cannot produce a spike at the Exp-GVF for any value of the hyperparameter $d$. Since the Gamma local prior lacks the ability to shrink the sampling variance $\sigma_i^2$ toward the Exp-GVF, we do not adopt this choice in our work and therefore investigate the use of the BP local prior with the hyperparameters $(a,b)=(2,1/2)$ for the proposed GL shrinkage prior (\ref{proposed_GL_shrinkage_IG_prior}).

\begin{prop}\label{proposition}
Suppose that the GL prior in (\ref{proposed_GL_shrinkage_IG_prior}) is assumed for $\sigma_i^2$ and that the Gamma prior with hyperparameters $c$ and $d$ in (\ref{Gamma_prior}) is assigned to the local parameter $\omega_i$. Let $p^{\rm Ga}(\sigma_i^2\mid\alpha,\boldsymbol{z}_i^\top\boldsymbol{\eta})$ denote the marginal prior of $\sigma_i^2$ under the Gamma prior. Then it satisfies the following properties:
\begin{enumerate}
    \item[{\rm (i)}] As $\sigma_i^2\to0$,
    $$
    p^{\rm Ga}(\sigma_i^2\mid\alpha,\boldsymbol{z}_i^\top\boldsymbol{\eta})\propto\frac{(\sigma_i^2)^{c-1}}{(\alpha\exp(\boldsymbol{z}_i^\top\boldsymbol{\eta}))^c}. 
    $$
    
    \item[{\rm (ii)}] As $\sigma_i^2\to\infty$,
    \begin{align*}
        p^{\rm Ga}(\sigma_i^2\mid\alpha,\boldsymbol{z}_i^\top\boldsymbol{\eta})\propto\alpha\exp(\boldsymbol{z}_i^\top\boldsymbol{\eta})(\sigma_i^2)^{-2}\left(\frac{1}{\log(\sigma_i^2/\exp(\boldsymbol{z}_i^\top\boldsymbol{\eta}))^\alpha}\right)^{c+1}\to0.
    \end{align*}
    
    \item[{\rm (iii)}] As $\sigma_i^2\to\exp(\boldsymbol{z}_i^\top\boldsymbol{\eta})$,
    $$
    p^{\rm Ga}(\sigma_i^2\mid\alpha,\boldsymbol{z}_i^\top\boldsymbol{\eta})\to
    C'\sqrt{\alpha}/\exp(\boldsymbol{z}_i^\top\boldsymbol{\eta})<\infty, 
    $$

    where $C'$ is some finite positive constant depending on the values of $c$ and $d$.
\end{enumerate}
\end{prop}

\begin{proof}
See Appendix \ref{proof_Proposition_2_1}.
\end{proof}

Figure \ref{fig:marginal_prior_density_sigma2_i} displays the marginal prior distributions of $\sigma_i^2$ under both choices of local priors where we set $\alpha=1$ and $\exp(\boldsymbol{z}_i^\top\boldsymbol{\eta})=1$. The plots in Figure \ref{fig:marginal_prior_density_sigma2_i} for the BP and Gamma priors illustrate the analytical results in Corollary \ref{cor} and Proposition \ref{proposition}, respectively. According to Figure \ref{fig:marginal_prior_density_sigma2_i} (left), the marginal prior of $\sigma_i^2$ decays to zero as $\sigma_i^2$ approaches zero for BP and Gamma local priors 
when $a$ and $c$ are larger than one, respectively. Also, the marginal prior of $\sigma_i^2$ has a spike at $\sigma_i^2=1$ (i.e., the assumed Exp-GVF value) under the BP priors with hyperparameter $b=1/2$.
In contrast, when the Gamma prior is utilized as a local prior, the marginal prior distribution of $\sigma_i^2$ remains finite and does not exhibit a spike at one even with $d = 1/2$.
Figure \ref{fig:marginal_prior_density_sigma2_i} (right) shows that when 
$a$ or $c$ increases, the right tails under both local priors become lighter, consistent with the formulation in Part (ii). 
The Horseshoe prior of \cite{2010_Carvalho}, corresponding to the BP prior with the hyperparameters $(a,b)=(1/2,1/2)$, yields a sufficiently heavy right tail and a spike at the Exp-GVF.
However, the Horseshoe prior produces excessive shrinkage around the origin, which is contrary to the property (P1) in Properties \ref{propertiesH}.

\begin{figure}[ht!]
     \centering
     \begin{subfigure}[b]{0.45\textwidth}
         \centering
         \includegraphics[width=\textwidth]{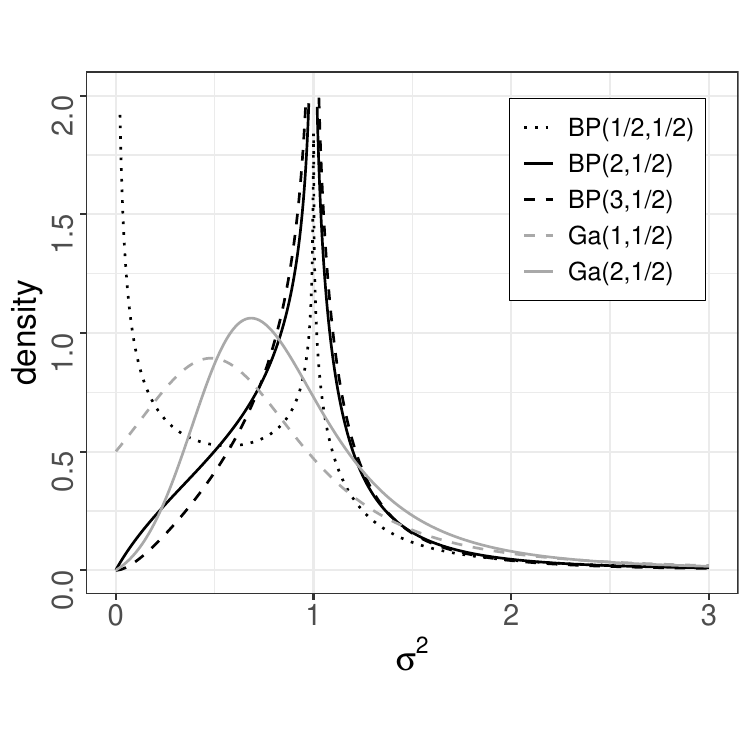}
     \end{subfigure}~
     \begin{subfigure}[b]{0.45\textwidth}
         \centering
         \includegraphics[width=\textwidth]{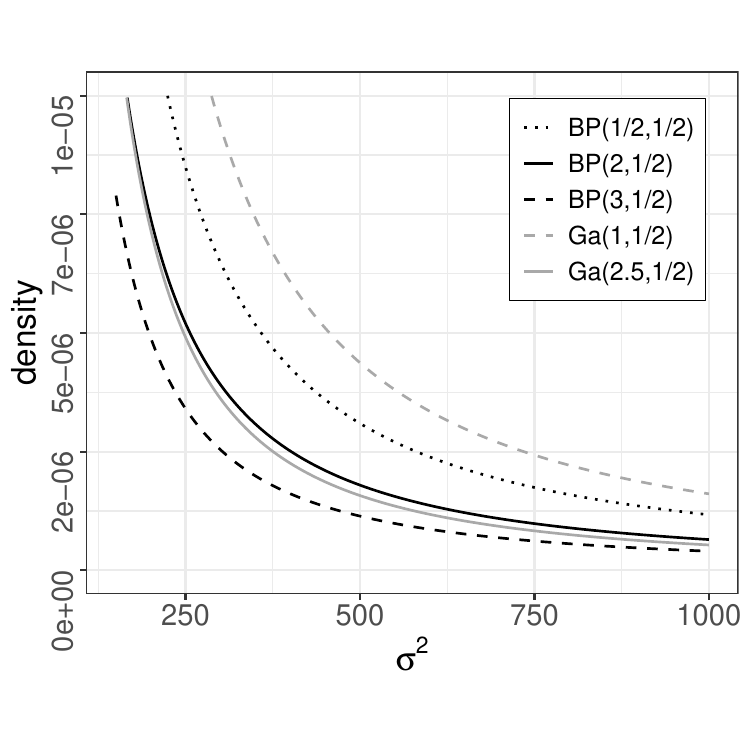}         
     \end{subfigure}
     \vspace{-0.6cm}
        \caption{Marginal prior distributions of $\sigma_i^2$ under the Beta Prime (BP) and Gamma local priors. The right figure focuses on the tail region. In Figure \ref{fig:marginal_prior_density_sigma2_i}, $\alpha=1$ and $\exp(\boldsymbol{z}_i^\top\boldsymbol{\eta})=1$.}
        \label{fig:marginal_prior_density_sigma2_i}
\end{figure}

We now illustrate the importance of including the local parameter $\omega_i$ in the proposed model. As described in Section \ref{sec:introduction}, our proposed model shares some similarities with the frequentist model of \cite{2009_Maples}. Specifically, both models incorporate the Exp-GVF to model the sampling variances $\sigma_i^2$. 
However, our model includes both global and local parameters 
in the prior of $\sigma_i^2$. In addition, in contrast to the model in \cite{2009_Maples},
our proposed model is formulated within a fully Bayesian framework. To display the differences  we compare the behavior of the marginal distribution of $D_i$ obtained under each model. 
The plots in Figure \ref{fig_marginal_D} show that including the parameter $\omega_i$  induces heavy-tailed robustness for large $D_i$ and strong shrinkage of $D_i$ toward the Exp-GVF.
Specifically, according to the plots in Figure \ref{fig_marginal_D} (top) the marginal distributions of $D_i$ under the proposed GL priors are more concentrated around the Exp-GVF 
than those under the model proposed by \cite{2009_Maples}
as the sample size $n_i$ increases.
The plots in Figure \ref{fig_marginal_D} (bottom)  also show that the differences in heavy-tailed behavior increase as the parameter $\alpha$ increases.


\begin{figure}[ht!]
     \centering
     \begin{subfigure}[b]{0.24\textwidth}
         \centering
         \includegraphics[width=\textwidth]{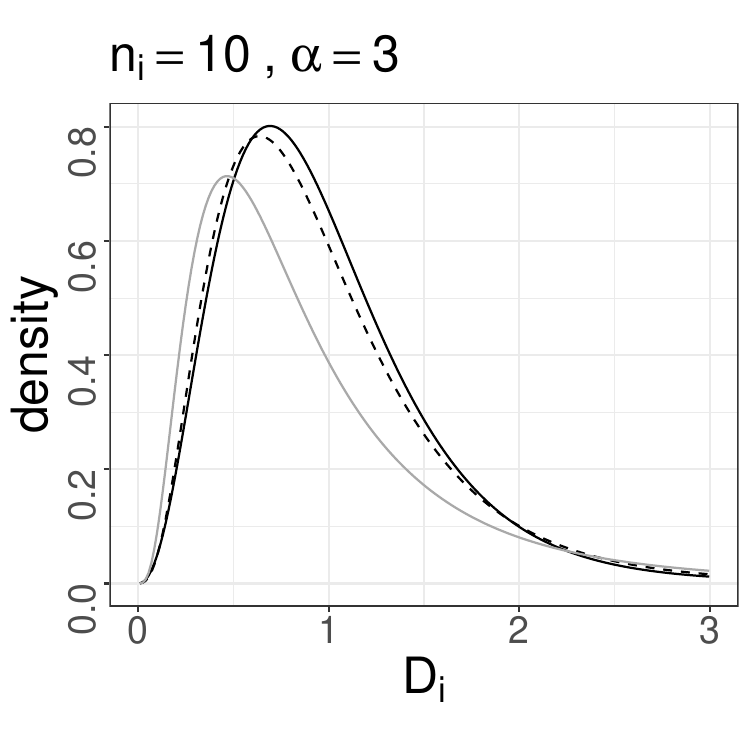}
     \end{subfigure}
     \begin{subfigure}[b]{0.24\textwidth}
         \centering
         \includegraphics[width=\textwidth]{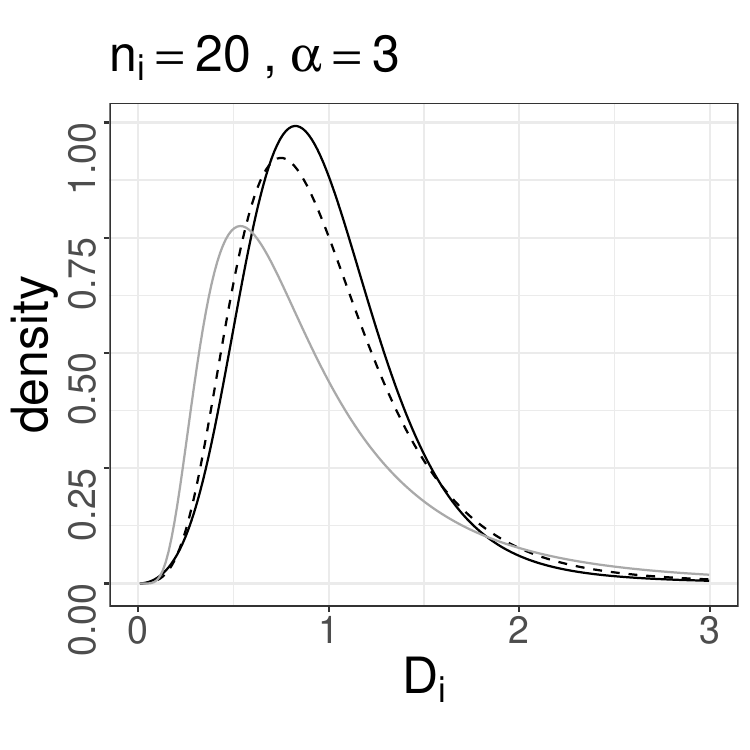}        
     \end{subfigure}
     \begin{subfigure}[b]{0.24\textwidth}
         \centering
         \includegraphics[width=\textwidth]{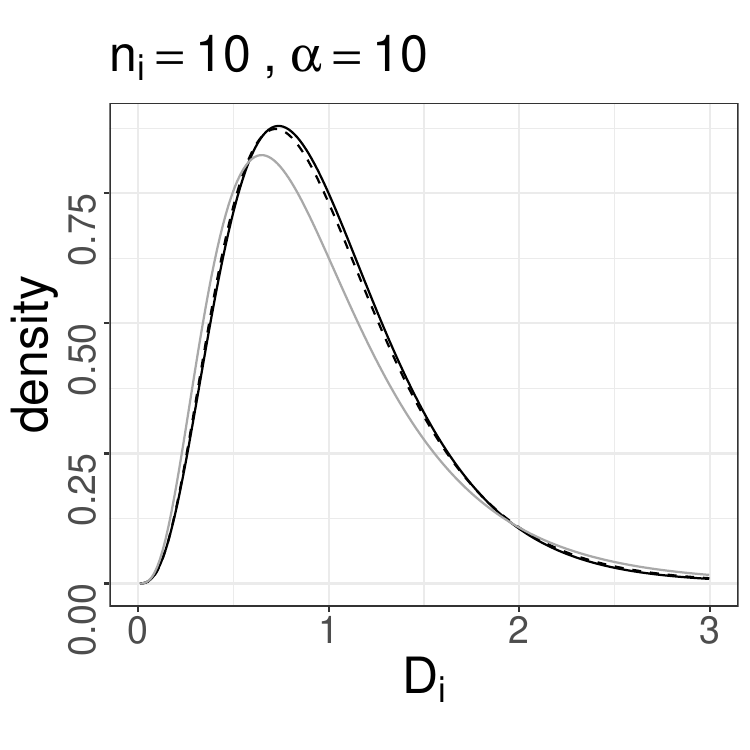}         
     \end{subfigure}
     \begin{subfigure}[b]{0.24\textwidth}
         \centering
         \includegraphics[width=\textwidth]{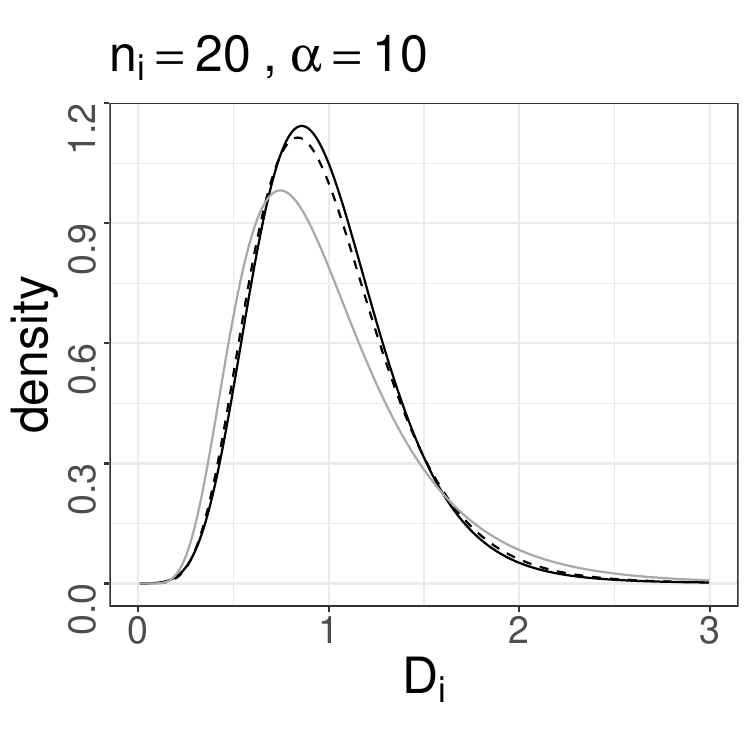}     
     \end{subfigure}
     \begin{subfigure}[b]{0.24\textwidth}
         \centering
         \includegraphics[width=\textwidth]{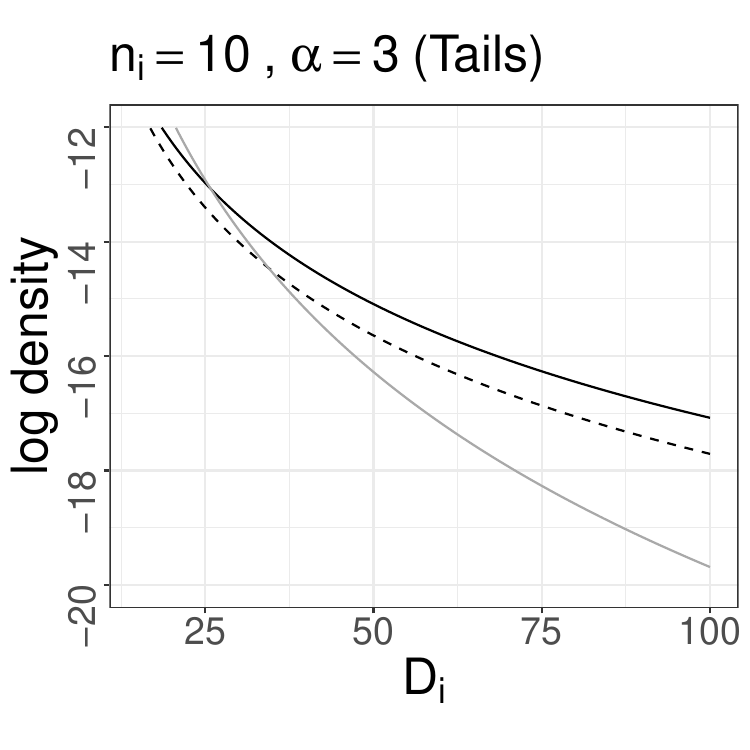}
     \end{subfigure}
     \begin{subfigure}[b]{0.24\textwidth}
         \centering
         \includegraphics[width=\textwidth]{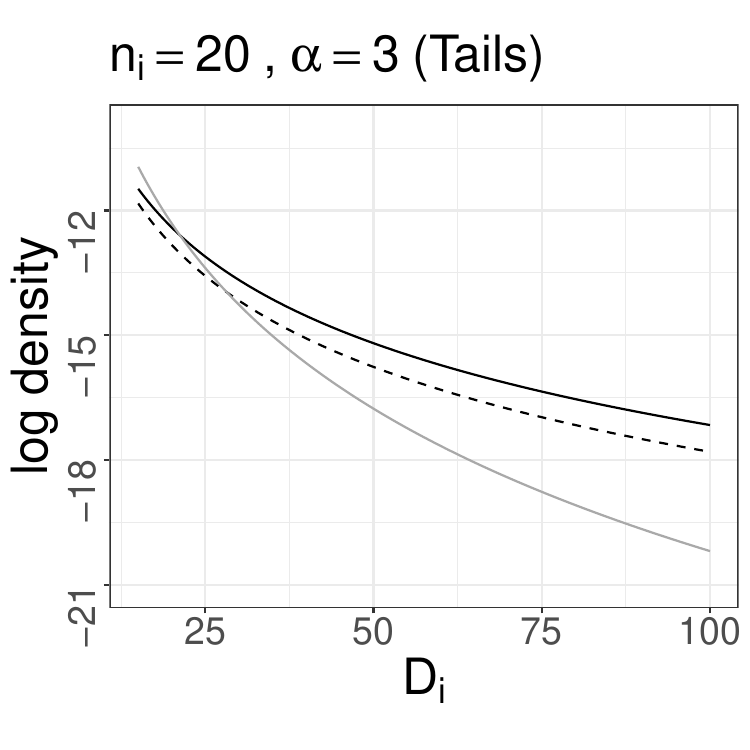}  
     \end{subfigure}
     \begin{subfigure}[b]{0.24\textwidth}
         \centering
         \includegraphics[width=\textwidth]{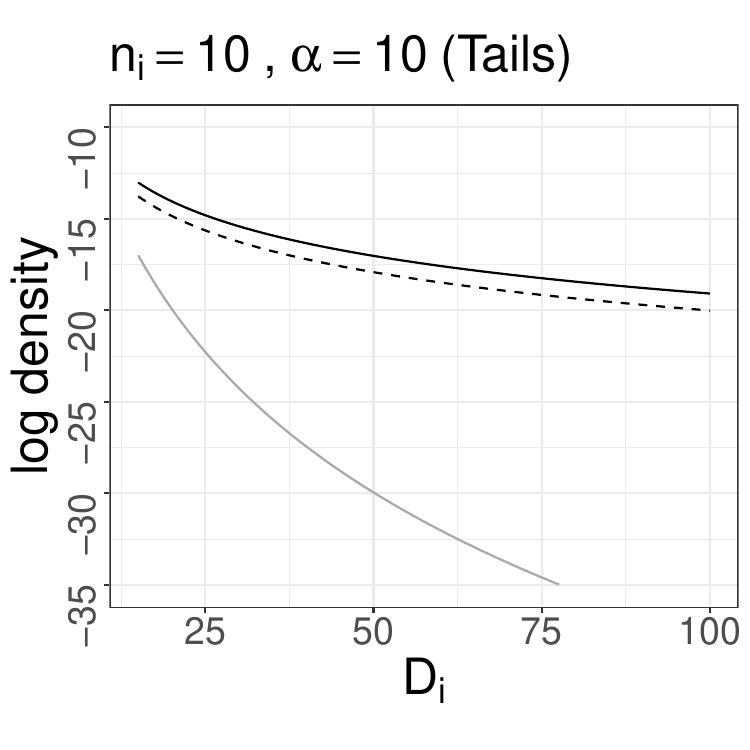}  
     \end{subfigure}
     \begin{subfigure}[b]{0.24\textwidth}
         \centering
         \includegraphics[width=\textwidth]{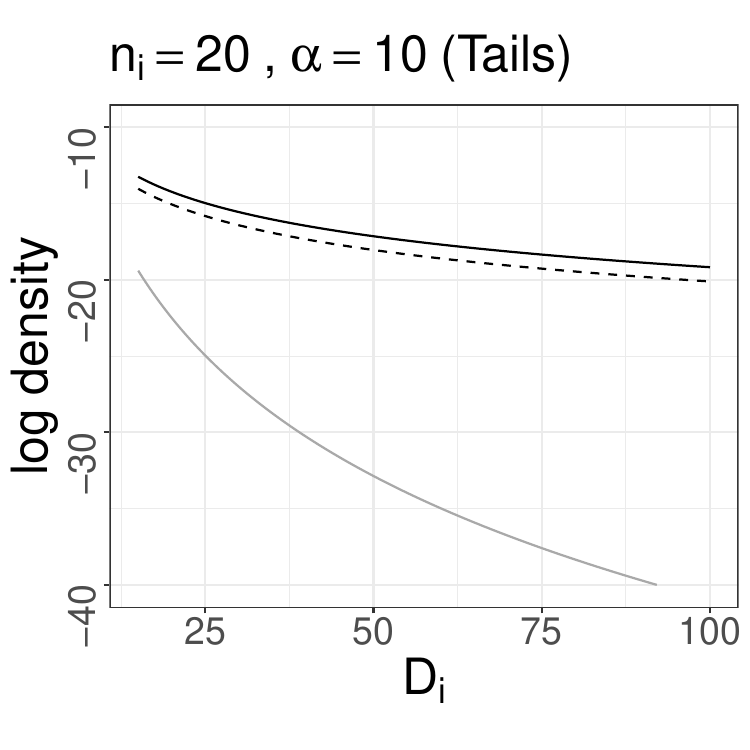}
     \end{subfigure}\\
     \vspace{-.5cm}
     \begin{subfigure}[b]{.5\textwidth}
         \centering
         \includegraphics[width=\textwidth]{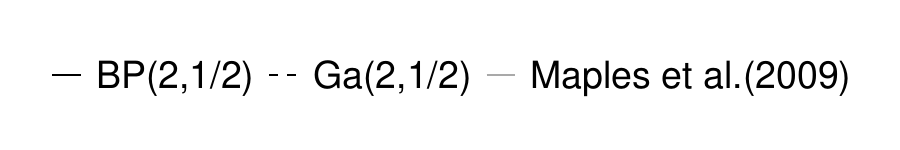}
     \end{subfigure}
     \vspace{-.5cm}
        \caption{Marginal distributions of $D_i$ obtained under the proposed models using the Beta Prime (BP) and Gamma local priors and the model proposed by \cite{2009_Maples}. The plots on the bottom focus on the tail region on a logarithmic scale. 
        In Figure \ref{fig_marginal_D}, $\exp(\boldsymbol{z}_i^\top\boldsymbol{\eta})=1$ and the local parameter in the proposed models is integrated out. }
        \label{fig_marginal_D}
\end{figure}

\subsubsection{Prior and posterior shrinkage factors}\label{subsec:shrinkage_factor}

\label{shrinkage_factors}

We now investigate the posterior concentration of the posterior shrinkage factor under our proposed Bayesian model.  Using the Gamma model (\ref{gamma_model_Di}) and the GL prior (\ref{proposed_GL_shrinkage_IG_prior}) the conditional distribution of $\sigma_i^2$ 
is an Inverse-Gamma, ${\rm IG}((n_i-1)/2+\alpha\omega_i+1,(n_i-1)D_i/2+\alpha\omega_i\exp(\boldsymbol{z}_i^\top\boldsymbol{\eta}))$.  
Therefore the conditional expectation of $\sigma_i^2$ can be written as follows:
\begin{align}\label{prior_shrinkage}
    \E(\sigma_i^2\mid \Delta_i^{\rm prior})&=\E_{\omega_i}[\E(\sigma_i^2\mid\omega_i,\Delta_i^{\rm prior})\mid \Delta_i^{\rm prior}]\nonumber\\
    &=\E_{\omega_i}\left[\frac{(n_i-1)D_i/2+\alpha\omega_i\exp(\boldsymbol{z}_i^\top\boldsymbol{\eta})}{(n_i-1)/2+\alpha\omega_i}\mid \Delta_i^{\rm prior}\right]\nonumber \\
    &=[1-\E_{\omega_i}(\kappa^{\rm prior}_i\mid \Delta_i^{\rm prior})]~D_i+\E_{\omega_i}(\kappa^{\rm prior}_i\mid \Delta_i^{\rm prior})~\exp(\boldsymbol{z}_i^\top\boldsymbol{\eta}),
\end{align}

where $\Delta_i^{\rm prior}=(D_i,n_i,\alpha,\boldsymbol{z}_i^\top\boldsymbol{\eta})$ is a compact notation of the model parameters and $\kappa^{\rm prior}_i=\alpha \omega_i/((n_i-1)/2+\alpha \omega_i)\in(0,1)$ denotes a prior shrinkage factor. The prior mean of the prior shrinkage factor, $\E_{\omega_i}(\kappa^{\rm prior}_i\mid \Delta_i^{\rm prior})$, produces the amount of shrinkage of $D_i$ toward the Exp-GVF.  When the amount of shrinkage approaches 0 or 1, the conditional expectation of $\sigma_i^2$ approaches $D_i$ or $\exp(\boldsymbol{z}_i^\top\boldsymbol{\eta})$, respectively.
Proposition \ref{proposition_post} shows that, under the modified FH model,
the posterior mean of $\sigma_i^2$ exhibits a weighted-average structure similar to that in (\ref{prior_shrinkage}).

\begin{prop}\label{proposition_post}
    Denote $\Delta_i^{\rm post}=\Delta_i^{\rm prior}\times(y_i,\boldsymbol{x}_i^\top\boldsymbol{\beta},\sigma_u^2)=(y_i,D_i,n_i,\alpha,\boldsymbol{z}_i^\top\boldsymbol{\eta},\boldsymbol{x}_i^\top\boldsymbol{\beta},\sigma_u^2)$. Using the modified FH model, defined by (\ref{FH_model}) and (\ref{gamma_model_Di}), and the GL prior in (\ref{proposed_GL_shrinkage_IG_prior}), the full conditional distribution of $\sigma_i^2$ 
    is ${\rm IG}(n_i/2+\alpha\omega_i+1$, $(y_i-\theta_i)^2/2+(n_i-1)D_i/2+\alpha\omega_i\exp(\boldsymbol{z}_i^\top\boldsymbol{\eta}))$. Then the posterior mean of $\sigma_i^2$ is given by

\begin{align}\label{post_shrinkage}
    \E&(\sigma_i^2\mid \Delta_i^{\rm post})\nonumber\\
    &=\E_{\theta_i}(\E_{\omega_i\mid\theta_i}(\E(\sigma_i^2\mid\theta_i,\omega_i,y_i,\Delta_i^{\rm prior}))\mid\Delta_i^{\rm post})\nonumber\\
    &=\E_{\theta_i}\left(\E_{\omega_i\mid\theta_i}\left(\frac{n_i/2}{n_i/2+\alpha\omega_i}\left[\frac{(y_i-\theta_i)^2}{n_i}+\frac{(n_i-1)D_i}{n_i}\right]+\frac{\alpha\omega_i}{n_i/2+\alpha\omega_i}\exp(\boldsymbol{z}_i^\top\boldsymbol{\eta})\right)\mid\Delta_i^{\rm post}\right)\nonumber\\
    &=[1-\E_{\theta_i}(\E_{\omega_i\mid\theta_i}(\kappa_i^{\rm post})\mid\Delta_i^{\rm post})]~\tilde{\sigma}_i^2(y_i,D_i)+\E_{\theta_i}(\E_{\omega_i\mid\theta_i}(\kappa^{\rm post}_i)\mid\Delta_i^{\rm post})~\exp(\boldsymbol{z}_i^\top\boldsymbol{\eta}),
\end{align}

where $\kappa^{\rm post}_i=\alpha\omega_i/(n_i/2+\alpha\omega_i)\in(0,1)$ denotes a posterior shrinkage factor and

\begin{align}\label{observed_var}
    \tilde{\sigma}_i^2(y_i,D_i)=\frac{1}{n_i}~\E_{\theta_i}\left[(y_i-\theta_i)^2\mid \Delta_i^{\rm post}\right]
    +\frac{(n_i-1)D_i}{n_i}.
\end{align}

The term $\tilde{\sigma}_i^2(y_i,D_i)$ is referred to as an observed variance.

\end{prop}

\begin{proof}
See Appendix \ref{proof_Proposition_2_2_theorem_2_1}.
\end{proof}

We note that the observed variance $\tilde{\sigma}_i^2(y_i,D_i)$ in Proposition \ref{proposition_post} is close to $D_i$ when the sample size $n_i$ is large. The exact representation of the expectation term in (\ref{observed_var}) can be found in Appendix \ref{proof_Proposition_2_2_theorem_2_1} of the Supplementary Material. Similar to $\E_{\omega_i}(\kappa^{\rm prior}_i\mid \Delta_i^{\rm prior})$ in (\ref{prior_shrinkage}), the posterior mean of the posterior shrinkage factor, $\E_{\theta_i}(\E_{\omega_i\mid\theta_i}(\kappa^{\rm post}_i)\mid\Delta_i^{\rm post})$, can be interpreted as the posterior amount of shrinkage of $\tilde{\sigma}_i^2(y_i,D_i)$ toward the Exp-GVF. 
Furthermore, the posterior amount of shrinkage depends on $n_i$ and how well the posterior estimate of the Exp-GVF explains the sampling variance $\sigma_i^2$.
Specifically, we expect that (i) $\E_{\theta_i}(\E_{\omega_i\mid\theta_i}(\kappa^{\rm post}_i)\mid\Delta_i^{\rm post})$ be close to 1 when $n_i$ is small and the posterior estimate of the Exp-GVF explains $\sigma_i^2$ well, 
and (ii) $\E_{\theta_i}(\E_{\omega_i\mid\theta_i}(\kappa^{\rm post}_i)\mid\Delta_i^{\rm post})$ be close to 0 when $n_i$ is large and the Exp-GVF does not provide a good fit for $\sigma_i^2$.

Theorem \ref{thm:prob_kappa} shows the concentration of the posterior distribution of the
posterior shrinkage factor $\kappa^{\rm post}_i$.
We note that both the global parameter $\alpha$ and the hyperparameter $a$ of the BP local prior control
 the posterior concentration of  $\kappa^{\rm post}_i$.
Thus, when the global parameter $\alpha$ is large the posterior distribution of the posterior shrinkage factor concentrates near one.
Therefore, when the covariates $\boldsymbol{z}_i$ in the Exp-GVF can explain the sampling variance $\sigma_i^2$,  we expect to have a large $\alpha$  so that the posterior mean of $\sigma_i^2$ is close to the posterior estimate of the Exp-GVF. Importantly, since the shrinkage factor $\kappa^{\rm post}_i$ also depends on the sample size, $n_i$, this shrinkage effect can be mitigated as  $n_i$ increases.   Furthermore, according to Theorem \ref{thm:prob_kappa}, the value of $a$ also controls the concentration of the posterior distribution of $\kappa^{\rm post}_i$. 
To avoid over-shrinkage for small areas for which the covariates in the Exp-GVF cannot explain $\sigma_i^2$  we should set a small value for $a$. 
Therefore, the choice $a=2$ for the BP local prior discussed previously in Section  \ref{prior_choices} is also supported by Theorem \ref{thm:prob_kappa}.



\begin{theorem}\label{thm:prob_kappa} (Posterior concentration of the posterior shrinkage factor)
Using the modified FH model, defined by (\ref{FH_model}) and (\ref{gamma_model_Di}), and the GL prior in (\ref{proposed_GL_shrinkage_IG_prior}), we suppose that the BP prior with hyperparameters $a$ and $b$ in (\ref{BP_prior}) is assigned to the local parameter $\omega_i$. Let $\kappa^{\rm post}_i=\alpha \omega_i/(n_i/2+\alpha \omega_i)$ denote the posterior shrinkage factor. For any $\varepsilon_i\in(0,1)$, as $1<\alpha\to\infty$, we have

\begin{align*}
\P(\kappa^{\rm post}_i  <\varepsilon_i\mid \Delta_i^{\rm post})\leq C''\alpha^{-a}\to0,
\end{align*}

where $\Delta_i^{\rm post}=(y_i,D_i,n_i,\alpha,\boldsymbol{z}_i^\top\boldsymbol{\eta},\boldsymbol{x}_i^\top\boldsymbol{\beta},\sigma_u^2)$ and $C''$ is some finite positive constant.

\end{theorem}

\begin{proof}
See Appendix \ref{proof_Proposition_2_2_theorem_2_1}.
\end{proof}

\section{Computation}\label{sec:computation}

To evaluate the performance of our proposed model, we consider three Bayesian SAE models in the literature. 
In Section \ref{subsec:existing_models}, we describe the MCMC algorithms utilized to obtain posterior estimators for those existing models. In Section  \ref{subsec:computation_proposed_model}, we present the proposed MCMC algorithms developed for posterior inference in the proposed model. In particular, we detail the adaptive sampling scheme used to efficiently draw samples from the posterior distribution. Furthermore, we establish a theoretical guarantee in Theorem \ref{thm:proper_posterior} showing that the posterior distribution of the proposed model is proper under a set of mild regularity conditions. These results ensure that the posterior is well-defined and suitable for Bayesian inference. The details of the existing and proposed Bayesian SAE models are given in Table \ref{tab:model_structures}.

\begin{table}[ht!]
    \centering
    \resizebox{!}{5cm}{
    \begin{tabular}{|c|l|c|c|c|c|}
    \hline
    &&\multicolumn{3}{c|}{Model for}  &\\
    \cline{3-5}
     &\multicolumn{1}{c|}{} &\multicolumn{1}{c|}{$y_i$ and $\theta_i$} & $D_i$ & $\sigma_i^2$ & MCMC\\
     \hline
     &&&& & \\
     \multirow{9}{*}{\rotatebox{90}{Existing Bayesian models}} & YC & Fay-Herriot &
        Gamma  &$\sigma_i^2\sim\text{IG}(a_i,b_i)$ &  \\
     &&Eq (\ref{FH_model}) & model&$a_i=b_i=0.0001$ & Gibbs\\
     &&$\pi(\boldsymbol{\beta},\sigma_u^2)\propto1$ &Eq (\ref{gamma_model_Di})&& sampling\\
     \cline{2-2}\cline{5-5}
     &&&&& \\
     &STK1 && &$\sigma_i^2\sim\text{IG}(a_i,b_i\gamma)$ & \\
     && &  &$a_i=2,~b_i=1/n_i$ & \\
     &&&& $\pi(\gamma)\propto1$ & \\
     &&&&& \\
          \cline{2-2}\cline{5-6}
          &&&&& \\
          &STK2 && &$\sigma_i^2\sim\text{IG}(a_i,b_i\gamma\exp(\boldsymbol{z}_i^\top\boldsymbol{\eta}))$ & Gibbs and adaptive \\
     &&  &  &$a_i=2,~b_i=1/n_i$ & Metropolis- \\
     &&&& $\pi(\gamma,\boldsymbol{\eta})\propto1$ & Hastings \\
     &&&&& \\
          \cline{1-2}\cline{5-6}
          &&&&& \\
     \multirow{4}{*}{\rotatebox{90}{Proposed model}}&GL-BP&&&$\sigma_i^2\sim\text{IG}(\alpha\omega_i+1,\alpha\omega_i\exp(\boldsymbol{z}_i^\top\boldsymbol{\eta}))$& Gibbs and adaptive \\
     &&&&$\omega_i\sim{\rm BP}(2,1/2)$& Metropolis-    \\
     &&&& $\alpha\sim{\rm Ga}(a_\alpha,b_\alpha)$ & Hastings  \\
     &&&& $\pi(\boldsymbol{\eta})\propto1$ &    \\
     &&&&&    \\
     &&&& &  \\
     \hline
    \end{tabular}
    }
    \caption{Settings and Markov chain Monte Carlo (MCMC) alternatives for the existing and proposed Bayesian SAE models.}
    \label{tab:model_structures}
\end{table}



\subsection{MCMC algorithms for the  existing Bayesian SAE models}\label{subsec:existing_models}

The three existing Bayesian models presented in Table  \ref{tab:model_structures} are based on the Gamma sampling model (\ref{gamma_model_Di}) for $D_i$, with different specifications of Inverse-Gamma priors assigned to $\sigma_i^2$.
The \cite{2006_You&Chapman} (YC) model considers a non-informative prior for $\sigma_i^2$ by setting small values of the shape ($a_i$) and scale ($b_i$) parameters such as $a_i=b_i=0.0001$. 
Uniform priors are assumed on $\boldsymbol{\beta}$ and $\sigma_u^2$, i.e., $\pi(\boldsymbol{\beta},\sigma_u^2)\propto1$.
The full conditional distributions of the unknown model parameters $(\theta_i,\sigma_i^2,\boldsymbol{\beta},\sigma_u^2)$ are obtained in closed form and the Gibbs sampling algorithm for YC is given in Algorithm \ref{alg:YC}, Appendix \ref{MCMC_schemes}. 
As noted, the posterior mean of $\sigma_i^2$  cannot be written as a weighted average and therefore does not admit a posterior shrinkage interpretation.

In \cite{2017_Sukasawa} Bayesian SAE models shrinking both means and variances were developed. Specifically, in order to obtain the posterior mean of $\sigma_i^2$ as a shrinkage estimator, \cite{2017_Sukasawa} proposed to use an Inverse-Gamma prior for $\sigma_i^2$ with shape $a_i$ and scale $b_i\gamma$, i.e., $\sigma_i^2\sim{\rm IG}(a_i,b_i\gamma)$, where a Uniform prior is assumed on $\gamma$, i.e., $\pi(\gamma)\propto1$. 
We refer to this model proposed by \cite{2017_Sukasawa} as STK1. Since under the STK1 model the full conditional distributions   for all model parameters $(\theta_i,\sigma_i^2,\boldsymbol{\beta},\sigma_u^2,\gamma)$ are available in closed form, posterior sampling using a Gibbs sampler is straightforward.
Specifically, the full conditional distributions of $(\theta_i,\sigma_i^2,\boldsymbol{\beta},\sigma_u^2)$ under STK1 are similar to those under YC, and the full conditional of $\gamma$ 
is a Gamma distribution. Details of the Gibbs sampling procedure can be found in Algorithm \ref{alg:STK1}, Appendix \ref{MCMC_schemes}.

To illustrate the difference between the STK1 model  formulation and our proposed model we consider the posterior mean of $\sigma_i^2$.
The posterior mean of $\sigma_i^2$ for STK1 with the suggested settings of $a_i=2$ and $b_i=1/n_i$ in \cite{2017_Sukasawa} can be written as a weighted average as follows:
\begin{align}\label{cond_mean_STK}
\E^{\rm STK1}(\sigma_i^2\mid y_i,D_i,n_i,\theta_i, \boldsymbol{x}_i^\top\boldsymbol{\beta},\sigma_u^2)=(1-\tilde\kappa_i^{\rm post})\tilde{\sigma}^2_i(y_i,\theta_i,D_i)+\tilde\kappa_i^{\rm post}~\E(\sigma_i^2),
\end{align}

where the observed variance is $\tilde{\sigma}^2_i(y_i,\theta_i,D_i)=[(y_i-\theta_i)^2+(n_i-1)D_i]/n_i$ and the prior mean is $\E(\sigma_i^2)=\gamma/n_i$. 
We note that  in (\ref{cond_mean_STK})
the posterior mean of $\sigma_i^2$ is conditional on $\theta_i$ therefore the posterior means in our proposed model (\ref{post_shrinkage}) and (\ref{cond_mean_STK}) are not comparable. 
In addition, the posterior shrinkage factor of STK1 in (\ref{cond_mean_STK}) is given by $\tilde\kappa_i^{\rm post}=1/(n_i/2+1)\in(0,1)$. 
In particular, the posterior mean in (\ref{cond_mean_STK}) approaches $\tilde{\sigma}^2_i(y_i,\theta_i,D_i)$ when $n_i$ is large and shrinks toward the prior mean of $\sigma_i^2$ when $n_i$ is small.
Therefore, under the STK1 model, the degree of shrinkage depends only on the sample size $n_i$.
In contrast, the shrinkage factor $\kappa_i^{\rm post}$ of our proposed model in (\ref{post_shrinkage}) depends not only on $n_i$ but also on the global and local parameters. 
This feature allows the proposed model to adaptively shrink the posterior mean of $\sigma^2_i$ toward (\ref{observed_var}) and the Exp-GVF, according to the sample size $n_i$
and how well the posterior estimate of the Exp-GVF fits the sampling variance.
 
In \cite{2017_Sukasawa} the STK1 model was extended by including the Exp-GVF in the scale of the Inverse-Gamma prior for $\sigma_i^2$, i.e., $\sigma_i^2\sim{\rm IG}(a_i,b_i\gamma\exp(\boldsymbol{z}_i^\top\boldsymbol{\eta}))$. This model is referred to as STK2. Importantly, the posterior mean of $\sigma_i^2$ for STK2 has the same expression as that for STK1 in (\ref{cond_mean_STK}) except that the prior mean is $\E(\sigma_i^2)=\gamma\exp(\boldsymbol{z}_i^\top\boldsymbol{\eta})/n_i$ for STK2. 
Therefore, STK2 suffers from the same issues as STK1, namely a lack of adaptivity to heterogeneity and insufficient shrinkage toward the Exp-GVF.
As noted in \cite{2017_Sukasawa}, the covariate $z_{i1}=1$ is not considered in the STK2 model  since $\gamma$ and the regression parameter $\eta_0$ corresponding to $z_{i1}=1$ are not identifiable. This issue occurs because the likelihood function depends on both parameters $\gamma$ and $\eta_0$ only through the term $\gamma \exp(\eta_0)$. In other words, for any constant $c>0$, replacing $(\gamma, \eta_0)$ with $(\gamma',\eta_0') = (c \gamma, \eta_0 - \log c)$ leaves the product $\gamma \exp(\eta_0) = \gamma' \exp(\eta_0')$ unchanged. These two pairs produce the same likelihood value, meaning that the parameters $\gamma$ and $\eta_0$ cannot be separately identified from the observed data. Therefore, when applying the STK2 model in Sections \ref{sec:simulation} and \ref{sec:application}, we omit the intercept $\eta_0$ in the Exp-GVF.

Since the full conditional distribution  of $\boldsymbol{\eta}$ is not available in closed form we  use a Metropolis–Hastings (MH) algorithm to obtain posterior samples
of $\boldsymbol{\eta}$.
The random-walk MH in \cite{2017_Sukasawa} considered a Normal dependent proposal with a fixed variance.
However, we found some convergence issues when the proposal variance in the MH algorithm was held fixed.
To address this computational issue, we recommend adjusting the covariance matrix of the proposal distribution for $\boldsymbol{\eta}$. 
The MCMC algorithm for STK2 with the Gibbs steps for $(\theta_i,\sigma_i^2,\boldsymbol{\beta},\sigma_u^2,\gamma)$ and the adaptive MH step for $\boldsymbol{\eta}$ is provided in Algorithm \ref{alg:STK2}, Appendix \ref{MCMC_schemes}.





\subsection{MCMC algorithm for the proposed model}\label{subsec:computation_proposed_model}
We refer to our proposed Bayesian SAE model under the BP local prior as GL-BP. 
For notational simplicity, we denote $\boldsymbol{\mathcal{D}}=\{y_i,D_i,n_i,\boldsymbol{z}_i,\boldsymbol{x}_i\}_{i=1}^m$ as the set of all observed data, and define the vectors $\boldsymbol{\theta}=(\theta_1,...,\theta_m)^\top$, $\boldsymbol{\sigma}^2=(\sigma_1^2,...,\sigma_m^2)^\top$, and $\boldsymbol{\omega}=(\omega_1,...,\omega_m)^\top$. For the proposed GL-BP model, the joint posterior distribution of $(\boldsymbol{\theta},\boldsymbol{\sigma}^2,\boldsymbol{\beta},\sigma_u^2,\alpha,\boldsymbol{\omega},\boldsymbol{\eta})$ is given by

\begin{align}\label{HB_posterior}
\pi(\boldsymbol{\theta},\boldsymbol{\sigma}^2,\boldsymbol{\beta}&,\sigma_u^2,\alpha,\boldsymbol{\omega},\boldsymbol{\eta}\mid \boldsymbol{\mathcal{D}}) \propto~(\sigma_u^2)^{-m/2}\times\pi(\alpha)\nonumber\\
&\times\prod_{i=1}^m\left[\left(\frac{1}{\sigma_i^2}\right)^{\frac{n_i}{2}+(\alpha\omega_i+1)+1}\times\frac{(\alpha\omega_i\exp(\boldsymbol{z}_i^\top\boldsymbol{\eta}))^{\alpha\omega_i+1}}{\Gamma(\alpha\omega_i+1)}\times\pi^{\rm BP}(\omega_i)\right]\nonumber\\
&\times\prod_{i=1}^m\exp\left\{-\frac{(y_i-\theta_i)^2/2+(n_i-1)D_i/2+\alpha\omega_i\exp(\boldsymbol{z}_i^\top\boldsymbol{\eta})}{\sigma_i^2}-\frac{(\theta_i-\boldsymbol{x}_i^\top\boldsymbol{\beta})^2}{2\sigma_u^2}\right\}.
\end{align}

In Theorem \ref{thm:proper_posterior}, we have the conditions to ensure that the joint posterior distribution (\ref{HB_posterior}) is proper. Similar to the existing models described in Section \ref{subsec:existing_models} we assume $\pi(\beta,\sigma_u^2,\boldsymbol{\eta})\propto1$ in our proposed GL-BP model.
According to Theorem \ref{thm:proper_posterior} the posterior distribution in (\ref{HB_posterior}) is proper if the prior for the global parameter $\alpha$,  $\pi(\alpha)$, is proper.
Therefore similar to \cite{2024_Hamura} we consider a  Gamma prior distribution for  $\alpha$ with $\alpha\sim{\rm Ga}(a_{\alpha},b_{\alpha})$. 


\begin{theorem}\label{thm:proper_posterior}
    The joint posterior distribution under $\omega_i\sim{\rm BP}(a,b)$ in (\ref{HB_posterior}) is proper if the prior $\pi(\alpha)$ for the global parameter $\alpha$ is proper, $m-p-2>0$, $n_i>1$, and $z_{ik},i=1,...,m$ have the same signs for $k=1,...,q$.
\end{theorem}

\begin{proof}
See Appendix \ref{proof_Theorem_3_1}.
\end{proof}


For posterior inference, similar to \cite{2024_Hamura} we consider the change of variable  $\xi_i=\alpha\omega_i$, for $i=1,...,m$, in the joint posterior given by (\ref{HB_posterior}). Then the joint posterior distribution (\ref{HB_posterior}) can be written as follows:
\begin{align}\label{HB_posterior_change_var}
\pi(\boldsymbol{\theta},\boldsymbol{\sigma}^2,\boldsymbol{\beta}&,\sigma_u^2,\alpha,\boldsymbol{\xi},\boldsymbol{\eta}\mid \boldsymbol{\mathcal{D}})\propto~(\sigma_u^2)^{-m/2}\times\alpha^{-m}\times\pi(\alpha)\nonumber\\
&\times\prod_{i=1}^m\left[\left(\frac{1}{\sigma_i^2}\right)^{\frac{n_i}{2}+(\xi_i+1)+1}\times\frac{\xi_i^{\xi_i} (\exp(\boldsymbol{z}_i^\top\boldsymbol{\eta}))^{\xi_i+1}}{\Gamma(\xi_i)}\times\pi^{\rm BP}(\xi_i/\alpha)\right]\nonumber\\
&\times\prod_{i=1}^m\exp\left\{-\frac{(y_i-\theta_i)^2/2+(n_i-1)D_i/2+\xi_i\exp(\boldsymbol{z}_i^\top\boldsymbol{\eta})}{\sigma_i^2}-\frac{(\theta_i-\boldsymbol{x}_i^\top\boldsymbol{\beta})^2}{2\sigma_u^2}\right\},
\end{align}

where $\boldsymbol{\xi}=(\xi_1,...,\xi_m)^\top$. As discussed in Section \ref{subsec:existing_models}, 
 $\gamma$ and the regression parameter $\eta_0$ for $z_{i1}=1$ are not identifiable in the STK2 model. In contrast to STK2, 
 the prior in the GL-BP model in (\ref{HB_posterior_change_var}) depends on $\xi_i$ and $\eta_0$ through both quantities $q_1=\xi_i$ and $q_2=\xi_i\exp(\eta_0)$. Thus, given $(q_1,q_2)$, we can obtain a unique pair of solutions $\xi_i=q_1$ and $\eta_0=\log(q_2/q_1)$. This implies that the parameters $\xi_i$ and $\eta_0$ in the GL-BP model are identifiable even when the intercept is included in the Exp-GVF.

From (\ref{HB_posterior_change_var}), the full conditional distributions of $(\theta_i,\sigma_i^2,\boldsymbol{\beta},\sigma_u^2)$ in the GL-BP model are in closed form similar to those in the existing models in Section \ref{subsec:existing_models}. 
The Gibbs sampling steps for drawing posterior samples of $(\theta_i,\sigma_i^2,\boldsymbol{\beta},\sigma_u^2)$  under the GL-BP  model are provided  in Algorithm \ref{alg:GL_BP}. To obtain the full conditional distribution of $\alpha$ in closed form we consider the BP prior as a convolution of two Gamma densities \citep{2017_Perez}. 
To this end, we introduce the new parameter $\boldsymbol{\rho}=(\rho_1,...,\rho_m)^\top$ in the joint posterior distribution in  (\ref{HB_posterior_change_var}) as follows:
\begin{align}\label{HB_posterior_change_var2}
\pi(\boldsymbol{\theta}&,\boldsymbol{\sigma}^2,\boldsymbol{\beta},\sigma_u^2,\alpha,\boldsymbol{\xi},\boldsymbol{\eta}, \boldsymbol{\rho}\mid \boldsymbol{\mathcal{D}})\propto~(\sigma_u^2)^{-m/2}\times\alpha^{-m}\times\pi(\alpha)\nonumber\\
&\times\prod_{i=1}^m\left[\left(\frac{1}{\sigma_i^2}\right)^{\frac{n_i}{2}+(\xi_i+1)+1}\times\frac{\xi_i^{\xi_i}(\exp(\boldsymbol{z}_i^\top\boldsymbol{\eta}))^{\xi_i+1}}{\Gamma(\xi_i)}\times  {\rm Ga}(\xi_i/\alpha\mid a,\rho_i)\times{\rm Ga}(\rho_i\mid b,1)\right]\nonumber\\
&\times\prod_{i=1}^m\exp\left\{-\frac{(y_i-\theta_i)^2/2+(n_i-1)D_i/2+\xi_i\exp(\boldsymbol{z}_i^\top\boldsymbol{\eta})}{\sigma_i^2}-\frac{(\theta_i-\boldsymbol{x}_i^\top\boldsymbol{\beta})^2}{2\sigma_u^2}\right\}.
\end{align}

By using (\ref{HB_posterior_change_var2}) we obtain that the full conditionals  of $\alpha$ and $\rho_i$ are a Generalized Inverse Gaussian (GIG) distribution and a Gamma distribution, respectively. Therefore Gibbs sampling steps for $\alpha$ and $\rho_i$  are included in Algorithm \ref{alg:GL_BP}. 
For the parameters $\xi_i$ and $\boldsymbol{\eta}$ the full conditional distributions are not in closed form. As proposed by   \cite{2024_Hamura}, posterior samples of
$\xi_i$  can be obtained using  the MH algorithm  proposed in \cite{2019_Miller}. 
However, because our  proposed model adopts a hierarchical structure that differs from that proposed model in \cite{2024_Hamura}, we found that the MH algorithm in \cite{2019_Miller} exhibited mixing issues for the parameter $\xi_i$. To improve the mixing behavior of the MCMC algorithm, we employ an adaptive MH scheme in which the variance of the Gaussian random-walk proposal is updated according to the acceptance rate. Specifically, the proposal variance is adjusted whenever the acceptance rate falls below 40\% or exceeds 60\%, thereby promoting more efficient exploration of the posterior distribution.
More specifically, proposed values of $\xi_i$  on the logarithmic scale are drawn from a Normal proposal distribution centered at the logarithm of the current value.
If the acceptance rate exceeds 0.6, the current proposal variance is multiplied by 1.1 to encourage broader exploration of the parameter space. On the other hand, if the acceptance rate falls below 0.4, the current proposal variance is divided by 1.1 to increase the acceptance probability. 
In our numerical studies, the  variance of the  proposal distribution is updated every $\ell=200$ iteration. 
Similarly, we implement an adaptive MH scheme for $\boldsymbol{\eta}$.
The details of the adaptive MH steps for $\xi_i$ and $\boldsymbol{\eta}$ are given in Algorithm \ref{alg:GL_BP}.




\clearpage

\begin{algorithm}[hp]
    \caption{MCMC algorithm for the proposed GL-BP model} \label{MCMC_HS}
    \begin{algorithmic}[1]
    \small
    \State \textbf{Parameters}: $\boldsymbol{\mathcal{V}}=(\boldsymbol{\theta},\boldsymbol{\sigma}^2,\boldsymbol{\beta},\sigma_u^2,\alpha,\boldsymbol{\rho},\boldsymbol{\xi},\boldsymbol{\eta})$. 
    \State \textbf{Input}: initial values $\boldsymbol{\mathcal{V}}^{(0)}$, observed data $\{n_i,y_i,D_i,\boldsymbol{x}_i,\boldsymbol{z}_i\}_{i=1}^m$,  $(\sigma_{\xi_1}^{2(1)},...,\sigma_{\xi_m}^{2(1)})$, and $\boldsymbol{\Sigma}_{\boldsymbol{\eta}}^{(1)}$.
    \State \textbf{Output}: posterior samples $\boldsymbol{\mathcal{V}}^{(s)}$ for $s=1,...,S$.
    \For {iteration $s=1,2,\ldots,S$}
	\State Sample $\theta_i^{(s)}\sim N\left((1-\gamma_i)y_i+\gamma_i\boldsymbol{x}_i^\top\boldsymbol\beta,\gamma_i\sigma_u^{2}\right)$, $\gamma_i=\frac{\sigma_i^{2}}{\sigma_i^{2}+\sigma_u^{2}}$ independently for $i=1,...,m$.
    \State Sample $\sigma_i^{2(s)}\sim {\rm IG}\left(\frac{n_i}{2}+\xi_i+1,~\frac{1}{2}[(y_i-\theta_i)^2+(n_i-1)D_i]+\xi_i\exp(\boldsymbol{z}_i^\top\boldsymbol{\eta})\right)$ independently for $i=1,...,m$.
    \State Sample $\boldsymbol{\beta}^{(s)}\sim N_p\left(\left(\sum_{i=1}^m\boldsymbol{x}_i\boldsymbol{x}_i^\top\right)^{-1}\left(\sum_{i=1}^m\boldsymbol{x}_i\theta_i\right),~\sigma_u^{2}\left(\sum_{i=1}^m\boldsymbol{x}_i\boldsymbol{x}_i^\top\right)^{-1}\right).$ 
    \State Sample $\sigma_u^{2(s)}\sim {\rm IG}\left(\frac{m}{2}-1,~\frac{1}{2}\sum_{i=1}^m(\theta_i-\boldsymbol{x}_i^\top\boldsymbol{\beta})^2\right)$.
    \State Sample $\alpha^{(s)}\sim{\rm GIG}(a_\alpha-am,~2b_\alpha,~2\sum_{i=1}^m\xi_i\rho_i)$ where ${\rm GIG}(a,b,\gamma)$ has density proportional to $x^{a-1}\exp(-bx/2-\gamma/2x)$.
    \State Sample $\rho_i^{(s)}\sim{\rm Ga}(a+b,1+\xi_i/\alpha)$  independently for $i=1,...,m$.
    \State Sample $\xi_i^{(s)},i=1,...,m$ independently using an adaptive Metropolis-Hasings (MH): 
        \begin{enumerate}
        \item Draw the proposal $\xi_i^*=\exp(\log(\xi_i^*))$ where $\log(\xi_i^*)\sim N(\log(\xi_i^{(s-1)}),\sigma_{\xi_i}^{2(s)})$ and 
        $$\sigma_{\xi_i}^{2(s)}=\begin{cases}
            \sigma_{\xi_i}^{2(s-1)}*1.1 & \text{if the acceptance rate is $>0.6$}\\
            \sigma_{\xi_i}^{2(s-1)}/1.1 & \text{if the acceptance rate is $<0.4$}
        \end{cases}$$
        The proposal variance $\sigma_{\xi_i}^{2(s)}$ is updated every $\ell$ iteration. 
        \item Compute the acceptance probability
        $$
        {\rm AP}(\xi_i)={\rm min}\left\{1,~\frac{\pi(\xi_i^*\mid\sigma_i^2,\rho_i,\alpha,\boldsymbol{\eta})\times1/\xi^{(s-1)}_i}{\pi(\xi_i^{(s-1)}\mid\sigma_i^2,\rho_i,\alpha,\boldsymbol{\eta})\times1/\xi_i^*}\right\},
        $$
        where the full conditional distribution of $\xi_i$ is $$\pi(\xi_i\mid\sigma_i^2,\rho_i,\alpha,\boldsymbol{\eta})\propto\frac{(\xi_i\exp(\boldsymbol{z}_i^\top\boldsymbol{\eta})/\sigma_i^2)^{\xi_i}}{\Gamma(\xi_i)}~\xi_i^{a-1}\exp\left\{-\frac{\xi_i\rho_i}{\alpha}-\frac{\xi_i\exp( \boldsymbol{z}_i^\top\boldsymbol{\eta})}{\sigma_i^2}\right\}.$$
        \item Sample $u_i\sim{\rm Uniform}(0,1)$.
        If $u_i<{\rm AP}(\xi_i)$,  $\xi_i^{(s)}=\xi_i^*$. Otherwise, $\xi_i^{(s)}=\xi_i^{(s-1)}$.
        \end{enumerate}
    \State Sample $\boldsymbol{\eta}^{(s)}$ using an adaptive MH algorithm:
        \begin{enumerate}
        \item Sample the proposal $\boldsymbol{\eta}^*\sim N_q(\boldsymbol{\eta}^{(s-1)},\boldsymbol{\Sigma}_{\boldsymbol{\eta}}^{(s)})$ where $\boldsymbol{\Sigma}_{\boldsymbol{\eta}}^{(s)}$ is updated every $\ell$ iteration based on the acceptance rate in the same manner as $\sigma_{\xi_i}^{2(s)}$.
        \item Compute the acceptance probability 
        \begin{align*}
        {\rm AP}(\boldsymbol{\eta})={\rm min}\left\{1,~\frac{\pi(\boldsymbol{\eta}^*\mid\boldsymbol{\sigma}^2,\boldsymbol{\xi})}{\pi(\boldsymbol{\eta}^{(s-1)}\mid\boldsymbol{\sigma}^2,\boldsymbol{\xi})}
        \right\},
        \end{align*}
        where $\pi(\boldsymbol{\eta}\mid\boldsymbol{\sigma}^2,\boldsymbol{\xi})\propto \prod_{i=1}^m\exp\left((\xi_i+1)~\boldsymbol{z}_i^\top\boldsymbol{\eta}\right)\exp\left(-\xi_i\exp(\boldsymbol{z}_i^\top\boldsymbol{\eta})/\sigma_i^2\right)$ is the full conditional distribution of $\boldsymbol{\eta}$.
        \item Sample $u\sim{\rm Uniform}(0,1)$. If $u<{\rm AP}(\boldsymbol{\eta})$, $\boldsymbol{\eta}^{(s)}=\boldsymbol{\eta}^*$. Otherwise,  $\boldsymbol{\eta}^{(s)}=\boldsymbol{\eta}^{(s-1)}$.
        \end{enumerate}
    \EndFor
	\end{algorithmic} 
    \label{alg:GL_BP}
\end{algorithm}

\clearpage

\section{Simulation studies}\label{sec:simulation}


We evaluate the performance of the proposed GL-BP model and the existing Bayesian SAE models  in Table \ref{tab:model_structures} given in Section \ref{sec:computation} through simulation studies. We follow similar simulation settings to those given in \cite{2017_Sukasawa} which were also previously considered by \cite{2003_Wang&Fuller} and \cite{2014_Maiti}. For the simulated data, 80\% of the $n_i$ are generated from the interval (7,20), and the remaining 20\% are generated from the interval (21,100).
 We generate observations for each small area by using 
\begin{align*}
y_{ij}=\beta_0+\beta_1x_i+u_i+e_{ij},\quad j=1,...,n_i,\quad i=1,...,m,    
\end{align*}
where $u_i\sim N(0,\sigma_u^2)$ and $e_{ij}\sim N(0,n_i\sigma_i^2)$. Then the model for the small area mean is 
\begin{align*}
    y_i=\beta_0+\beta_1x_i+u_i+e_i,\quad i=1,...,m,  
\end{align*}
 where $y_i=\bar{y}_{i\cdot}=n_i^{-1}\sum_{j=1}^{n_i} y_{ij}$ and $e_i=n_i^{-1}\sum_{j=1}^{n_i} e_{ij}$. Therefore, $y_i\mid\theta_i \sim N(\theta_i,\sigma_i^2)$ where $\theta_i=\beta_0+\beta_1x_i+u_i$, that is, $\theta_i\sim N(\beta_0+\beta_1x_i,\sigma_u^2)$, and $e_i\sim N(0,\sigma_i^2)$. The parameter of interest is the mean $\theta_i$ in the $i$-th small area. The direct estimator of $\sigma_i^2$ is given by
\begin{align*}
    D_i=\frac{1}{n_i}\frac{1}{n_i-1}\sum_{j=1}^{n_i}(y_{ij}-y_i)^2,
\end{align*}
noting that $D_i\mid\sigma_i^2\sim{\rm Ga}((n_i-1)/2,(n_i-1)/(2\sigma_i^2))$. We generate the covariate $x_i$ using a Uniform distribution on the interval (1, 2), consider the regression parameters as $\beta_0 = 0.5$ and $\beta_1 = 0.8$ and the variance of the random effects is set equal to $\sigma_u^2 = 3$. We consider $m = 30$  as in \cite{2014_Maiti} and \cite{2017_Sukasawa}. To evaluate the performance of the models  we consider four cases to generate the sampling variances $\sigma_i^2$ as follows:

\begin{enumerate}
    \item Case 1: $\sigma_i^2 = 5\exp(-0.3\log(n_i))$,
    \item Case 2: $\sigma_i^2 = 5\lambda_i\exp(-0.3\log(n_i))$ where $\lambda_i \sim 0.2\text{Log-Normal}(0, 0.01^2) + 0.8\delta_{\mu_{i}}$, 
    \item Case 3: $\sigma_i^2 = 5\lambda_i\exp(-0.3\log(n_i))$ where $\lambda_i \sim 0.5\text{Log-Normal}(0, 0.01^2) + 0.5\delta_{\mu_{i}}$, 
    \item Case 4: $\sigma_i^2 = 5\lambda_i\exp(-0.3\log(n_i))$ where $\lambda_i \sim 0.8\text{Log-Normal}(0, 0.01^2) + 0.2\delta_{\mu_{i}}$,     
\end{enumerate}

where $\mu_i=1$ and $\delta_{\mu_{i}}$ denotes a point mass at $\mu_i$ and $\text{Log-Normal}(\mu_{\text{L-N}}, \sigma^2_{\text{L-N}})$ denotes a Log-Normal distribution
with mean $\mu_{\text{L-N}}$ and variance $\sigma^2_{\text{L-N}}$. In Case 1 we consider that the variance parameters are equal to the Exp-GVFs. Cases 2-4 incorporate random effects in  20\%, 50\% and 80\% of the small areas where $n_i<20$.


\subsection{Results of the simulation studies}\label{subsec:simulation_results}

We apply the following four models to the simulated data where the prior distribution of $\sigma_i^2$ is as follows:
\begin{align*}
    {\rm YC}&:\sigma_i^2\sim{\rm IG}(0.0001,0.0001), && \hspace{.21cm}{\rm STK1}:\sigma_i^2\sim{\rm IG}(a_i,b_i\gamma),\\
    {\rm STK2}&:\sigma_i^2\sim{\rm IG}(a_i,b_i\gamma\exp(\eta\log(n_i))), &&\text{GL-BP}:\sigma_i^2\sim{\rm IG}(\alpha\omega_i+1,\alpha\omega_i\exp(\eta_0+\eta_1\log(n_i))).
\end{align*}

We set  $a_i=2$ and $b_i=1/n_i$ for STK1 and STK2, as recommended by \cite{2017_Sukasawa}.
For GL-BP, we use $\omega_i\sim{\rm BP}(2,1/2)$, as discussed in Section \ref{subsec:GL_shrinkage_prior}, and $\alpha\sim{\rm Ga}(1,1)$, as discussed in Section \ref{subsec:computation_proposed_model}. 
For each model, the corresponding MCMC algorithm is run for 80,000 iterations with the first 40,000 discarded as burn-in. The remaining samples are thinned by keeping one out of every four samples, resulting in 10,000 posterior samples. The posterior estimates of the model parameters are calculated as the means of these remaining 10,000 MCMC samples. We compute the Average Squared Deviations (ASDs) and the Absolute Bias based on $R=1000$ replications as in \cite{2017_Sukasawa} as follows:

\begin{align*}
    {\rm ASD}=\frac{1}{mR}\sum_{i=1}^m\sum_{r=1}^R\left(\hat{\theta}_i^{(r)}-\theta_i^{(r)}\right)^2,\quad 
    {\rm Bias}=\frac{1}{mR}\sum_{i=1}^m\left|\sum_{r=1}^R\left(\hat{\theta}_i^{(r)}-\theta_i^{(r)}\right)\right|,
\end{align*}

where $\theta_i^{(r)}$ and $\hat\theta_i^{(r)}$ are the true value and the posterior mean estimates in the $r$-th replication, respectively. Moreover, to evaluate the uncertainty produced by the posterior estimators, we compute the coverage probabilities (CPs) of the 95\% credible intervals as follows:

\begin{align*}
    {\rm CP}=\frac{1}{mR}\sum_{i=1}^m\sum_{r=1}^R\text{1}\left\{\theta_i\in[\hat{L}^{(r)}_{\theta_i},\hat{U}^{(r)}_{\theta_i}]\right\},
\end{align*}


where $\hat{L}^{(r)}_{\theta_i}$ and $\hat{U}^{(r)}_{\theta_i}$ are the 2.5\% and 97.5\% quantiles of the posterior samples of $\theta_i$ in the $r$-th iteration, respectively.

 
\begin{table}[ht]
\centering
\scriptsize
\renewcommand{\arraystretch}{1.3}


\begin{tabular}{@{\hspace{-0.7cm}}  c @{\hspace{0.5cm}} c}

\begin{tabular}{l|c|c|c|c|c}
\toprule

\multicolumn{6}{c}{\bfseries Case 1} \\

\midrule


& \multicolumn{2}{c}{\bfseries $\theta_i$}
& \multicolumn{2}{|c|}{\bfseries $\sigma_i^2$} \\

\cmidrule(lr){2-3} \cmidrule(lr){4-6}

\textbf{Model}
& \textbf{ASD } & \textbf{Bias} 
& \textbf{ASD } & \textbf{Bias} & \textbf{CP (\%)} \\

\midrule

    YC & 1.358 & 0.085 & 0.777 & 0.069 & 94.0 \\ 
  STK1 & 1.376 & 0.085 & 0.671 & 0.295 & 93.0 \\ 
  STK2 & 1.344 & 0.089 & 0.672 & 0.268 & 94.7 \\ 
  GL-BP & 1.339 & 0.088 & 0.134 & 0.046 & 94.9\\ 
   
\midrule

\multicolumn{6}{c}{\bfseries Case 3} \\

\midrule

& \multicolumn{2}{c}{\bfseries $\theta_i$}
& \multicolumn{2}{|c|}{\bfseries $\sigma_i^2$} \\

\cmidrule(lr){2-3} \cmidrule(lr){4-6}

\textbf{Model}
& \textbf{ASD } & \textbf{Bias} 
& \textbf{ASD } & \textbf{Bias} & \textbf{CP (\%)} \\

\midrule

  YC & 1.393 & 0.067 & 0.867 & 0.050 & 93.9 \\ 
  STK1 & 1.416 & 0.067 & 0.713 & 0.286 & 92.9 \\ 
  STK2 & 1.365 & 0.066 & 0.743 & 0.270 & 94.9 \\ 
   GL-BP & 1.347 & 0.068 & 0.167 & 0.054 & 95.1 \\

\bottomrule
\end{tabular}

&

\renewcommand{\arraystretch}{1.3}


\begin{tabular}{l|c|c|c|c|c}
\toprule


\multicolumn{6}{c}{\bfseries Case 2} \\
\midrule


& \multicolumn{2}{c}{\bfseries $\theta_i$}
& \multicolumn{2}{|c|}{\bfseries $\sigma_i^2$} \\

\cmidrule(lr){2-3} \cmidrule(lr){4-6}

\textbf{Model}
& \textbf{ASD } & \textbf{Bias} 
& \textbf{ASD } & \textbf{Bias} & \textbf{CP (\%)} \\

\midrule

     YC & 1.449 & 0.060 & 0.819 & 0.052 & 93.4 \\ 
  STK1 & 1.477 & 0.061 & 0.679 & 0.287 & 92.2 \\ 
  STK2 & 1.417 & 0.058 & 0.714 & 0.267 & 94.7 \\ 
   GL-BP & 1.404 & 0.056 & 0.145 & 0.050 & 94.9 \\

\midrule
 
\multicolumn{6}{c}{\bfseries Case 4} \\

\midrule

& \multicolumn{2}{c}{\bfseries $\theta_i$}
& \multicolumn{2}{|c|}{\bfseries $\sigma_i^2$} \\

\cmidrule(lr){2-3} \cmidrule(lr){4-6}

\textbf{Model}
& \textbf{ASD } & \textbf{Bias} 
& \textbf{ASD } & \textbf{Bias} & \textbf{CP (\%)} \\

\midrule

  YC & 1.417 & 0.045 & 0.834 & 0.047 & 93.4 \\ 
  STK1 & 1.440 & 0.046 & 0.681 & 0.274 & 92.4 \\ 
  STK2 & 1.388 & 0.045 & 0.768 & 0.288 & 94.7 \\ 
   GL-BP & 1.371 & 0.046 & 0.168 & 0.052 & 94.8 \\

\bottomrule
\end{tabular}

\end{tabular}

\caption{Simulation Cases 1-4. Simulated Average Squared Deviations (ASDs), absolute biases, and coverage probabilities (CPs) of 95\% nominal levels.}
\label{Table_S}
\end{table}
Table \ref{Table_S} shows the ASDs,  absolute biases, and  CPs obtained in the different simulation cases. 
According to the results, the proposed  GL-BP model  yields smaller  ASDs and absolute biases  for both $\theta_i$ and $\sigma_i^2$ than those obtained under  YC, STK1 and STK2. 
We note that the improvement achieved by the proposed GL-BP model is not affected by the proportion of areas with random effects. Importantly, we also note that by using the Bayesian credible 
intervals  obtained under the GL-BP model we attain the 95\% nominal levels in the different simulation cases.

\section{Applications}\label{sec:application}
We apply the Bayesian SAE models in Table \ref{tab:model_structures} to the Corn data in Section \ref{subsec:corndata} and the Educational Attainment Index--At Least High School (PEAI--AHS) data in Section \ref{subsec:PEAI-AHS}. We use $\omega_i\sim{\rm BP}(2,1/2)$ and $\alpha\sim{\rm Ga}(1,1)$ for the proposed GL-BP models. For the models in \cite{2017_Sukasawa}, we use $a_i=2$ and $b_i=1/n_i$. 
For each model, the corresponding MCMC algorithm is run for 80,000 iterations. The posterior estimates of the model parameters are calculated as the means of 10,000 MCMC posterior samples obtained after a burn-in of 40,000 iterations with every  fourth sample retained.

\subsection{Corn data}\label{subsec:corndata}

The data set on Corn production from the U.S. Department of Agriculture was previously studied by many researchers, e.g., \cite{1988_Battese}, \cite{2006_You&Chapman}, \cite{2012_Dass}, and \cite{2017_Sukasawa}, to evaluate their proposed models. The original data analyzed by \cite{1988_Battese} consists of 12 Iowa counties with sample sizes in each county ranging from 1 to 5. However, according to Theorem \ref{thm:proper_posterior}, our proposed model requires $n_i>1$ to ensure that the posterior distribution is proper. Hence, 
we use the modified data consisting of $m=8$ counties with $n_i\geq3$ as given in Table 6 of \cite{2012_Dass} and studied previously in \cite{2017_Sukasawa}.  In this application the mean  crop hectares for corn $(y_i)$ are the direct survey estimates. The covariates  obtained from the Land Observatory Satellites (LANDSAT) are
the mean numbers  of pixels of corn $(x_{1i})$ and the mean numbers  of pixels of soybean $(x_{2i})$. The values of $y_i$, $x_{1i}$, $x_{2i}$, and $\sqrt{D_i}$ are scaled by 100. Under those settings, the modified FH model is given by
\begin{align*}
    y_i\mid\theta_i,\sigma_i^2\sim N(\theta_i,&\sigma_i^2),\quad\theta_i\mid\beta_0,\beta_1,\beta_2,\sigma_u^2\sim N(\beta_0+\beta_1x_{1i}+\beta_2x_{2i},\sigma_u^2),\\
    D_i&\mid\sigma_i^2\sim{\rm Ga}((n_i-1)/2,(n_i-1)/(2\sigma_i^2)),
\end{align*}

where $i=1,...,8$. We consider $\log(n_i)$ as a covariate in the Exp-GVF for the STK2 and GL-BP models. The posterior distribution under the proposed GL-BP model is  proper because the covariate  $\log(n_i)$ is positive for all counties (see Theorem \ref{thm:proper_posterior}). We apply the four Bayesian SAE models in Table \ref{tab:model_structures} to the Corn data set. 
We consider a  Uniform  prior for $\sigma_i^2$, $\pi(\sigma_i^2)\propto1/\sigma_i^2$, under the YC model. 
For the remaining models, we use Inverse-Gamma priors: $\sigma_i^2\sim{\rm IG}(a_i,b_i\gamma)$ for STK1, $\sigma_i^2\sim{\rm IG}(a_i,b_i\gamma\exp(\eta_1 \log(n_i)))$ for STK2, and $\sigma_i^2\sim{\rm IG}(\alpha\omega_i+1,\alpha\omega_i\exp(\eta_0+\eta_1\log(n_i)))$ for GL-BP. We assume $\pi(\boldsymbol{\beta},\sigma_u^2)\propto1$ for all four models.

\begin{table}[ht!]
    \centering
    \resizebox{.9\textwidth}{!}{
    \begin{tabular}{l|ccccccccc}
    \hline
     & $\beta_0$ & $\beta_1$ & $\beta_2$ & $\sigma_u^2$ & $\gamma$ & $\eta_0$ & $\eta_1$ & $\alpha$ & DIC\\
    \hline
      YC & -1.721 & 0.741 & 0.362 & 0.297 & - & - &  - & - &  -10.34 \\
      STK1 & -1.537 & 0.666  & 0.369  & 0.262 & 0.566 & - & - & - & -16.58 \\
      STK2 & -1.396 & 0.633 & 0.347 & 0.298 & 1.629 & - & -0.631 & - & -16.96  \\
      GL-BP & -1.783 & 0.723 & 0.409 & 0.332 & - & -0.509 & -1.346 & 1.058 & \textbf{-17.63} \\
      \hline
    \end{tabular}}
    \caption{Posterior mean estimates of model parameters and Deviance Information Criterion (DIC) in the Corn data set. 
    }
    \label{tab:corndata}
\end{table}

For model comparison, we calculate the Deviance Information Criterion (DIC) of \cite{2002_Spiegelhalter} given by ${\rm DIC}=2\overline{D(\boldsymbol{\phi})}-D(\bar{\boldsymbol{\phi}})$ where $\boldsymbol{\phi}=(\boldsymbol{\beta},\sigma_u^2,\sigma_1^2,...,\sigma_m^2)$, $D(\boldsymbol{\phi})$ is (-2) times log-marginal likelihood function, and $\overline{D(\boldsymbol{\phi})}$ and $\bar{\boldsymbol{\phi}}$ are the posterior means of $D(\boldsymbol{\phi})$ and $\boldsymbol{\phi}$, respectively. 
For the Bayesian SAE models with unknown sampling variances in Table \ref{tab:model_structures}, the marginal likelihood is calculated as $\prod_{i=1}^m[ N(y_i\mid\boldsymbol{x}_i^\top\boldsymbol{\beta},\sigma_i^2+\sigma_u^2)\times{\rm Ga}(D_i\mid (n_i-1)/2,(n_i-1)/(2\sigma_i^2))]$.

Table \ref{tab:corndata} presents the posterior estimates of the model parameters and DIC values. According to the DIC values in Table \ref{tab:corndata}, the proposed GL-BP model 
is the most suitable for the Corn data set. 

In Figure $\ref{fig:corndata}$ (left), we compare the DVEs with the posterior estimates of $\sigma_i^2$ produced by YC, STK1, STK2, and GL-BP. 
Figure $\ref{fig:corndata}$ (left) illustrates that the posterior estimator of $\sigma_i^2$ of GL-BP (blue) is, in almost all counties, a weighted average of the DVE (black) and the Exp-GVF (red). The GL-BP model imposes shrinkage toward the posterior estimates of the Exp-GVF possibly due to the unreliability of the DVEs  arising from the very small sample sizes, $3\leq n_i\leq5$.
Consistent with the findings of \cite{2017_Sukasawa}, the posterior estimate of $\sigma_i^2$ under STK1 and STK2 shrinks the DVE toward a prior mean 
whereas the posterior estimate obtained under YC remains close to $D_i$. Figure \ref{fig:corndata} (right) displays  the 95\% credible intervals of $\theta_i$ for each model. The result shows that GL-BP produces similar but slightly narrower credible intervals of $\theta_i$ than those obtained with STK1 and STK2. As originally noted by \cite{2017_Sukasawa}, the credible interval of YC in County 1 is much narrower than that obtained with
STK1, STK2 and GL-BP. However,  the estimate of $\theta_1$ obtained with YC may not be reliable in the Corn data set because of the instability of the variance estimate when $n_i=3$.
Figure \ref{fig:corndata2} illustrates the posterior distributions of the small area posterior estimates for  $\sqrt{\sigma_i^2}$ and $\theta_i$ under the GL-BP model. 
As illustrated in Figure \ref{fig:corndata2},  except for County 1 which exhibits a very small DVE,
the posterior distributions under the proposed model closely align with the survey estimates.

\begin{figure}[ht!]
\centering
    \begin{subfigure}[t]{0.5\textwidth}
         \centering
         \includegraphics[width=\textwidth]{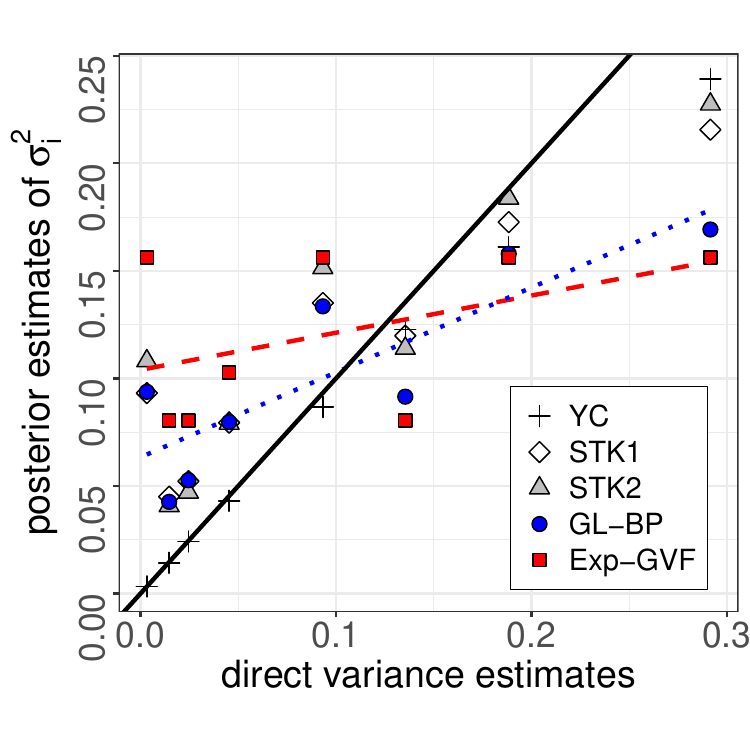}
     \end{subfigure}~
     \begin{subfigure}[t]{0.5\textwidth}
         \centering
         \includegraphics[width=\textwidth]{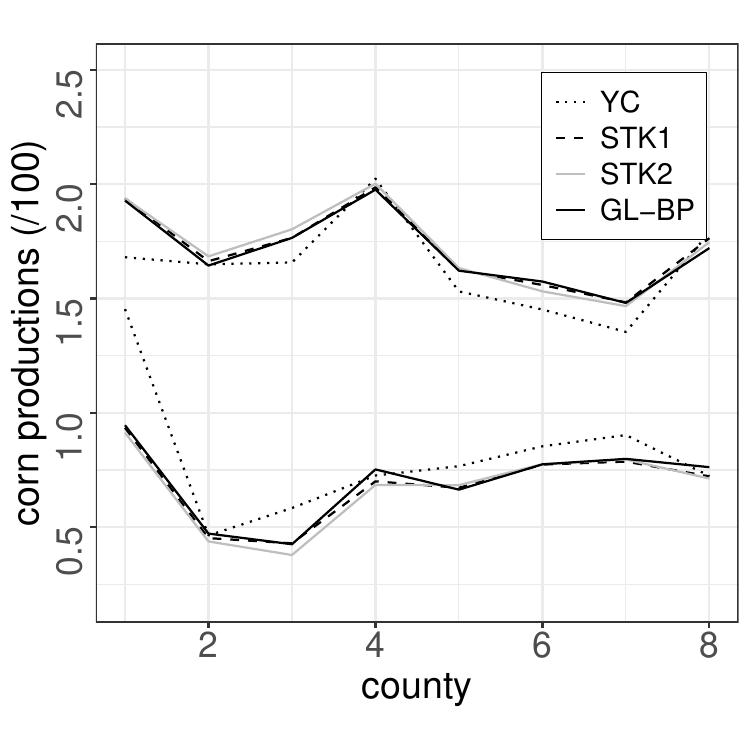}
     \end{subfigure}
     \vspace{-1cm}
     \caption{ Direct Variance Estimates (DVEs) versus posterior estimates of $\sigma_i^2$ in the Corn data (left). The black solid line represents the DVEs. The red dashed line represents a linear trend of the posterior estimates of the Exp-GVF, $\exp(\eta_0+\eta_1\log(n_i))$, under the proposed GL-BP model. The blue dotted line represents a linear trend of the posterior estimates of $\sigma_i^2$ produced by the GL-BP model.
   95\% credible intervals of $\theta_i$ in the Corn data (right).}
        \label{fig:corndata}
\end{figure}

\begin{figure}[ht!]
\centering
    \begin{subfigure}[t]{0.5\textwidth}
         \centering
         \includegraphics[width=\textwidth]{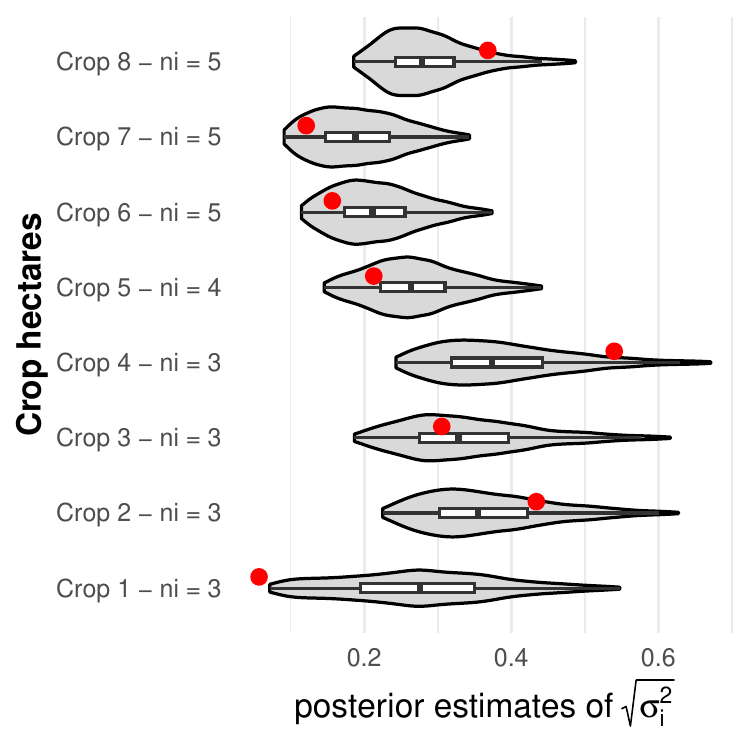}
     \end{subfigure}~
     \begin{subfigure}[t]{0.5\textwidth}
         \centering
         \includegraphics[width=\textwidth]{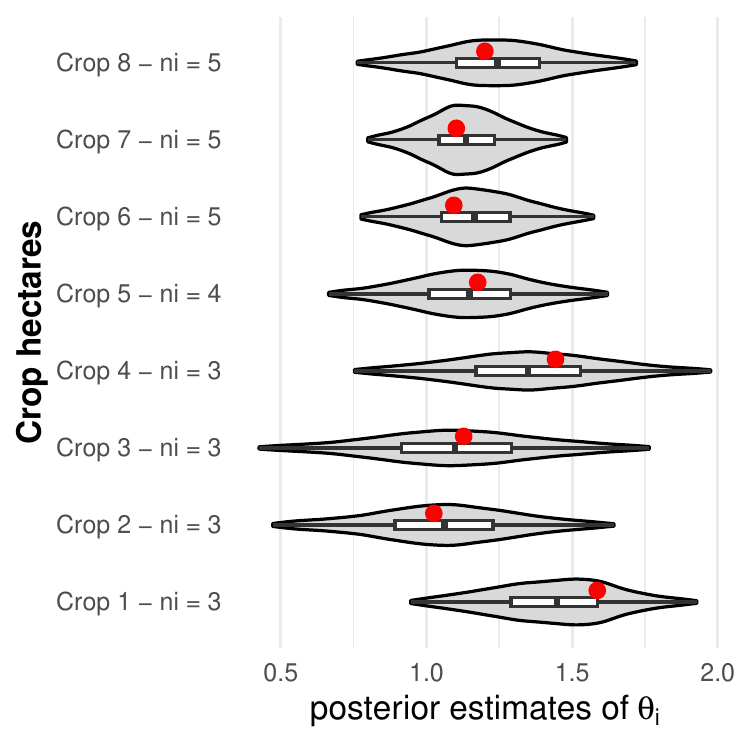}
     \end{subfigure}
     \vspace{-1cm}
     \caption{Posterior estimates of $\sqrt{\sigma_i^2}$ (left) and $\theta_i$ (right) in the Corn data. The red dots on the left and right plots are the squared roots of the Direct Variance Estimates (DVEs) and the direct survey estimates ($y_i$), respectively.  }
        \label{fig:corndata2}
\end{figure}


\subsection{Prevalence of the Educational Attainment Index--At Least High School (PEAI--AHS) at subnational levels} \label{subsec:PEAI-AHS}

Educational Attainment Index (EAI) refers to the highest level of education successfully completed by an individual (United Nations Educational, Scientific and Cultural Organization \citep{UNESCO}). The EAIs are important for policy making in low- and middle-income countries to track the progress of education and assess disparities, helping governments design policies that align with the Sustainable Development Goals (SDGs) proposed by the United Nations \citep{UN2015Agenda2030}. Particularly, the EAI is closely related to SDG indicators on education and human development, which aim to ensure \textit{inclusive and equitable quality education and promote lifelong learning opportunities for all} \citep{UN_SDG_ExtendedReport_2025}. Although national-level estimates of the EAI are available from the UNESCO Institute for Statistics (UIS) \citep{UNESCO} and the World Bank \citep{World-Bank}, subnational-level estimates  are  more important for informing public policy in developing countries.
\begin{figure}[ht!]
    \vspace{-0.8cm}
\centering    
    \begin{subfigure}[t]{0.67\textwidth}
 \hspace{-3.5cm}
    \vspace{-0.4cm}
         \centering
         \includegraphics[width=\textwidth]{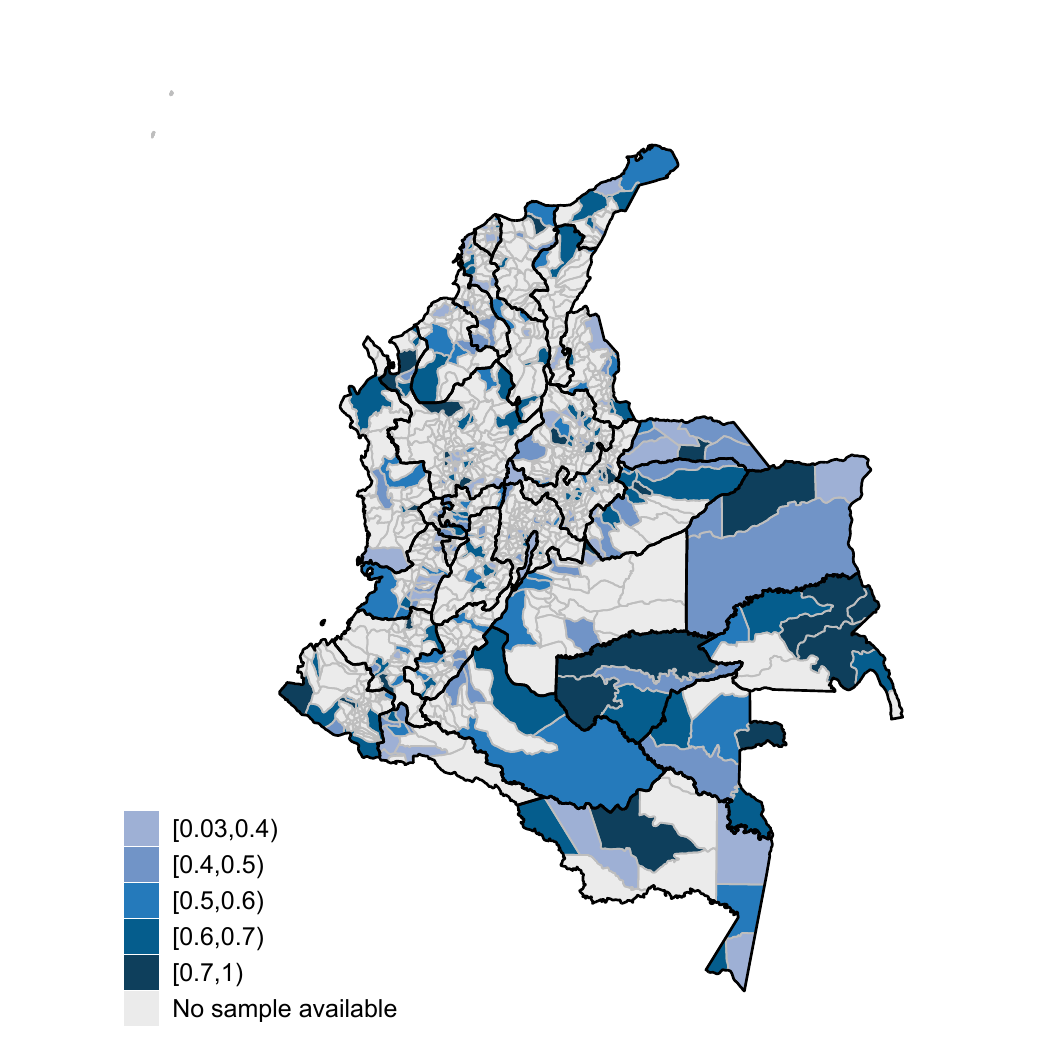}
     \end{subfigure}~
     \begin{subfigure}[t]{0.57\textwidth}
    \hspace{-6.5cm}
         \centering
         \includegraphics[width=\textwidth]{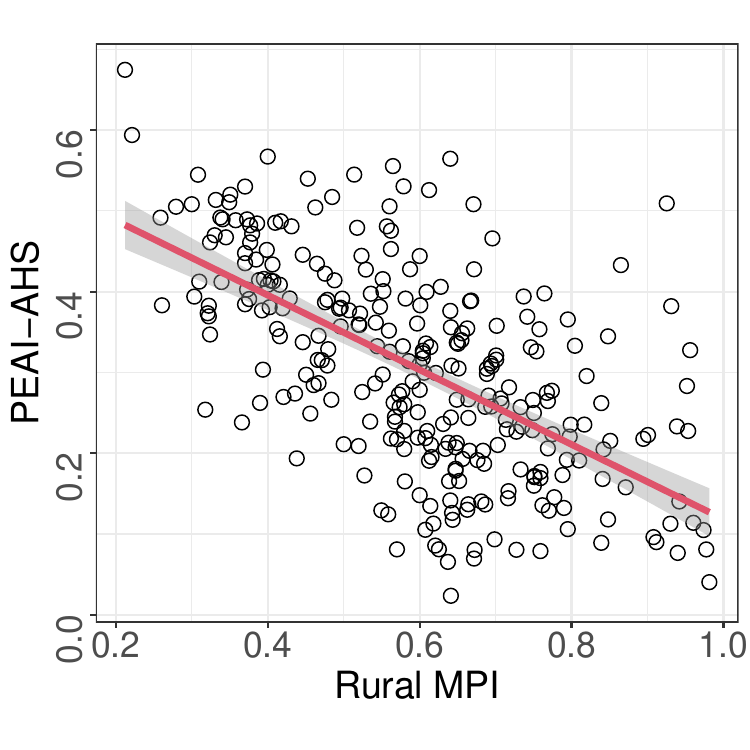}
     \end{subfigure}
     \vspace{-0.5cm}
     \caption{Direct estimates of the Prevalence of the Educational Attainment Index--At Least High School (PEAI-AHS; $y_i$) in municipalities in Colombia in 2015 (left).  $y_i$ versus the Rural Multidimensional Poverty Indexes (Rural MPI; $x_i$) in 2015 (right).}
        \label{fig:data}
\end{figure} 
In this section, we estimate the Prevalence of the Educational Attainment Index--At Least High School (PEAI--AHS) in $m=289$ municipalities in Colombia in 2015 for women aged 29–47. 
To produce the direct estimate of the PEAI–AHS ($y_i$) and the DVEs of the PEAI–AHS ($D_i$), we use the microdata provided by the Demographic and Health Survey (DHS) in Colombia in 2015  \citep{Profamilia_COL_DHS_2005} and the direct estimators proposed by \cite{Hajek1971}. 

\subsubsection{Exploratory data analysis of the PEAI--AHS application}\label{subsec:exploratory_PEAI-AHS}

According to \cite{cox2010educational}, Latin America has some of the highest levels of educational inequality, which are closely associated with income inequality and poverty. In the specific case of Colombia, the link between low educational attainment and poverty is especially strong due to socio-economic disadvantages and insufficient school resources, both of which significantly hinder student achievement as pointed out by \cite{RangelLleras2010} and \cite{BorchersCunha2025}. Therefore, we use the Rural Multidimensional Poverty Index (Rural MPI; $x_i$) in 2015 as a covariate in the FH model (\ref{FH_model}) calculated by the Colombia NSO with the methodology developed by the Food and Agriculture Organization of the United Nations (FAO) together with Oxford Poverty and Human Development Initiative (OPHI) \citep{FAO_OPHI_2022_RMPI}. In this PEAI--AHS application, the correlation between $y_i$ and $x_i$ is -0.57. Figure \ref{fig:data} (left) displays that lower direct estimates $y_i$ are observed in most municipalities with approximately 92\% of them having an education index below 0.5. Figure \ref{fig:data} (right) illustrates the relationship between $y_i$ and $x_i$.

We consider two covariates for the Exp-GVF. The first covariate is the logarithm of the sample size $n_i$. 
Because the number of households recorded in the census is closely related to population growth and often serves as a key auxiliary variable in population-based surveys such as the DHS, the second covariate  is the logarithm of the projected number of households for 2015  ($z_i$).
The plots in Figures \ref{fig:PEAI-AHS} illustrate the  relationships between $\log(D_i)$ and $\log(n_i)$ with a correlation of -0.88, and between $\log(D_i)$ and $\log(z_i)$ with a correlation of -0.43. 
Specifically, in this SAE application we consider the following two GVF models:
\begin{align*}
    \text{GVF1}:~\log(D_i)&=\eta_0+\eta_1\log(n_i)+\varepsilon_i,\\   \text{GVF2}:~\log(D_i)&=\eta_0+\eta_1\log(n_i)+\eta_2\log(z_i)+\varepsilon_i.
\end{align*}

Figure \ref{fig:PEAI-AHS-GVFplots} (left) compares the square roots of the DVEs and the GVF-smoothed variances obtained with  the exponential of the fitted values produced with the GVF model in (\ref{GVF_model}). 
Figure \ref{fig:PEAI-AHS-GVFplots} (left) displays that the DVEs are highly variable when $n_i$ is small. The smoothed variance estimators obtained with both GVF1 and GVF2 methods help smooth out the DVEs. When $n_i$ is large, the DVEs and smoothed variance estimates are similar. This behavior has been widely observed in SAE applications and was recently highlighted by \cite{2023_You&Hidiroglou}. As pointed out by   \cite{2003_Wang&Fuller}  and  \cite{2009_Maples}, the use of the degrees of freedom $\nu_i=n_i - 1$ in (\ref{gamma_model_Di})
where  $D_i\mid\sigma_i^2\sim{\rm Ga}(\nu_i/2,\nu_i/(2\sigma_{i}^{2}))$ is theoretically justified only under simple random sampling. Since the DHS survey under consideration employs a complex sampling design and the parameter of interest is a prevalence, this assumption does not hold exactly, and an approximation to the degrees of freedom is therefore required. 
 In this paper  to estimate $\nu_i$  we consider the Delta method. The Delta-method approximation is obtained
by considering the variance estimator $D_i=y_i(1-y_i)/n_i$ as a smooth function of the approximately Normal direct estimator $y_i$. By matching
the resulting approximate variance of $D_i$ with the variance of the Gamma distribution in (\ref{gamma_model_Di}) we obtain  $\hat{\nu}_i=2n_iy_i(1-y_i)/(1-2y_i)^2$. The details
of the implementation of the Delta method are given in Appendix \ref{application_sup} of the Supplementary Material. Figure \ref{fig:PEAI-AHS-GVFplots} (right) illustrates  the linear relationship between
$n_i$ and the estimated degrees of freedom $\hat{\nu}_i$.

\begin{figure}[ht!]
\centering
     \begin{subfigure}[t]{0.47\textwidth}
         \centering
         \includegraphics[width=\textwidth]{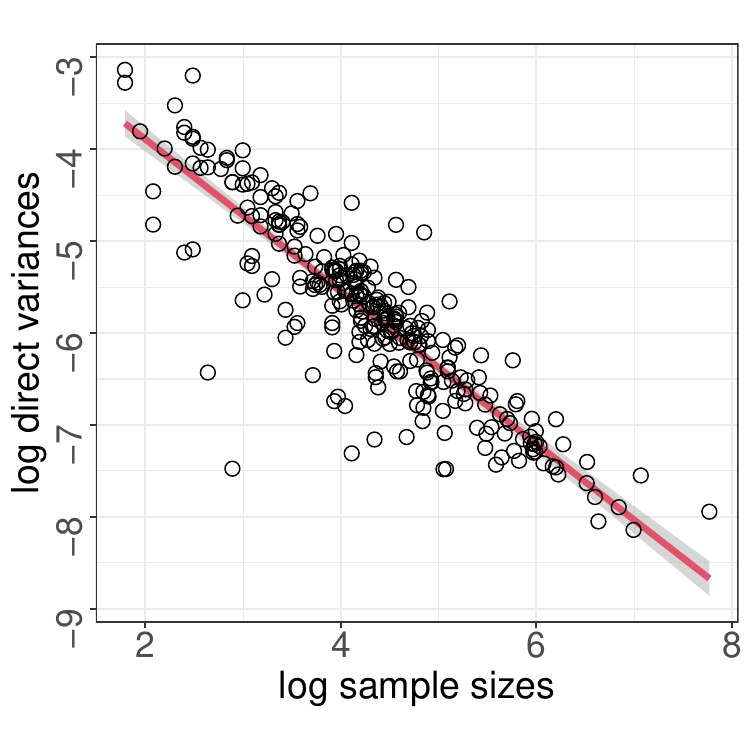}
     \end{subfigure}~
     \begin{subfigure}[t]{0.47\textwidth}
         \centering
         \includegraphics[width=\textwidth]{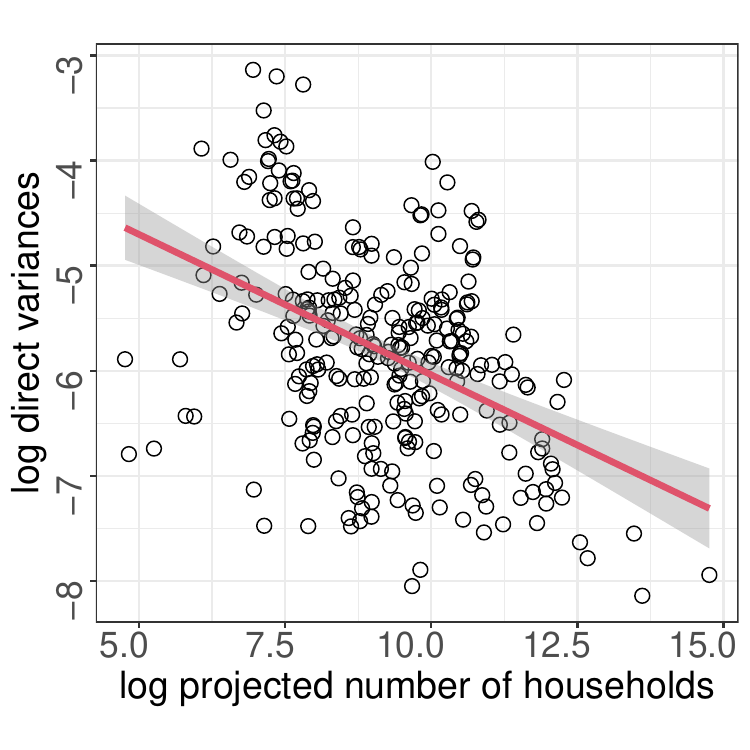}
     \end{subfigure}
     \vspace{-0.3cm}
     \caption{Log Direct Variance Estimates (log DVEs) versus covariates considered in the GVF method: log sample size (left) and log projected number of households in 2015 (right) in the PEAI–AHS data.}
        \label{fig:PEAI-AHS}
\end{figure}

\begin{figure}[ht!]
\centering    
     \begin{subfigure}[t]{0.60\textwidth}
         \centering   \hspace{-2.3cm}
         \includegraphics[width=\textwidth]{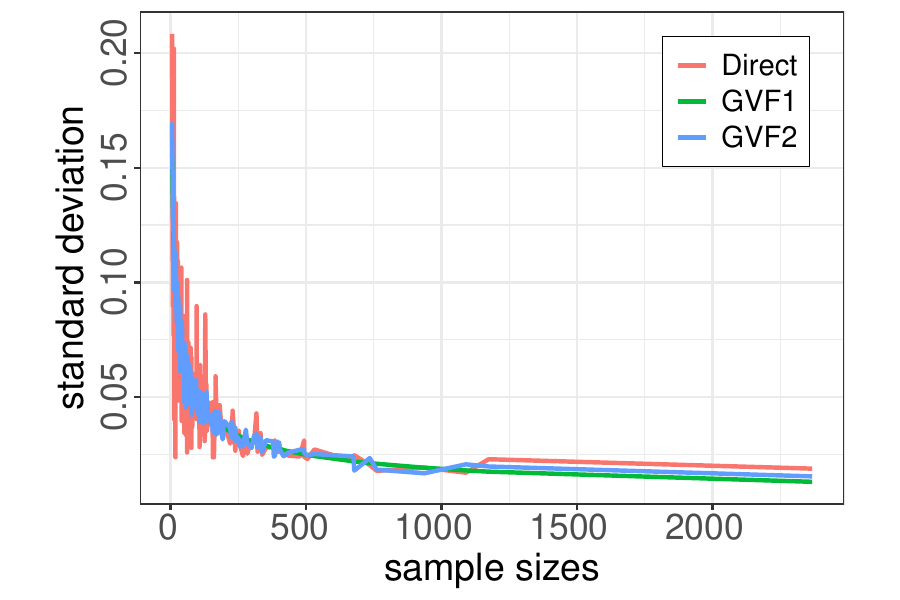}
     \end{subfigure}~
     \begin{subfigure}[t]{0.47\textwidth}
         \centering \hspace{-3.5cm}
         \includegraphics[width=\textwidth]{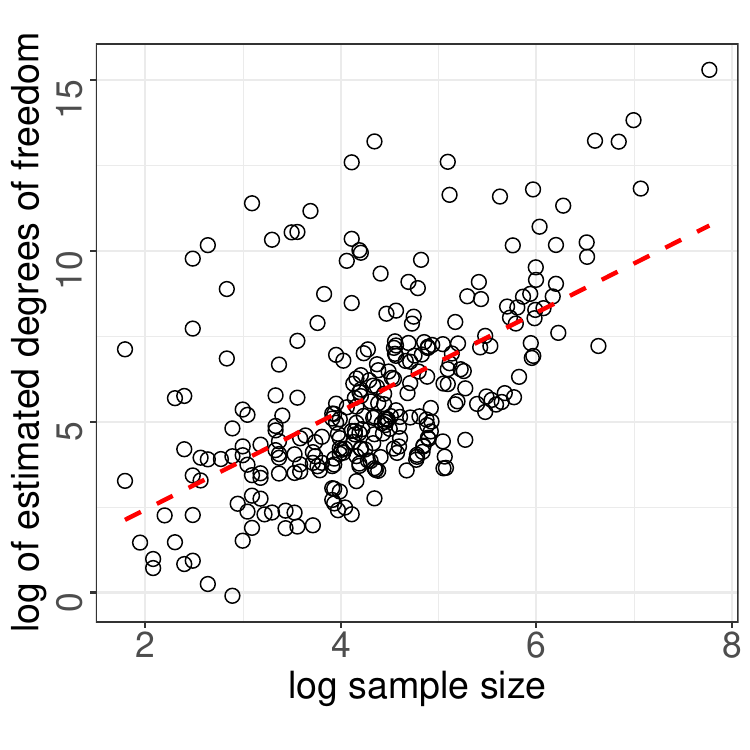}
     \end{subfigure}
     \vspace{-0.3cm}
     \caption{Comparison of square roots of  Direct Variance Estimates (DVEs) and GVF-smoothed variance estimates in the PEAI--AHS data (left). Estimated degrees of freedom versus sample sizes on the logarithmic scale (right). }
    \label{fig:PEAI-AHS-GVFplots}
\end{figure}

\subsubsection{Results in the PEAI--AHS application}\label{subsec:results_PEAI-AHS}
Under the settings described in Section \ref{subsec:exploratory_PEAI-AHS}, the modified FH model is given by
\begin{align*}
    y_i\mid\theta_i,\sigma_i^2\sim N(\theta_i&,\sigma_i^2),\quad \theta_i\mid\beta_0,\beta_1,\sigma_u^2\sim N(\beta_0+\beta_1x_i,\sigma_u^2),\\
    D_i&\mid\sigma_i^2\sim{\rm Ga}(\hat{\nu}_i/2,\hat{\nu}_i/(2\sigma_i^2)),
\end{align*}

for $i=1,...,289$ where $\hat{\nu}_i=2n_iy_i(1-y_i)/(1-2y_i)^2$. 
We implement the following six models to the PEAI--AHS  application:
\begin{enumerate}
    \item YC where $\pi(\sigma_i^2)\propto1/\sigma_i^2$,
    \item STK1 where $\sigma_i^2\sim{\rm IG}(a_i,b_i\gamma)$,
    \item STK2-GVF1 where $\sigma_i^2\sim{\rm IG}(a_i,b_i\gamma\exp(\eta_1\log(n_i))$, 
    \item STK2-GVF2 where $\sigma_i^2\sim{\rm IG}(a_i,b_i\gamma\exp(\eta_1\log(n_i)+\eta_2\log(z_i))$,
    \item GL-BP-GVF1 where $\sigma_i^2\sim{\rm IG}(\alpha\omega_i+1,\alpha\omega_i\exp(\eta_0+\eta_1\log(n_i)))$,
    \item GL-BP-GVF2 where $\sigma_i^2\sim{\rm IG}(\alpha\omega_i+1,\alpha\omega_i\exp(\eta_0+\eta_1\log(n_i)+\eta_2\log(z_i)))$.
\end{enumerate}

Since $n_i>1$ and both covariates, $\log(n_i)$ and $\log(z_i)$, are positive for all municipalities, the conditions for posterior propriety of the GL-BP models in Theorem \ref{thm:proper_posterior} are satisfied. 
We assume $\pi(\boldsymbol{\beta})\propto1$ and $\sigma_u^2\sim{\rm IG}(0.1,0.1)$ for all six models. Table \ref{tab:PEAI-AHS} presents the posterior estimates of the model parameters and DIC values. 
According to the DICs, GL-BP-GVF1 is the most suitable model for the PEAI-AHS data. 

\begin{table}[ht!]
    \centering
    \resizebox{\textwidth}{!}{
    \begin{tabular}{l|ccccccccc}
    \hline
     & $\beta_0$ & $\beta_1$ & $\sigma_u^2$ & $\gamma$ & $\eta_0$ & $\eta_1$ & $\eta_2$ & $\alpha$ & DIC\\
    \hline 
      YC & 0.587 & -0.481 & 0.008 & - & - & - &  - & - & -4211.40 \\
      & (0.023) & (0.037) & (0.001) & - & - & - & - &  - &  \\
      STK1 & 0.586 & -0.474 & 0.007 & 0.451 & - & - & - & - & -4259.12 \\
      & (0.023) & (0.036) & (0.001) & (0.02) & - & - & - & - & \\
      STK2-GVF1 & 0.586 & -0.478 & 0.008 & 0.278 & - & 0.116 & - & - &  -4259.39 \\
      & (0.023) & (0.037) & (0.001) & (0.06) & - & (0.047) & - & - & \\
      STK2-GVF2 & 0.586 & -0.475 & 0.007 & 0.215 & - & -0.022 & 0.154 & - &  -4258.75 \\
      & (0.023) & (0.037) & (0.001) & (0.263) & - & (0.070) & (0.103) & - & \\
      GL-BP-GVF1 & 0.586 & -0.472 & 0.007 & - & -1.460 & -0.966 & - & 3.197 & \textbf{-4265.64} \\
      & (0.023) & (0.036) & (0.001)& - & (0.018) & (0.004) &- & (0.584) & \\
      GL-BP-GVF2 & 0.587 & -0.473 & 0.007 & - & -0.664 & -1.023 & -0.052 & 1.535 & -4226.69 \\
      & (0.022) & (0.036) & (0.001)& - & (0.011) & (0.012) & (0.007) & (0.277) & \\
      \hline
    \end{tabular}}
    \caption{Posterior mean estimates and standard errors (parenthesis) of model parameters and Deviance Information Criterion (DIC) in the PEAI--AHS application.}
    \label{tab:PEAI-AHS}
\end{table}

We further analyze the results from GL-BP-GVF1. 
In Figure \ref{fig:results_PEAI-AHS_gvf1} (left), we compare the DVEs with the posterior estimates of $\sigma_i^2$ obtained with YC, STK1, STK2-GVF1, and GL-BP-GVF1. 
Figure \ref{fig:results_PEAI-AHS_gvf1} (left) illustrates that the proposed GL-BP model imposes considerable shrinkage of the posterior mean of $\sigma_i^2$ toward  the posterior estimates of the Exp-GVF1 across most municipalities.
As in the Corn data set,  YC does not produce shrinkage estimates of $\sigma_i^2$ whereas the STK1 and STK2-GVF1 models yield similar posterior estimates of $\sigma_i^2$.



To illustrate the behavior of the posterior shrinkage factor $\kappa_i^{\rm post}$ in (\ref{post_shrinkage}), we consider the small area absolute scaled discrepancies 
 given by 
\begin{align}\label{eq:abs_scaled_disc}
    \frac{|\log D_i-\widehat{\boldsymbol{z}_i^\top\boldsymbol{\eta}}|}{\sqrt{\frac{1}{m-2}\sum_{i=1}^m(\log(D_i)-\widehat{\boldsymbol{z}_i^\top\boldsymbol{\eta}})^2}},
\end{align} 

where $\widehat{\boldsymbol{z}_i^\top\boldsymbol{\eta}}$ is the logarithm of the posterior estimate of the Exp-GVF1, $\eta_0+\eta_1\log(n_i)$. 
As discussed in Section \ref{subsec:shrinkage_factor}, the shrinkage factors $\kappa_i^{\rm post}$ are adaptive according to the quality of the posterior estimates of the Exp-GVF and the sample sizes. Specifically, we expect the posterior mean of the posterior shrinkage factor to be closer to 1 (or 0) when the corresponding absolute scaled discrepancy becomes smaller (or larger).
Figure \ref{fig:results_PEAI-AHS_gvf1} (right) illustrates that, for municipalities in the same sample size interval, the posterior mean estimate of the shrinkage factor approaches 0 when the absolute scaled discrepancy increases.  Figure \ref{fig:results_PEAI-AHS_gvf1} (right) also shows that  when $n_i<30$ the posterior mean estimate of the shrinkage factor is close to 1 making 
the posterior estimate of $\sigma_i^2$ of GL-BP-GVF1  close to the posterior estimate of the Exp-GVF1. As expected, when the sample size increases, the posterior mean estimate of $\kappa_i^{\rm post}$ becomes closer to 0. 

\begin{figure}[ht!]
\centering
    \begin{subfigure}[t]{0.5\textwidth}
         \centering
         \includegraphics[width=\textwidth]{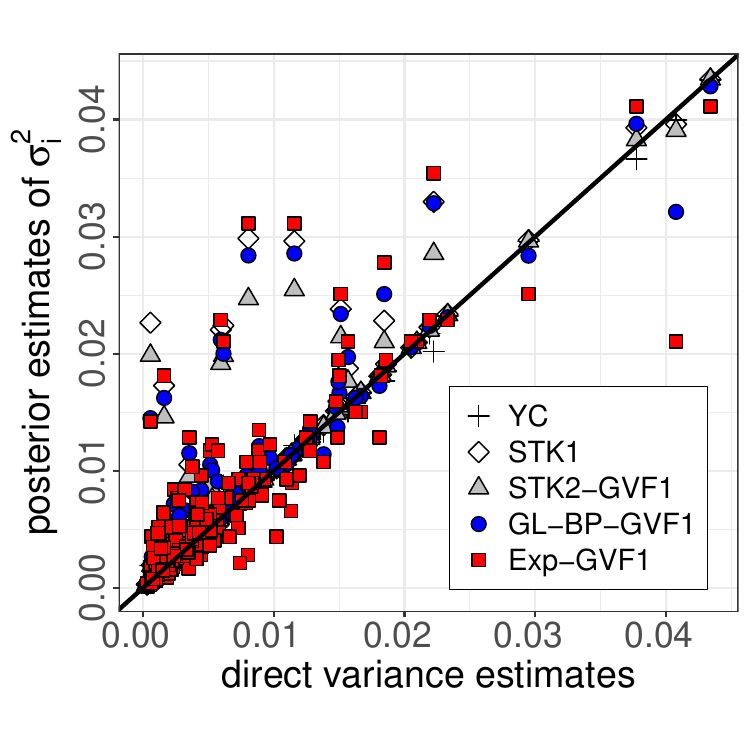}
     \end{subfigure}~
     \begin{subfigure}[t]{0.5\textwidth}
         \centering
         \includegraphics[width=\textwidth]{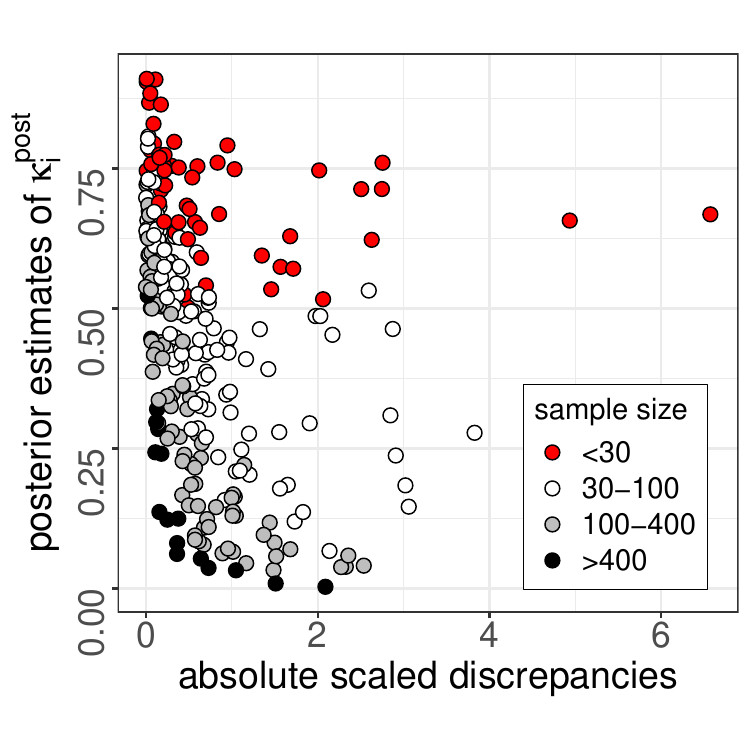}
     \end{subfigure}
     \vspace{-1cm}
     \caption{Direct Variance Estimates (DVEs) versus posterior estimates of $\sigma_i^2$ in the PEAI--AHS application (left). The black solid line represents the DVEs. The red dots (labeled as Exp-GVF1) represent the posterior estimates of the Exp-GVF1, $\exp(\eta_0+\eta_1 \log (n_i))$, under the proposed GL-BP-GVF1 model. The blue dots represent the posterior estimates of $\sigma_i^2$ with the GL-BP-GVF1 model. Posterior mean estimates of the shrinkage factors in (\ref{post_shrinkage}) obtained with the GL-BP-GVF1 model versus absolute scaled discrepancies given by (\ref{eq:abs_scaled_disc}), according to the sample sizes (right). 
}
        \label{fig:results_PEAI-AHS_gvf1}
\end{figure}


Figure \ref{fig:results2_PEAI-AHS_gvf1} (left) displays the direct estimates of the PEAI–AHS ($y_i$) and the posterior estimates of $\theta_i$ obtained  with YC, STK1, STK2-GVF1, and GL-BP-GVF1. 
According to Figure \ref{fig:results2_PEAI-AHS_gvf1} (left) the four Bayesian models produce small area posterior estimates of $\theta_i$ closely related to the direct estimates of the PEAI–AHS.  
The ratios of the coefficients of variations (CVs) obtained  with the GL-BP-GVF1 model and the direct estimates are displayed in Figure \ref{fig:results2_PEAI-AHS_gvf1} (right).  The result illustrates that the proposed GL-BP-GVF1 model substantially improves the precision of the small area estimates by producing smaller CVs in around 95.5\% of the municipalities.

\begin{figure}[ht!]
\centering
    \begin{subfigure}[t]{0.5\textwidth}
    \vspace{-0.4cm}
         \centering
         \includegraphics[width=\textwidth]{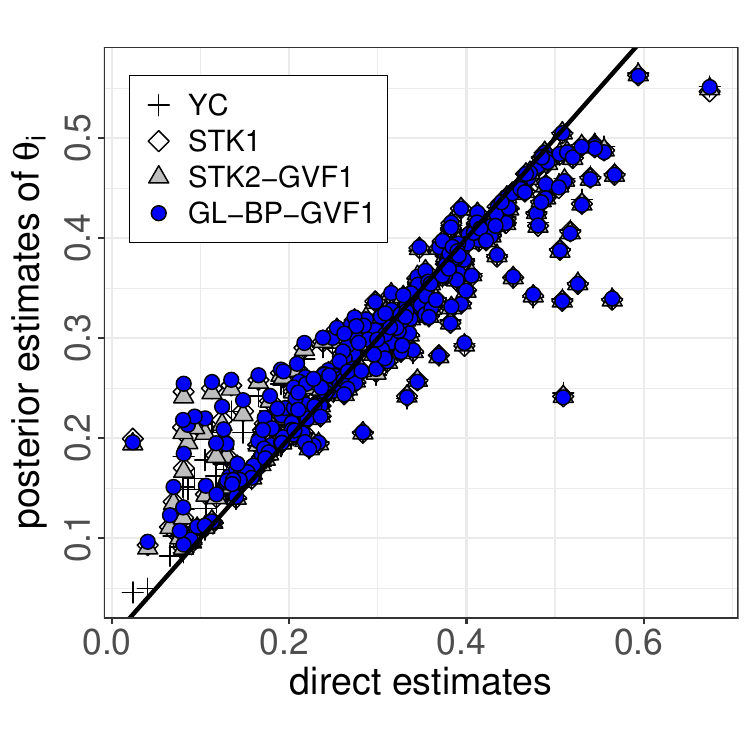}
     \end{subfigure}~
     \begin{subfigure}[t]{0.5\textwidth}
    \vspace{-0.4cm}
         \centering
         \includegraphics[width=\textwidth]{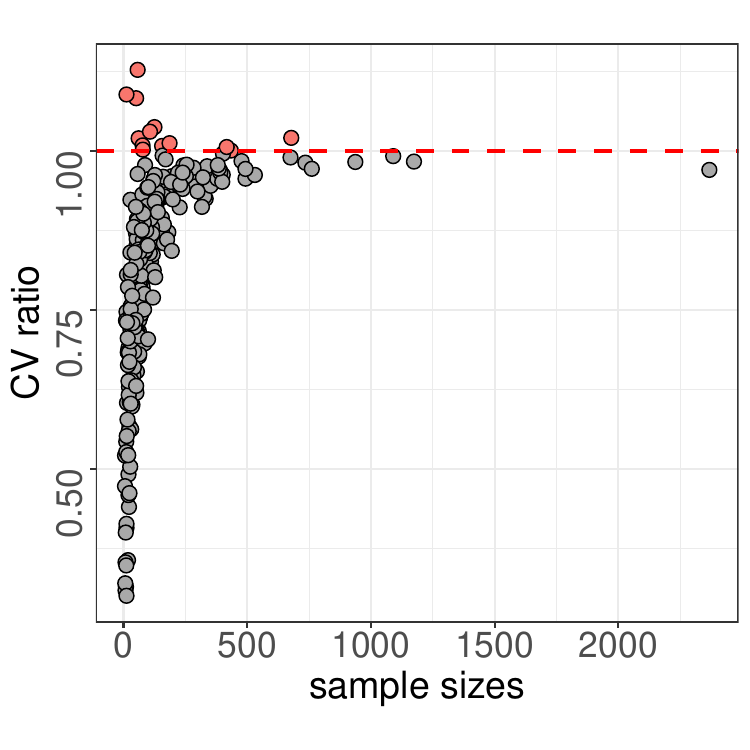}
     \end{subfigure}
     \vspace{-0.5cm}
     \caption{Direct estimates of the PEAI–AHS ($y_i$) versus posterior estimates of $\theta_i$ in the PEAI--AHS application (left). The  solid black line represents $y_i$.
     Ratios of Coefficients of Variations (CVs) obtained with the GL-BP-GVF1 model and the direct estimator (right). }
        \label{fig:results2_PEAI-AHS_gvf1}
\end{figure}


\section{Conclusions}\label{sec:conclusions}

In this work, we have proposed a new Bayesian model that combines the use of GL priors and the GVF method to improve the posterior estimation of the sampling variances in SAE. Our proposal has methodological, computational, and applied contributions.

From a methodological perspective, we proposed a new Bayesian SAE model that incorporates GL priors and the GVF method when the sampling variances in the FH model (\ref{FH_model}) are assumed to be unknown. We established important theoretical properties of the local priors and the resulting posterior shrinkage factors which lead to improved estimation of the sampling variances $\sigma_i^2$.

In computational terms, due to the formulation of the proposed models, we developed suitable computational strategies to obtain posterior samples when the full conditional distributions involve Gamma functions. Our simulation studies illustrate that the proposed model improves the estimation of both small area means and variances.

From the applied side, we considered a real-world application  that involves the estimation of the Prevalence of the Educational Attainment Index at the municipality levels in Colombia.  The results show that our proposed model improves the precision of the small area estimates and allows the level of shrinkage to be adjusted according to the sample size and the quality of the posterior estimate of the Exp-GVF. Our approach can be particularly useful for estimating relevant indicators in official statistics, especially when sample sizes are small or the DVEs are not reliable.

\clearpage

\bibliographystyle{apalike}
\bibliography{reference}

\begin{thebibliography}{}

\bibitem[Battese et~al., 1988]{1988_Battese}
Battese, G.~E., Harter, R.~M., and Fuller, W.~A. (1988).
\newblock An {E}rror-{C}omponents {M}odel for {P}rediction of {C}ounty {C}rop
  {A}reas {U}sing {S}urvey and {S}atellite {D}ata.
\newblock {\em Journal of the American Statistical Association},
  83(401):28--36.

\bibitem[Borchers and da~Cunha, 2025]{BorchersCunha2025}
Borchers, J. and da~Cunha, M.~S. (2025).
\newblock School performance and inequality of opportunities in {L}atin
  {A}merica.
\newblock {\em Studies in Educational Evaluation}, 86.

\bibitem[Boss et~al., 2024]{Bayesian_group}
Boss, J., Datta, J., Wang, X., Park, S.~K., Kang, J., and Mukherjee, B. (2024).
\newblock {Group Inverse-Gamma Gamma Shrinkage for Sparse Linear Models with
  Block-Correlated Regressors}.
\newblock {\em Bayesian Analysis}, 19(3):785 -- 814.

\bibitem[Brown and Griffin, 2010]{2010_Griffin&Brown}
Brown, P.~J. and Griffin, J.~E. (2010).
\newblock {Inference with normal-gamma prior distributions in regression
  problems}.
\newblock {\em Bayesian Analysis}, 5(1):171 -- 188.

\bibitem[Carvalho et~al., 2010]{2010_Carvalho}
Carvalho, C.~M., Polson, N.~G., and Scott, J.~G. (2010).
\newblock The horseshoe estimator for sparse signals.
\newblock {\em Biometrika}, 97(2):465--480.

\bibitem[Cox et~al., 2010]{cox2010educational}
Cox, C., Attewell, P., and Newman, K. (2010).
\newblock {E}ducational {I}nequality in {L}atin {A}merica.
\newblock {\em Growing gaps. Educational inequality around the world}, pages
  33--58.

\bibitem[Dass et~al., 2012]{2012_Dass}
Dass, S.~C., Maiti, T., Ren, H., and Sinha, S. (2012).
\newblock Confidence interval estimation of small area parameters shrinking
  both means and variances.
\newblock {\em Survey Methodology}, 38(2):173--187.

\bibitem[DHS, 2015]{Profamilia_COL_DHS_2005}
DHS (2015).
\newblock Encuesta nacional de demograf{\'i}a y salud 2015 (ends 2015).

\bibitem[FAO-OPHI, 2022]{FAO_OPHI_2022_RMPI}
FAO-OPHI (2022).
\newblock Measuring rural poverty with a multidimensional approach: {T}he
  {R}ural {M}ultidimensional {P}overty {I}ndex.
\newblock Technical Report Statistical Development Series No. 19, Food and
  Agriculture Organization of the United Nations (FAO), Rome.

\bibitem[Fay and Herriot, 1979]{1979_Fay&Herriot}
Fay, R. and Herriot, R. (1979).
\newblock {E}stimates of {I}ncome for {S}mall {P}laces: {A}n {A}pplication of
  {J}ames-{S}tein {P}rocedures to {C}ensus {D}ata.
\newblock {\em Journal of the American Statistical Association},
  74(366):269--277.

\bibitem[H{\'a}jek, 1971]{Hajek1971}
H{\'a}jek, J. (1971).
\newblock Comment on “{A}n {E}ssay on the {L}ogical {F}oundations of {S}urvey
  {S}ampling, {P}art {O}ne”.
\newblock In Godambe, V.~P. and Sprott, D.~A., editors, {\em The Foundations of
  Survey Sampling}, page 236. Holt, Rinehart, and Winston, New York.

\bibitem[Hamura et~al., 2024]{2024_Hamura}
Hamura, Y., Onizuka, T., Hashimoto, S., and Sugasawa, S. (2024).
\newblock {S}parse {B}ayesian {I}nference on {G}amma-{D}istributed
  {O}bservations {U}sing {S}hape-{S}cale {I}nverse-{G}amma {M}ixtures.
\newblock {\em Bayesian Analysis}, 19(1):77--97.

\bibitem[Kubacki and Jedrzejczak, 2012]{2012_Kubacki&Jedrzej}
Kubacki, J. and Jedrzejczak, A. (2012).
\newblock The {C}omparison of {G}eneralized {V}ariance {F}unction with {O}ther
  {M}ethods of {P}recision {E}stimation for {P}olish {H}ousehold {B}udget
  {S}urvey.
\newblock {\em Studia Ekonomiczne}, 120:58--69.

\bibitem[Maiti et~al., 2014]{2014_Maiti}
Maiti, T., Ren, H., and Sinha, S. (2014).
\newblock Prediction {E}rror of {S}mall {A}rea {P}redictors {S}hrinking {B}oth
  {M}eans and {V}ariances.
\newblock {\em Scandinavian Journal of Statistics}, 41(3):775--790.

\bibitem[Maples et~al., 2009]{2009_Maples}
Maples, J., Bell, W., and Huang, E.~T. (2009).
\newblock Small {A}rea {V}ariance {M}odeling with {A}pplication to {C}ounty
  {P}overty {E}stimates from the {A}merican {C}ommunity {S}urvey.
\newblock {\em Proceedings of the American Statistical Association Section on
  Survey Research Methods}, pages 5056--5067.

\bibitem[McIllece, 2018]{2018_McIllece}
McIllece, J.~J. (2018).
\newblock On {G}eneralized {V}ariance {F}unctions for {S}ample {M}eans and
  {M}edians.
\newblock {\em JSM2018-survey research methods section}, pages 584--594.

\bibitem[Miller, 2019]{2019_Miller}
Miller, J.~W. (2019).
\newblock Fast and {A}ccurate {A}pproximation of the {F}ull {C}onditional for
  {G}amma {S}hape {P}arameters.
\newblock {\em Journal of Computational and Graphical Statistics},
  28(2):476--480.

\bibitem[Morales et~al., 2021]{2021_Morales}
Morales, D., Esteban, M.~D., Pérez, A., and Hobza, T. (2021).
\newblock {\em A Course on Small Area Estimation and Mixed Models: Methods,
  Theory and Applications in R}.
\newblock Statistics for Social and Behavioral Sciences, Springer.

\bibitem[Park and Casella, 2008]{2008_Park&Casella}
Park, T. and Casella, G. (2008).
\newblock The {B}ayesian {L}asso.
\newblock {\em Journal of the american statistical association},
  103(482):681--686.

\bibitem[P{\'e}rez et~al., 2017]{2017_Perez}
P{\'e}rez, M.-E., Pericchi, L.~R., and Ram{\'i}rez, I.~C. (2017).
\newblock The {S}caled {B}eta2 {D}istribution as a {R}obust {P}rior for
  {S}cales.
\newblock {\em Bayesian Analysis}, 12(3):615--637.

\bibitem[Rangel and Lleras, 2010]{RangelLleras2010}
Rangel, C. and Lleras, C. (2010).
\newblock Educational inequality in {C}olombia: family background, school
  quality and student achievement in {C}artagena.
\newblock {\em International Studies in Sociology of Education},
  20(4):291--317.

\bibitem[Rivest and Vandal, 2002]{2002_Rivest&Vandal}
Rivest, L.~P. and Vandal, N. (2002).
\newblock Mean squared error estimation for small areas when the small area
  variances are estimated.
\newblock {\em Proceedings of the International Conference on Recent Advances
  in Survey Sampling}, pages 10--13.

\bibitem[Souza et~al., 2009]{2009_Souza}
Souza, D.~F., Moura, F. A.~S., and Migon, H.~S. (2009).
\newblock Small area population prediction via hierarchical models.
\newblock {\em Survey Methodology}, 35:203 –214.

\bibitem[Spiegelhalter et~al., 2002]{2002_Spiegelhalter}
Spiegelhalter, D.~J., Best, N.~G., Carlin, B.~P., and Van Der~Linde, A. (2002).
\newblock Bayesian {M}easures of {M}odel {C}omplexity and {F}it.
\newblock {\em Journal of the Royal Statistical Society Series B: Statistical
  Methodology}, 64(4):583--639.

\bibitem[Sugasawa et~al., 2017]{2017_Sukasawa}
Sugasawa, S., Tamae, H., and Kubokawa, T. (2017).
\newblock Bayesian {E}stimators for {S}mall {A}rea {M}odels {S}hrinking {B}oth
  {M}eans and {V}ariances.
\newblock {\em Scandinavian Journal of Statistics}, 44(1):150--167.

\bibitem[Tang et~al., 2018]{2018_TangGhosh}
Tang, X., Ghosh, M., Ha, N.~S., and Sedransk, J. (2018).
\newblock Modeling {R}andom {E}ffects {U}sing {G}lobal-{L}ocal {S}hrinkage
  {P}riors in {S}mall {A}rea {E}stimation.
\newblock {\em Journal of the American Statistical Association},
  113(524):1476--1489.

\bibitem[UN, 2025]{UN_SDG_ExtendedReport_2025}
UN (2025).
\newblock The {S}ustainable {D}evelopment {G}oals {E}xtended {R}eport 2025:
  {I}nputs and information provided as of 30 april 2025.
\newblock \url{https://unstats.un.org/sdgs/report/2025/extended-report/}.
\newblock Accessed: 2025‑12‑07.

\bibitem[UNESCO, 2005]{UNESCO}
UNESCO (2005).
\newblock {\em World Education Indicators 2005 Education Trends in Perspective:
  Education Trends in Perspective}.
\newblock OECD publishing.

\bibitem[{United Nations}, 2015]{UN2015Agenda2030}
{United Nations} (2015).
\newblock Transforming our world: the 2030 {A}genda for {S}ustainable
  {D}evelopment. {R}esolution adopted by the {G}eneral {A}ssembly on 25
  {S}eptember 2015.
\newblock
  \url{https://sustainabledevelopment.un.org/post2015/transformingourworld}.
\newblock Accessed: 2025‑12‑07.

\bibitem[Wang and Fuller, 2003]{2003_Wang&Fuller}
Wang, J. and Fuller, W.~A. (2003).
\newblock The {M}ean {S}quared {E}rror of {S}mall {A}rea {P}redictors
  {C}onstructed with {E}stimated {A}rea {V}ariances.
\newblock {\em Journal of the American Statistical Association},
  98(1):716--723.

\bibitem[Wolter, 2007]{2007_Wolter}
Wolter, K.~M. (2007).
\newblock {\em Introduction to Variance Estimation: Statistics for Social
  Science and Behavorial Sciences}.
\newblock Springer New York, New York, NY.

\bibitem[World-Bank, 2024]{World-Bank}
World-Bank (2024).
\newblock Education {S}tatistics ({E}d{S}tats): {UIS} {D}ata.
\newblock
  \url{https://databank.worldbank.org/source/education-statistics-(uis)}.
\newblock World Bank EdStats database containing education indicators sourced
  from the UNESCO Institute for Statistics (UIS).

\bibitem[You, 2021]{2021_You}
You, Y. (2021).
\newblock Small area estimation using {F}ay-{H}erriot area level model with
  sampling variance smoothing and modeling.
\newblock {\em Survey Methodology}, 47(2):361--370.

\bibitem[You and Chapman, 2006]{2006_You&Chapman}
You, Y. and Chapman, B. (2006).
\newblock Small {A}rea {E}stimation {U}sing {A}rea {L}evel {M}odels and
  {E}stimated {S}ampling {V}ariances.
\newblock {\em Survey Methodology}, 32:97--103.

\bibitem[You and Hidiroglou, 2023]{2023_You&Hidiroglou}
You, Y. and Hidiroglou, M.~A. (2023).
\newblock Application of {S}ampling {V}ariance {S}moothing {M}ethods for
  {S}mall {A}rea {P}roportion {E}stimation.
\newblock {\em Journal of Official Statistics}, 39(4):571--590.

\bibitem[Zhang et~al., 2019]{2019_Zhang}
Zhang, G., Cheng, Y., and Lu, Y. (2019).
\newblock Generalised variance functions for longitudinal survey data.
\newblock {\em Statistical Theory and Related Fields}, 3(2):150--157.

\end{thebibliography}

\clearpage

\appendix

\section{Supplementary Material}

This section contains the supplementary material of our manuscript. Section \ref{proof_Corollary_2_1}, Section 
\ref{proof_Proposition_2_1}, Section \ref{proof_Proposition_2_2_theorem_2_1}, and Section \ref{proof_Theorem_3_1}
contain
the proofs of Corollary \ref{cor}, Proposition \ref{proposition}, Proposition \ref{proposition_post} and Theorem \ref{thm:prob_kappa}, and
Theorem \ref{thm:proper_posterior}, respectively. In Section \ref{MCMC_schemes} the MCMC algorithms for the existing models
in Table \ref{tab:model_structures} are presented in detail. The procedure to estimate the degrees of freedom in the PEAI–AHS application is discussed in Section \ref{application_sup}.

\section{Proofs}
\label{sec:proofs}

For the proofs in Sections \ref{proof_Corollary_2_1} and 
\ref{proof_Proposition_2_1}, we use the notation $f(x)\sim g(x)$ as $x\to0$ to denote $\lim_{x\to0} f(x)/g(x)=1$ and use the following facts:

\begin{itemize}
    \item The bounds for a Gamma function 
    \begin{align}\label{bounds_gamma_func}
    \sqrt{2\pi}x^{x-1/2}e^{-x}\leq\Gamma(x)\leq\sqrt{2\pi}x^{x-1/2}e^{-x}e^{1/(12x)},\quad x>0.
    \end{align}
    
    \item $\pi^{\rm BP}(\omega_i)\propto\omega_i^{a-1}(1+\omega_i)^{-(a+b)}\sim \omega_i^{a-1}$ as $\omega_i\to0$.

    \item $\pi^{\rm Ga}(\omega_i)\propto\omega_i^{c-1}e^{-d\omega_i}\sim \omega_i^{c-1}$ as $\omega_i\to0$.
\end{itemize}

\subsection{Proof of Corollary \ref{cor}}
\label{proof_Corollary_2_1}

\begin{proof} 
    We denote $\tilde\eta_i=\exp(\boldsymbol{z}_i^\top\boldsymbol{\eta})$. Let $\xi_i=\alpha\tilde\eta_i/\sigma_i^2-\alpha-\alpha\log(\tilde\eta_i/\sigma_i^2)$. Importantly, we can see that

    \begin{align*}
        \xi_i\sim\begin{cases}
            \alpha\tilde\eta_i/\sigma_i^2\to\infty &{\rm as}~\sigma_i^2\to0,\\
            -\alpha\log(\tilde\eta_i/\sigma_i^2)\propto\alpha\log(\sigma_i^2)\to\infty&{\rm as}~\sigma_i^2\to\infty,\\
            \frac{\alpha}{2\tilde\eta_i^2}(\sigma_i^2-\tilde\eta_i)^2\to0&{\rm as}~\sigma_i^2\to\tilde\eta_i.
        \end{cases}
    \end{align*}

    Next, we define a function
    \begin{align}\label{define_func_f}
        f(\omega_i,\alpha)=\frac{(\alpha\omega_i)^{\alpha\omega_i-1}e^{-\alpha\omega_i}}{\Gamma(\alpha\omega_i)},\quad \omega_i,\alpha\in(0,\infty).
    \end{align}

    Then the marginal prior distribution of $\sigma_i^2$ under the BP prior can be expressed using the function $f(\cdot)$ in (\ref{define_func_f}) as follows:

    \begin{align}\label{16}
        p^{\rm BP}(\sigma_i^2\mid\alpha,\boldsymbol{z}_i^\top\boldsymbol{\eta})&=\int_0^\infty\pi(\sigma_i^2\mid\omega_i,\alpha,\boldsymbol{z}_i^\top\boldsymbol{\eta})\pi^{\rm BP}(\omega_i)~d\omega_i\nonumber\\
        &=\int_0^\infty \frac{(\alpha\omega_i\tilde\eta_i)^{\alpha\omega_i+1}}{\Gamma(\alpha\omega_i+1)}\left(\frac{1}{\sigma_i^2}\right)^{\alpha\omega_i+2}\exp\left\{-\frac{\alpha\omega_i\tilde\eta_i}{\sigma_i^2}\right\}~\pi^{\rm BP}(\omega_i)~d\omega_i\nonumber\\
        &=\tilde\eta_i\frac{1}{(\sigma_i^2)^2}\int_0^\infty\frac{(\alpha\omega_i)^{\alpha\omega_i}}{\Gamma(\alpha\omega_i)}\left(\frac{\tilde\eta_i}{\sigma_i^2}\right)^{\alpha\omega_i}\exp\left\{-\frac{\alpha\omega_i\tilde\eta_i}{\sigma_i^2}\right\}~\pi^{\rm BP}(\omega_i)~d\omega_i\nonumber\\
        &=\alpha\tilde\eta_i\frac{1}{(\sigma_i^2)^2}\int_0^\infty\omega_i e^{-\xi_i\omega_i} f(\omega_i,\alpha)\pi^{\rm BP}(\omega_i)~d\omega_i.
    \end{align}

    The proof of Part (i) of Theorem S1 in \cite{2024_Hamura} shows that, as $1<\xi_i\to\infty$,

    \begin{align}\label{pf_bound_hamura}
        \left|\frac{\int_0^\infty \omega_ie^{-\xi_i\omega_i}f(\omega_i,\alpha)\pi(\omega_i)~d\omega_i}{\int_0^1 \omega_ie^{-\xi_i\omega_i}f(\omega_i,\alpha)\pi(\omega_i)~d\omega_i}-1\right|\to0,
    \end{align}

    for any proper prior $\pi(\omega_i)$. Following (\ref{pf_bound_hamura}), we obtain that, as $\xi_i\to\infty$,
    
    \begin{align*}
        \frac{(\sigma_i^2)^2p^{\rm BP}(\sigma_i^2\mid\alpha,\boldsymbol{z}_i^\top\boldsymbol{\eta})}{\pi^{\rm BP}(1/\xi_i)/\xi_i^2}&=\frac{\alpha\tilde\eta_i}{\pi^{\rm BP}(1/\xi_i)/\xi_i^2}\int_0^\infty \omega_i e^{-\xi_i\omega_i} f(\omega_i,\alpha)~\pi^{\rm BP}(\omega_i)~d\omega_i\\
        &\sim\frac{\alpha\tilde\eta_i}{\pi^{\rm BP}(1/\xi_i)/\xi_i^2}\int_0^1 \omega_i e^{-\xi_i\omega_i} f(\omega_i,\alpha)~\pi^{\rm BP}(\omega_i)~d\omega_i\\
        &=\frac{\alpha\tilde\eta_i}{\pi^{\rm BP}(1/\xi_i)/\xi_i^2}\int_0^{\xi_i}\frac{1}{\xi_i} \frac{\omega_i}{\xi_i} e^{-\omega_i} f(\omega_i/\xi_i,\alpha)~\pi^{\rm BP}(\omega_i/\xi_i)~d\omega_i\\
        &=\alpha\tilde\eta_i\int_0^\infty\omega_ie^{-\omega_i}\times\frac{1(\omega_i<\xi_i)\pi^{\rm BP}(\omega_i/\xi_i)}{\pi^{\rm BP}(1/\xi_i)}\times 1(\omega_i<\xi_i)f(\omega_i/\xi_i,\alpha)~d\omega_i.
    \end{align*}

    For all $\omega_i>0$ and $\xi_i>1$, we have that
    
    \begin{align*}
        \frac{1(\omega_i<\xi_i)\pi^{\rm BP}(\omega_i/\xi_i)}{\pi^{\rm BP}(1/\xi_i)}=1(\omega_i<\xi_i)\frac{\omega_i^{a-1}(1+\omega_i/\xi_i)^{-(a+b)}}{(1+1/\xi_i)^{-(a+b)}}\leq\omega_i^{a-1}\frac{{\rm sup}_{u_i\in(0,1)}(1+u_i)^{-(a+b)}}{{\rm inf}_{u_i\in(0,1)}(1+u_i)^{-(a+b)}}=C\omega_i^{a-1}
    \end{align*}

    where $C=2^{a+b}$ and that

    \begin{align*}
         1(\omega_i<\xi_i)f(\omega_i/\xi_i,\alpha)=1(\omega_i<\xi_i)\frac{(\alpha\omega_i/\xi_i)^{\alpha\omega_i/\xi_i}e^{-\alpha\omega_i/\xi_i}}{\Gamma(\alpha\omega_i/\xi_i+1)}\leq\underset{u_i\in(0,\alpha)}{{\rm sup}}\frac{u_i^{u_i}}{\Gamma(u_i+1)}<\infty.
    \end{align*}

    Hence, by applying the Dominated Convergence Theorem (DCT), we obtain that, as $1<\xi_i\to\infty$,

    \begin{align}\label{17}
        \frac{(\sigma_i^2)^2 p^{\rm BP}(\sigma_i^2\mid\alpha,\boldsymbol{z}_i^\top\boldsymbol{\eta})}{\pi^{\rm BP}(1/\xi_i)/\xi_i^2}
        &\sim \alpha\tilde\eta_i\int_0^\infty\omega_ie^{-\omega_i}\times\omega_i^{a-1}\times1~d\omega_i=\alpha\tilde\eta_i\Gamma(a+1).
    \end{align}

    From (\ref{17}), we complete the proof for Part (i). As $\sigma_i^2\to0$, we have $\xi_i\sim\alpha\tilde\eta_i/\sigma_i^2\to\infty$ and 
    
        \begin{align*}
            p^{\rm BP}(\sigma_i^2\mid\alpha,\boldsymbol{z}_i^\top\boldsymbol{\eta})&\sim \alpha\tilde \eta_i\Gamma(a+1)\frac{1}{(\sigma_i^2)^2}\frac{1}{\xi_i^2}\pi^{\rm BP}(1/\xi_i)\\
            &=\alpha\tilde\eta_i\Gamma(a+1)\frac{1}{(\sigma_i^2)^2}\frac{1}{\xi_i^2}\times\frac{\Gamma(a+b)}{\Gamma(a)\Gamma(b)}\frac{(1/\xi_i)^{a-1}}{(1+1/\xi_i)^{a+b}}\\
            &\sim\alpha\tilde\eta_i\frac{a\Gamma(a+b)}{\Gamma(b)}\frac{1}{(\sigma_i^2)^2}\frac{1}{\xi_i^{a+1}} \\
            &\sim \alpha\tilde\eta_i\frac{a\Gamma(a+b)}{\Gamma(b)} \frac{1}{(\sigma_i^2)^2}\frac{1}{(\alpha\tilde\eta_i/\sigma_i^2)^{a+1}}\\
            &=\frac{a\Gamma(a+b)}{(\alpha\tilde\eta_i)^{a}\Gamma(b)}(\sigma_i^2)^{a-1}\to\begin{cases}
        0 & \text{if}~a>1,\\
        b/(\alpha\tilde\eta_i)<\infty & \text{if}~a=1,\\
        \infty&\text{if}~a<1.
    \end{cases}
        \end{align*}

        Next, we use the expression (\ref{17}) to prove Part (ii). As $\sigma_i^2\to\infty$, we have $\xi_i\sim-\alpha\log(\tilde\eta_i/\sigma_i^2)\propto\alpha\log(\sigma_i^2)
        \to\infty$ and    
        
        \begin{align*}
        p^{\rm BP}(\sigma_i^2\mid\alpha,\boldsymbol{z}_i^\top\boldsymbol{\eta})&\sim\alpha\tilde\eta_i\Gamma(a+1)\frac{1}{(\sigma_i^2)^2}\frac{1}{\xi_i^2}\pi^{\rm BP}(1/\xi_i)\\
        &\sim\alpha\tilde\eta_i\frac{a\Gamma(a+b)}{\Gamma(b)}\frac{1}{(\sigma_i^2)^2}\frac{1}{\xi_i^{a+1}}\sim \alpha\tilde\eta_i\frac{a\Gamma(a+b)}{\Gamma(b)}\frac{1}{(\sigma_i^2)^2}\left(\frac{1}{\log(\sigma_i^2/\tilde\eta_i)^\alpha}\right)^{a+1}\to 0.
        \end{align*}
        
    Finally, we prove Part (iii). We start with the case of $b<1/2$. From (\ref{16}), as $\sigma_i^2\to\tilde\eta_i$, we obtain that

    \begin{align}\label{18}
        p^{\rm BP}&(\sigma_i^2\mid\alpha,\boldsymbol{z}_i^\top\boldsymbol{\eta})\nonumber\\
        &\sim\alpha\tilde\eta_i\frac{1}{\tilde\eta_i^2}\int_0^\infty\omega_ie^{-\xi_i\omega_i}f(\omega_i,\alpha)\pi^{\rm BP}(\omega_i)~d\omega_i\nonumber\\
        &=\frac{\alpha}{\tilde\eta_i}\left[\int_0^1\omega_ie^{-\xi_i\omega_i}f(\omega_i,\alpha)\pi^{\rm BP}(\omega_i)~d\omega_i+\int_1^\infty\omega_ie^{-\xi_i\omega_i}f(\omega_i,\alpha)\pi^{\rm BP}(\omega_i)~d\omega_i\right]\nonumber\\
        &=\frac{\alpha}{\tilde\eta_i}\left[\int_0^1\omega_ie^{-\xi_i\omega_i}f(\omega_i,\alpha)\pi^{\rm BP}(\omega_i)~d\omega_i+\int_{\xi_i}^\infty\frac{1}{\xi_i}\frac{\omega_i}{\xi_i}e^{-\omega_i}f(\omega_i/\xi_i,\alpha)\pi^{\rm BP}(\omega_i/\xi_i)~d\omega_i\right].
    \end{align}

    We follow the argument used in the proof of Part (ii) of Theorem S1 in \cite{2024_Hamura} to evaluate the first and second terms in (\ref{18}). For the first term, using the Monotone Convergence Theorem (MCT), as $\xi_i\to0$, we have

    \begin{align*}
    \int_0^1\omega_i e^{-\xi_i\omega_i}f(\omega_i,\alpha)\pi^{\rm BP}(\omega_i)~d\omega_i \to  \int_0^1\omega_i  f(\omega_i,\alpha)\pi^{\rm BP}(\omega_i)~d\omega_i<\infty.
    \end{align*}

    For the second term, after rearranging the expression, we apply the DCT as $\xi_i\to0$ to obtain

    \begin{align*}
        \int_{\xi_i}^\infty&\frac{1}{\xi_i}\frac{\omega_i}{\xi_i}e^{-\omega_i}f(\omega_i/\xi_i,\alpha)\pi^{\rm BP}(\omega_i/\xi_i)~d\omega_i\\
        &\sim\Gamma(1/2-b)\left[\frac{1}{\xi_i^2}f(1/\xi_i,\alpha)\pi^{\rm BP}(1/\xi_i)\right]\\
        &\sim\Gamma(1/2-b)\frac{1}{\xi_i^2}f(1/\xi_i,\alpha)\xi_i^{b+1}\\
        &=\frac{\Gamma(1/2-b)}{\xi_i^{1-b}}\frac{(\alpha/\xi_i)^{\alpha/\xi_i-1}e^{-\alpha/\xi_i}}{\Gamma(\alpha/\xi_i)}=\frac{\Gamma(1/2-b)}{\xi_i^{1/2-b}}\frac{(\alpha/\xi_i)^{\alpha/\xi_i-1/2}e^{-\alpha/\xi_i}}{\sqrt{\alpha}\Gamma(\alpha/\xi_i)}\propto\frac{1}{\xi_i^{1/2-b}}\to\infty.
    \end{align*}

    Since the second term in (\ref{18}) is infinite as $\xi_i\to0$ when $b<1/2$, we complete the proof that, as $\sigma_i^2\to\tilde\eta_i$, we have $\xi_i\sim\frac{\alpha}{2\tilde\eta_i^2}(\sigma_i^2-\tilde\eta_i)^2\to0$ and

    \begin{align*}
    p^{\rm BP}(\sigma_i^2\mid\alpha,\boldsymbol{z}_i^\top\boldsymbol{\eta})\to\infty  \quad{\rm if}~ b<1/2. 
    \end{align*}

    Next, suppose that $b=1/2$. We use the expression (\ref{16}). By the MCT and the upper bound in (\ref{bounds_gamma_func}), as $\sigma_i^2\to\tilde\eta_i$, we have $\xi_i\to0$ and

    \begin{align}\label{pf_eq_BP}
        p^{\rm BP}(\sigma_i^2\mid\alpha,\boldsymbol{z}_i^\top\boldsymbol{\eta})&\sim\frac{\alpha}{\tilde\eta_i}\int_0^\infty\omega_ie^{-\xi_i\omega_i} f(\omega_i,\alpha)\pi^{\rm BP}(\omega_i)~d\omega_i\nonumber\\
        &\to\frac{\alpha}{\tilde\eta_i}\int_0^\infty\omega_if(\omega_i,\alpha)\pi^{\rm BP}(\omega_i)~d\omega_i\nonumber\\
        &=\frac{\alpha}{\tilde\eta_i}\int_0^\infty\omega_i\frac{(\alpha\omega_i)^{\alpha\omega_i-1}e^{-\alpha\omega_i}}{\Gamma(\alpha\omega_i)}\pi^{\rm BP}(\omega_i)~d\omega_i\nonumber\\
        &=\frac{\alpha}{\tilde\eta_i}\int_0^\infty\sqrt{\omega_i}\frac{(\alpha\omega_i)^{\alpha\omega_i-1/2}e^{-\alpha\omega_i}}{\sqrt{\alpha}\Gamma(\alpha\omega_i)}\pi^{\rm BP}(\omega_i)~d\omega_i\\
        &\geq\frac{\alpha}{\tilde\eta_i}\frac{1}{\sqrt{2\pi\alpha}}\frac{\Gamma(a+1/2)}{\Gamma(a)\Gamma(1/2)}\int_0^\infty \frac{\omega_i^{a-1/2}}{(1+\omega_i)^{a+1/2}}\exp\left(-\frac{1}{12\alpha\omega_i}\right)~d\omega_i=\infty.\nonumber
    \end{align}

    Finally, suppose that $b>1/2$. Using (\ref{pf_eq_BP}) and the lower bound in (\ref{bounds_gamma_func}), as $\sigma_i^2\to\tilde\eta_i$, we have $\xi_i\to0$ and

    \begin{align*}
        p^{\rm BP}(\sigma_i^2\mid\alpha,\boldsymbol{z}_i^\top\boldsymbol{\eta})&\sim
        \frac{\alpha}{\tilde\eta_i}\int_0^\infty\sqrt{\omega_i}\frac{(\alpha\omega_i)^{\alpha\omega_i-1/2}e^{-\alpha\omega_i}}{\sqrt{\alpha}\Gamma(\alpha\omega_i)}\pi^{\rm BP}(\omega_i)~d\omega_i\\
        &\leq\frac{\alpha}{\tilde\eta_i}\frac{1}{\sqrt{2\pi\alpha}}\frac{\Gamma(a+b)}{\Gamma(a)\Gamma(b)}\int_0^\infty\omega_i^{(1/2+a)-1}(1+\omega_i)^{-(a+b)}~d\omega_i\\
        &=\frac{\alpha}{\tilde\eta_i}\frac{1}{\sqrt{2\pi\alpha}}\frac{\Gamma(a+b)}{\Gamma(a)\Gamma(b)}\frac{\Gamma(a+1/2)\Gamma(b-1/2)}{\Gamma(a+b)}<\infty.
    \end{align*}

\end{proof}

\subsection{Proof of Proposition \ref{proposition}}
\label{proof_Proposition_2_1}

\begin{proof}    
The proof of Proposition \ref{proposition} follows the same structure as that of Corollary \ref{cor} in Section \ref{proof_Corollary_2_1}. The marginal prior distribution of $\sigma_i^2$ under the Gamma prior can be expressed using the function $f(\cdot)$ in (\ref{define_func_f}) as follows:

    \begin{align*}
        p^{\rm Ga}(\sigma_i^2\mid\alpha,\boldsymbol{z}_i^\top\boldsymbol{\eta})&=\int_0^\infty\pi(\sigma_i^2\mid\omega_i,\alpha,\boldsymbol{z}_i^\top\boldsymbol{\eta})\pi^{\rm Ga}(\omega_i)~d\omega_i\\
        &=\alpha\tilde\eta_i \frac{1}{(\sigma_i^2)^2}\int_0^\infty\omega_ie^{-\xi_i\omega_i}f(\omega_i,\alpha)~\pi^{\rm Ga}(\omega_i)~d\omega_i.
    \end{align*}

    By the DCT, as $\xi_i\to\infty$, we obtain a result under the Gamma prior, which is similar to that in (\ref{17}) under the BP prior, as follows:
    
    \begin{align*}
        \frac{(\sigma_i^2)^2~p^{\rm Ga}(\sigma_i^2\mid\alpha,\boldsymbol{z}_i^\top\boldsymbol{\eta})}{\pi^{\rm Ga}(1/\xi_i)/\xi_i^2}&\sim\alpha\tilde\eta_i\int_0^\infty\omega_i e^{-\omega_i}\times\frac{1(\omega_i<\xi_i)\pi^{\rm Ga}(\omega_i/\xi_i)}{\pi^{\rm Ga}(1/\xi_i)}\times 1(\omega_i<\xi_i)f(\omega_i/\xi_i,\alpha)~d\omega_i\\
        &\to\alpha\tilde\eta_i\int_0^\infty\omega_ie^{-\omega_i}\times\omega_i^{c-1}\times1~d\omega_i=\alpha\tilde\eta_i\Gamma(c+1).
    \end{align*}

    Hence, as $\sigma_i^2\to0$, we have $\xi_i\sim\alpha\tilde\eta_i/\sigma_i^2\to\infty$ and 
    
    \begin{align*}
        p^{\rm Ga}(\sigma_i^2\mid\alpha,\boldsymbol{z}_i^\top\boldsymbol{\eta})&\sim\alpha\tilde\eta_i\Gamma(c+1)\frac{1}{(\sigma_i^2)^2}\frac{1}{\xi_i^2}\pi^{\rm Ga}(1/\xi_i)\\
        &=\alpha\tilde\eta_i\Gamma(c+1)\frac{1}{(\sigma_i^2)^2}\frac{1}{\xi_i^2}\times\frac{d^c}{\Gamma(c)}(1/\xi_i)^{c-1}\exp(-d/\xi_i)\\
        &\sim\alpha\tilde\eta_i cd^c\frac{1}{(\sigma_i^2)^2}\frac{1}{\xi_i^{c+1}} \\
        &\sim\alpha\tilde\eta_i cd^c\frac{1}{(\sigma_i^2)^2}\frac{1}{(\alpha\tilde\eta_i/\sigma_i^2)^{c+1}}\\
        &=\frac{cd^c}{(\alpha\tilde\eta_i)^c}(\sigma_i^2)^{c-1}\to\begin{cases}
        0 & \text{if}~c>1,\\
        d/(\alpha\tilde\eta_i)<\infty & \text{if}~c=1,\\
        \infty&\text{if}~c<1,
    \end{cases}
    \end{align*}

    which completes the proof for Part (i). Next, we prove Part (ii). As $\sigma_i^2\to\infty$, we have $\xi_i\sim-\alpha\log(\tilde\eta_i/\sigma_i^2)\propto\alpha\log(\sigma_i^2)\to\infty$ and    
    
    \begin{align*}
        p^{\rm Ga}(\sigma_i^2\mid\alpha,\boldsymbol{z}_i^\top\boldsymbol{\eta})&\sim\alpha\tilde\eta_i\Gamma(c+1)\frac{1}{(\sigma_i^2)^2}\frac{1}{\xi_i^2}\pi^{\rm Ga}(1/\xi_i)\\
        &\sim\alpha\tilde\eta_i cd^c\frac{1}{(\sigma_i^2)^2}\frac{1}{\xi_i^{c+1}} \sim \alpha\tilde\eta_i  cd^c\frac{1}{(\sigma_i^2)^2}\left(\frac{1}{\log(\sigma_i^2/\tilde\eta_i)^\alpha}\right)^{c+1}\to0.
        \end{align*}

    Finally, we prove Part (iii). By the MCT and the lower bound in (\ref{bounds_gamma_func}), as $\sigma_i^2\to\tilde\eta_i$, we have $\xi_i\sim\frac{\alpha}{2\tilde\eta_i^2}
    (\sigma_i^2-\tilde\eta_i)^2\to0$ and 
            
        \begin{align*}
        p^{\rm Ga}(\sigma_i^2\mid\alpha,\boldsymbol{z}_i^\top\boldsymbol{\eta})&\sim\frac{\alpha}{\tilde\eta_i}\int_0^\infty\omega_i e^{-\xi_i\omega_i}f(\omega_i,\alpha)\pi^{\rm Ga}(\omega_i) ~d\omega_i\\
        &\to\frac{\alpha}{\tilde\eta_i}\int_0^\infty\omega_i f(\omega_i,\alpha)\pi^{\rm Ga}(\omega_i)~d\omega_i\\
        &=\frac{\alpha}{\tilde\eta_i}\int_0^\infty\sqrt{\omega_i}\frac{(\alpha\omega_i)^{\alpha\omega_i-1/2}e^{-\alpha\omega_i}}{\sqrt{\alpha}\Gamma(\alpha\omega_i)}\pi^{\rm Ga}(\omega_i)~d\omega_i\\
        &\leq\frac{\alpha}{\tilde\eta_i}\frac{1}{\sqrt{2\pi\alpha}}\frac{ d^c}{\Gamma(c)}\int_0^\infty\omega_i^{(1/2+c)-1}e^{-d\omega_i}~d\omega_i=\frac{\alpha}{\tilde\eta_i}\frac{1}{\sqrt{2\pi\alpha}}\frac{ d^c}{\Gamma(c)}\frac{\Gamma(c+1/2)}{d^{c+1/2}}<\infty.
        \end{align*}
\end{proof}

\subsection{Proofs of Proposition \ref{proposition_post} and Theorem \ref{thm:prob_kappa}}
\label{proof_Proposition_2_2_theorem_2_1}

In this section, we let $\Delta_i^{\rm post}=(y_i,D_i,n_i,\alpha,\boldsymbol{z}_i^\top\boldsymbol{\eta},\boldsymbol{x}_i^\top\boldsymbol{\beta},\sigma_u^2)$, $\tilde\eta_i=\exp(\boldsymbol{z}_i^\top\boldsymbol{\eta})$, and $\tilde n_i=(n_i-1)/2$.

\begin{proof}[Proof of Proposition \ref{proposition_post}]
Assume $\mu=\E(\theta_i\mid \Delta_i^{\rm post})$. Then by the identity $(y_i-\theta_i)^2=(\mu-y_i)^2+2(\mu-y_i)(\theta_i-\mu)+(\theta_i-\mu)^2$, we obtain the following decomposition:

\begin{align*}
    \E_{\theta_i}\left[(y_i-\theta_i)^2\mid \Delta_i^{\rm post}\right]=\left(\E(\theta_i\mid \Delta_i^{\rm post})-y_i\right)^2+{\rm Var}\left(\theta_i\mid \Delta_i^{\rm post}\right).
\end{align*}

Using the fact that the full conditional distribution of $\theta_i$ is $p(\theta_i\mid y_i,\boldsymbol{x}_i^\top\boldsymbol{\beta},\sigma_u^2,\sigma_i^2)=N((1-\gamma_{i,\sigma_i^2})y_i+\gamma_{i,\sigma_i^2}\boldsymbol{x}_i^\top\boldsymbol{\beta},\gamma_{i,\sigma_i^2}\sigma_u^2)$ where $\gamma_{i,\sigma_i^2}=\sigma_i^2/(\sigma_i^2+\sigma_u^2)$, the first and second terms of the decomposition can be expressed as follows:

\begin{align*}
    \left(\E(\theta_i\mid \Delta_i^{\rm post})-y_i\right)^2&=\left(\E_{\sigma_i^2}(\E(\theta_i\mid y_i,\boldsymbol{x}_i^\top  \boldsymbol{\beta},\sigma_u^2,\sigma_i^2)\mid \Delta_i^{\rm post})-y_i\right)^2\\
    &=\left(\E_{\sigma_i^2}\left((1-\gamma_{i,\sigma_i^2})y_i+\gamma_{i,\sigma_i^2}\boldsymbol{x}_i^\top\boldsymbol{\beta}\mid \Delta_i^{\rm post}\right)-y_i\right)^2\\
    &=\left(-\E_{\sigma_i^2}(\gamma_{i,\sigma_i^2}\mid \Delta_i^{\rm post})y_i+\E_{\sigma_i^2}(\gamma_{i,\sigma_i^2}\mid \Delta_i^{\rm post})~\boldsymbol{x}_i^\top\boldsymbol{\beta}\right)^2\\
    &=\left(y_i-\boldsymbol{x}_i^\top\boldsymbol{\beta}\right)^2\left(\E_{\sigma_i^2}(\gamma_{i,\sigma_i^2}\mid \Delta_i^{\rm post})\right)^2
\end{align*}

and

\begin{align*}
    {\rm Var}&\left(\theta_i\mid \Delta_i^{\rm post}\right)\\
    &=\E_{\sigma_i^2}({\rm Var}(\theta_i\mid y_i,\boldsymbol{x}_i^\top  \boldsymbol{\beta},\sigma_u^2,\sigma_i^2)\mid \Delta_i^{\rm post})+{\rm Var}_{\sigma_i^2}(\E(\theta_i\mid y_i,\boldsymbol{x}_i^\top  \boldsymbol{\beta},\sigma_u^2,\sigma_i^2)\mid \Delta_i^{\rm post})\\
    &=\E_{\sigma_i^2}\left(\gamma_{i,\sigma_i^2}\sigma_u^2\mid \Delta_i^{\rm post}\right)+{\rm Var}_{\sigma_i^2}\left((1-\gamma_{i,\sigma_i^2})y_i+\gamma_{i,\sigma_i^2}\boldsymbol{x}_i^\top\boldsymbol{\beta}\mid \Delta_i^{\rm post}\right)\\
    &=\sigma_u^2\E_{\sigma_i^2}(\gamma_{i,\sigma_i^2}\mid \Delta_i^{\rm post})+(y_i-\boldsymbol{x}_i^\top\boldsymbol{\beta})^2{\rm Var}_{\sigma_i^2}(\gamma_{i,\sigma_i^2}\mid \Delta_i^{\rm post})\\
    &=\sigma_u^2\E_{\sigma_i^2}(\gamma_{i,\sigma_i^2}\mid \Delta_i^{\rm post})+(y_i-\boldsymbol{x}_i^\top\boldsymbol{\beta})^2\left[\E_{\sigma_i^2}(\gamma^2_{i,\sigma_i^2}\mid \Delta_i^{\rm post})-\left(\E_{\sigma_i^2}(\gamma_{i,\sigma_i^2}\mid \Delta_i^{\rm post})\right)^2\right].
\end{align*}

Combining the above two results yields the following exact representation of the expectation term in (\ref{observed_var}):

\begin{align*}
    \E_{\theta_i}&\left[(y_i-\theta_i)^2\mid \Delta_i^{\rm post}\right]=\sigma_u^2\E_{\sigma_i^2}(\gamma_{i,\sigma_i^2}\mid \Delta_i^{\rm post})+(y_i-\boldsymbol{x}_i^\top\boldsymbol{\beta})^2\E_{\sigma_i^2}(\gamma^2_{i,\sigma_i^2}\mid \Delta_i^{\rm post}),
\end{align*}

where

\begin{align*}
    \mathbb{E}_{\sigma^2_i}(\gamma^k_{i,\sigma^2_i}\mid \Delta_i^{\rm post}) 
    &= \dfrac{\int_0^\infty\int_{-\infty}^\infty \int_{0}^{\infty}\gamma^k_{i,\sigma^2_i} ~p(\sigma_i^2, \theta_i,\omega_i\mid\Delta_i^{\rm post})~d\sigma^2_i d\theta_id\omega_i }{\int_0^\infty\int_{-\infty}^\infty\int_{0}^{\infty}  p(\sigma_i^2, \theta_i,\omega_i\mid\Delta_i^{\rm post})~ d\sigma^2_id\theta_id\omega_i},\quad {\rm for}~k=1,2,
\end{align*}

where

\begin{align}\label{pf_prop2}
    p(\sigma_i^2, \theta_i,\omega_i&\mid\Delta_i^{\rm post})=\frac{\omega_i^{a-1}}{(1+\omega_i)^{a+b}}~N(\theta_i\mid\boldsymbol{x}_i^\top\boldsymbol{\beta},\sigma_u^2)\nonumber\\
    &\times\frac{(\alpha\omega_i\tilde\eta_i)^{\alpha\omega_i+1}}{\Gamma(\alpha\omega_i+1)}\left(\frac{1}{\sigma_i^2}\right)^{\frac{n_i}{2}+(\alpha\omega_i+1)+1}\exp\left\{-\frac{(y_i-\theta_i)^2/2+\tilde n_iD_i+\alpha\omega_i\tilde\eta_i}{\sigma_i^2}\right\}.
\end{align}


\end{proof}

\begin{proof}[Proof of Theorem \ref{thm:prob_kappa}]
Let $\varepsilon_i\in(0,1)$. Then

\begin{align*}
    \P(\kappa^{\rm post}_i<\varepsilon_i\mid \Delta_i^{\rm post})=\int_0^{c_i/\alpha} p(\omega_i\mid \Delta_i^{\rm post})d\omega_i/\int_0^\infty   p(\omega_i\mid \Delta_i^{\rm post}) d\omega_i,\quad
\end{align*}

where $c_i= (n_i/2)(\varepsilon_i/(1-\varepsilon_i))$, and from (\ref{pf_prop2}),

\begin{align*}
     p(\omega_i\mid \Delta_i^{\rm post})&=
      \int_{-\infty}^\infty\int_0^\infty p(\sigma_i^2,\theta_i,\omega_i\mid\Delta_i^{\rm post})~d\sigma_i^2d\theta_i\\
     &=\frac{\Gamma( n_i/2+\alpha\omega_i+1)}{\Gamma(\alpha\omega_i+1)}\frac{\omega_i^{a-1}}{(1+\omega_i)^{a+b}}\nonumber\\
    &\quad\int_{-\infty}^\infty\frac{(\alpha\omega_i\tilde\eta_i)^{\alpha\omega_i+1}}{\left((y_i-\theta_i)^2/2+\tilde n_iD_i+\alpha\omega_i\tilde\eta_i\right)^{ n_i/2+\alpha\omega_i+1}}N(\theta_i\mid\boldsymbol{x}_i^\top\boldsymbol{\beta},\sigma_u^2)~d\theta_i.
\end{align*}

We make a change of variable $\omega_i=(\tilde n_iD_i)u_i$, so now $\P(\kappa^{\rm post}_i<\varepsilon_i\mid \Delta_i^{\rm post})$ can be expressed as follows:

\begin{align*}
    \P(\kappa^{\rm post}_i<\varepsilon_i\mid \Delta_i^{\rm post})=I_0/I_1
\end{align*}

where

\begin{align}\label{define_I0I1}
    I_0=\int_0^{c_i/\alpha\tilde n_iD_i} g(u_i,\alpha) du_i,\quad I_1=\int_0^{\infty} g(u_i,\alpha) du_i,
\end{align}

and

\begin{align}
    g(u_i,\alpha)&=\alpha\tilde\eta_i\frac{\Gamma(n_i/2+\alpha(\tilde n_iD_i)u_i+1)}{\Gamma(\alpha(\tilde n_iD_i)u_i+1)}\frac{u_i^{a}}{(1+\tilde n_iD_iu_i)^{a+b}}~ 
    \nonumber\\
    &\quad\int_{-\infty}^\infty\frac{(\alpha u_i\tilde\eta_i)^{\alpha(\tilde n_iD_i)u_i}}{\left(1+(y_i-\theta_i)^2/(2\tilde n_iD_i)+\alpha u_i\tilde\eta_i\right)^{ n_i/2+\alpha(\tilde n_iD_i)u_i+1}}N(\theta_i\mid\boldsymbol{x}_i^\top\boldsymbol{\beta},\sigma_u^2)~d\theta_i\label{19}.
\end{align}

First, we find the upper bound of $I_0$ defined in (\ref{define_I0I1}). Using (\ref{19}), we have

\begin{align}
    g(u_i,\alpha)
    &\leq \alpha\tilde\eta_i\frac{\Gamma(n_i/2+\alpha(\tilde n_iD_i)u_i+1)}{\Gamma(\alpha(\tilde n_iD_i)u_i+1)}\frac{u_i^{a}}{(1+\tilde n_iD_iu_i)^{a+b}}
    \nonumber\\
    &\quad\int_{-\infty}^\infty\frac{(\alpha u_i\tilde\eta_i)^{\alpha(\tilde n_iD_i)u_i}}{\left(1+\alpha u_i\tilde\eta_i\right)^{ n_i/2+\alpha(\tilde n_iD_i)u_i+1}}N(\theta_i\mid\boldsymbol{x}_i^\top\boldsymbol{\beta},\sigma_u^2)~d\theta_i\nonumber\\
    &=\alpha\tilde\eta_i\frac{\Gamma(n_i/2+\alpha(\tilde n_iD_i)u_i+1)}{\Gamma(\alpha(\tilde n_iD_i)u_i+1)}\frac{u_i^{a}}{(1+\tilde n_iD_iu_i)^{a+b}}\frac{(\alpha u_i\tilde\eta_i)^{\alpha(\tilde n_iD_i)u_i}}{\left(1+\alpha u_i\tilde\eta_i\right)^{ n_i/2+\alpha(\tilde n_iD_i)u_i+1}}\nonumber\\
    &\leq \alpha\tilde\eta_i\frac{\Gamma(n_i/2+\alpha(\tilde n_iD_i)u_i+1)}{\Gamma(\alpha(\tilde n_iD_i)u_i+1)}\frac{u_i^{a}}{(1+\tilde n_iD_iu_i)^{a+b}}\frac{1}{\left(1+\alpha u_i\tilde\eta_i\right)^{ n_i/2+1}}\label{20}.
\end{align}

Using the bounds of a Gamma function in (\ref{bounds_gamma_func}), we can show that

\begin{align}
    \frac{\Gamma(n_i/2+\alpha(\tilde n_iD_i)u_i+1)}{\Gamma(\alpha(\tilde n_iD_i)u_i+1)}&\leq\frac{( n_i/2+\alpha(\tilde n_iD_i)u_i+1)^{ n_i/2+\alpha(\tilde n_iD_i)u_i+1-1/2}}{(\alpha(\tilde n_iD_i)u_i+1)^{\alpha(\tilde n_iD_i)u_i+1-1/2}e^{n_i/2}}\nonumber\\
    &=\frac{(n_i/2+\alpha(\tilde n_iD_i)u_i+1)^{n_i/2}}{e^{ n_i/2}}\left(1+\frac{n_i/2}{\alpha(\tilde n_iD_i)u_i+1}\right)^{\alpha(\tilde n_iD_i)u_i+1/2}\nonumber\\
    &\leq\frac{(n_i/2+\alpha(\tilde n_iD_i)u_i+1)^{n_i/2}}{e^{ n_i/2}}\left(1+\frac{n_i/2}{\alpha(\tilde n_iD_i)u_i+1}\right)^{\alpha(\tilde n_iD_i)u_i+1}\nonumber\\
    &\leq(n_i/2+\alpha(\tilde n_iD_i)u_i+1)^{n_i/2}\label{21}.
\end{align}

Thus, from (\ref{20}) and (\ref{21}), we obtain that

\begin{align}
    g(u_i,\alpha)
    &\leq \alpha\eta_i\frac{(n_i/2+\alpha(\tilde n_iD_i)u_i+1)^{n_i/2}}{(1+\alpha u_i\eta_i)^{n_i/2+1}}\frac{u_i^{a}}{(1+\tilde n_iD_iu_i)^{a+b}}\label{22}.
\end{align}

By (\ref{22}), if $n_i/2>a$, we obtain the upper bound of $I_0$ defined in (\ref{define_I0I1}) as follows:

\begin{align}
    I_0&\leq \alpha\tilde\eta_i\int_0^{c_i/\alpha\tilde n_iD_i}\frac{( n_i/2+\alpha(\tilde n_iD_i)u_i+1)^{n_i/2}}{(1+\alpha u_i\tilde\eta_i)^{n_i/2+1}}\frac{u_i^{a}}{(1+\tilde n_iD_iu_i)^{a+b}}~du_i\nonumber\\
    &= \tilde\eta_i\int_0^{c_i/\tilde n_iD_i}\frac{(n_i/2+\tilde n_iD_it_i+1)^{n_i/2}}{(1+\tilde\eta_it_i)^{ n_i/2+1}}\frac{(t_i/\alpha)^a}{(1+\tilde n_iD_it_i/\alpha)^{a+b}}~ dt_i;&&t_i=\alpha u_i\nonumber\\
    &\leq C_1~\alpha^{-a}\int_0^{c_i/\tilde n_iD_i}\frac{t_i^a}{(1+\tilde\eta_it_i)^{n_i/2+1}}dt_i\nonumber\\
    &=C_1~\alpha^{-a}\int_0^{\tilde\eta_ic_i/\tilde n_iD_i}\frac{(v_i/\tilde\eta_i)^a}{(1+v_i)^{n_i/2+1}}\frac{dv_i}{\tilde\eta_i};&&v_i=\tilde\eta_it_i\nonumber\\
    &\leq C_2~\alpha^{-a}\int_0^\infty\frac{v_i^a}{(1+v_i)^{n_i/2+1}}dv_i
    ~ \leq C_3~\alpha^{-a}.\label{23}
\end{align}

Next, we find the lower bound of $I_1$ (i.e., the upper bound of $I_1^{-1}$) defined in (\ref{define_I0I1}). From (\ref{19}), we have

\begin{align}
    g(u_i,\alpha)
    &\geq \alpha\tilde\eta_i\frac{\Gamma(n_i/2+\alpha(\tilde n_iD_i)u_i+1)}{\Gamma(\alpha(\tilde n_iD_i)u_i+1)}\frac{u_i^{a}}{(1+\tilde n_iD_iu_i)^{a+b}}\nonumber\\
    &\quad\int_{-\infty}^\infty\frac{\exp\left\{-((y_i-\theta_i)^2/2+\tilde n_iD_i)/\tilde\eta_i\right\}}{\left(1+(y_i-\theta_i)^2/(2\tilde n_iD_i)+\alpha u_i\tilde\eta_i\right)^{ n_i/2+1}}N(\theta_i\mid\boldsymbol{x}_i^\top\boldsymbol{\beta},\sigma_u^2)~d\theta_i.\label{24}
\end{align}

Using the bounds of a Gamma function in (\ref{bounds_gamma_func}), we can show that

\begin{align}
    \frac{\Gamma(n_i/2+\alpha(\tilde n_iD_i)u_i+1)}{\Gamma(\alpha(\tilde n_iD_i)u_i+1)}&\geq\frac{( n_i/2+\alpha(\tilde n_iD_i)u_i+1)^{( n_i/2+\alpha(\tilde n_iD_i)u_i+1)-1/2}}{(\alpha(\tilde n_iD_i)u_i+1)^{(\alpha(\tilde n_iD_i)u_i+1)-1/2}e^{1/12+n_i/2}}\nonumber\\
    &=\frac{(n_i/2+\alpha(\tilde n_iD_i)u_i+1)^{n_i/2}}{e^{1/12+ n_i/2}}\frac{(n_i/2+\alpha(\tilde n_iD_i)u_i+1)^{\alpha(\tilde n_iD_i)u_i+1/2}}{(\alpha(\tilde n_iD_i)u_i+1)^{\alpha(\tilde n_iD_i)u_i+1/2}}\nonumber\\
    &\geq\frac{(n_i/2+\alpha(\tilde n_iD_i)u_i+1)^{n_i/2}}{e^{1/12+ n_i/2}}.\label{25}
\end{align}

Thus, from (\ref{24}) and (\ref{25}), we obtain that

\begin{align*}
    g(u_i,\alpha)&\geq \alpha\tilde\eta_i \frac{(n_i/2+\alpha(\tilde n_iD_i)u_i+1)^{n_i/2}}{e^{1/12+n_i/2+\tilde n_iD_i/\tilde\eta_i}}\frac{u_i^a}{(1+\tilde n_iD_iu_i)^{a+b}}\nonumber\\
    &\quad\int_{-\infty}^\infty\frac{\exp\left\{-(y_i-\theta_i)^2/(2\tilde\eta_i)\right\}}{\left(1+(y_i-\theta_i)^2/(2\tilde n_iD_i)+\alpha u_i\tilde\eta_i\right)^{ n_i/2+1}}N(\theta_i\mid\boldsymbol{x}_i^\top\boldsymbol{\beta},\sigma_u^2)~d\theta_i\nonumber\\
    &\geq C_4~\alpha(n_i/2+\alpha(\tilde n_iD_i)u_i+1)^{n_i/2}\frac{u_i^a}{(1+\tilde n_iD_iu_i)^{a+b}}\nonumber\\
    &\quad\int_{-\infty}^\infty\frac{\exp\left\{-(y_i-\theta_i)^2/(2\tilde\eta_i)\right\}}{\left(1+(y_i-\theta_i)^2/(2\tilde n_iD_i)+\alpha u_i\tilde\eta_i\right)^{ n_i/2+1}}N(\theta_i\mid\boldsymbol{x}_i^\top\boldsymbol{\beta},\sigma_u^2)~d\theta_i.
\end{align*}

Therefore, the lower bound of $I_1$ defined in (\ref{define_I0I1}) is as follows:

\begin{align}
    I_1&\geq C_4\int_0^\infty\int_{-\infty}^\infty  \alpha(n_i/2+\alpha(\tilde n_iD_i)u_i+1)^{n_i/2}\frac{u_i^a}{(1+\tilde n_iD_iu_i)^{a+b}}\nonumber\\
    &\quad\frac{\exp\left\{-(y_i-\theta_i)^2/(2\tilde\eta_i)\right\}}{\left(1+(y_i-\theta_i)^2/(2\tilde n_iD_i)+\alpha u_i\tilde\eta_i\right)^{ n_i/2+1}}N(\theta_i\mid\boldsymbol{x}_i^\top\boldsymbol{\beta},\sigma_u^2)~d\theta_idu_i\nonumber\\
    &\geq C_4\int_1^2\int_{-\infty}^\infty  \alpha(n_i/2+\alpha(\tilde n_iD_i)u_i+1)^{n_i/2}\frac{u_i^a}{(1+\tilde n_iD_iu_i)^{a+b}}\nonumber\\
    &\quad\frac{\exp\left\{-(y_i-\theta_i)^2/(2\tilde\eta_i)\right\}}{\left(1+(y_i-\theta_i)^2/(2\tilde n_iD_i)+\alpha u_i\tilde\eta_i\right)^{ n_i/2+1}}N(\theta_i\mid\boldsymbol{x}_i^\top\boldsymbol{\beta},\sigma_u^2)~d\theta_idu_i\nonumber\\
    &\geq C_4\int_{-\infty}^\infty  \alpha(n_i/2+\alpha(\tilde n_iD_i)+1)^{n_i/2}\frac{1^a}{(1+2\tilde n_iD_i)^{a+b}}\nonumber\\
    &\quad\frac{\exp\left\{-(y_i-\theta_i)^2/(2\tilde\eta_i)\right\}}{\left(1+(y_i-\theta_i)^2/(2\tilde n_iD_i)+2\alpha \tilde\eta_i\right)^{ n_i/2+1}}N(\theta_i\mid\boldsymbol{x}_i^\top\boldsymbol{\beta},\sigma_u^2)~d\theta_i\nonumber\\
    &\geq C_5\int_{-\infty}^\infty  \frac{\alpha(n_i/2+\alpha(\tilde n_iD_i)+1)^{n_i/2}}{\left(1+(y_i-\theta_i)^2/(2\tilde n_iD_i)+2\alpha \tilde\eta_i\right)^{ n_i/2+1}}\exp\left\{-\frac{(y_i-\theta_i)^2}{2\tilde\eta_i}\right\}N(\theta_i\mid\boldsymbol{x}_i^\top\boldsymbol{\beta},\sigma_u^2)~d\theta_i.\label{26}
\end{align}

For $\alpha>1$ and $\theta_i\in\mathbb{R}$, we note that

\begin{align*}
    &\frac{\alpha(n_i/2+\alpha(\tilde n_iD_i)+1)^{n_i/2}}{\left(1+(y_i-\theta_i)^2/(2\tilde n_iD_i)+2\alpha \tilde\eta_i\right)^{ n_i/2+1}}\exp\left\{-\frac{(y_i-\theta_i)^2}{2\tilde\eta_i}\right\}N(\theta_i\mid\boldsymbol{x}_i^\top\boldsymbol{\beta},\sigma_u^2)\\
    &\quad =\frac{((n_i/2+1)/\alpha+ \tilde n_iD_i)^{n_i/2}}{\left((1+(y_i-\theta_i)^2/(2\tilde n_iD_i))/\alpha+2 \tilde\eta_i\right)^{ n_i/2+1}}\exp\left\{-\frac{(y_i-\theta_i)^2}{2\tilde\eta_i}\right\}N(\theta_i\mid\boldsymbol{x}_i^\top\boldsymbol{\beta},\sigma_u^2)\\
    &\quad\leq \frac{(n_i/2+1+\tilde n_iD_i)^{n_i/2}}{(2\tilde\eta_i)^{n_i/2+1}} ~\exp\left\{-\frac{(y_i-\theta_i)^2}{2\tilde\eta_i}\right\}N(\theta_i\mid\boldsymbol{x}_i^\top\boldsymbol{\beta},\sigma_u^2).
\end{align*}

Since the function in (\ref{26}) is bounded by the integrable function of $\theta_i$ for all $\alpha>1$, we apply the DCT to obtain that

\begin{align}
    \lim_{\alpha\to\infty}I_1\geq C_5\int_{-\infty}^\infty  \frac{ (\tilde n_iD_i)^{n_i/2}}{  (2\tilde\eta_i)^{ n_i/2+1}}\exp\left\{-\frac{(y_i-\theta_i)^2}{2\tilde\eta_i}\right\}N(\theta_i\mid\boldsymbol{x}_i^\top\boldsymbol{\beta},\sigma_u^2)~d\theta_i=C_6<\infty.\label{27}
\end{align}

From the upper bound of $I_0$ in (\ref{23}) and the upper bound of $I_1^{-1}$ in (\ref{27}), we have

\begin{align*}
    0\leq \lim_{\alpha\to\infty}\P(\kappa^{\rm post}_i<\varepsilon_i\mid \Delta_i^{\rm post})=\lim_{\alpha\to\infty}\frac{I_0}{I_1}
    \leq \lim_{\alpha\to\infty}\frac{C_3}{C_6} \alpha^{-a}=0.
\end{align*}

\end{proof}

\subsection{Proof of Theorem \ref{thm:proper_posterior}}
\label{proof_Theorem_3_1}

\begin{proof}
Denote $\tilde n_i=(n_i-1)/2$ and let $\textbf{X}=(\boldsymbol{x}_1,...,\boldsymbol{x}_m)$ be a ($p\times m$) matrix of covariates in the FH model (\ref{FH_model}). We will prove that the joint posterior distribution in (\ref{HB_posterior}) is proper, that is,
\begin{align*}
\int\pi(\boldsymbol{\theta},\boldsymbol{\sigma}^2,\boldsymbol{\beta},\sigma_u^2,\alpha,\boldsymbol{\omega},\boldsymbol{\eta}\mid \boldsymbol{\mathcal{D}})~d\boldsymbol{\theta}~d\boldsymbol{\sigma}^2~d\boldsymbol{\beta}~d\sigma_u^2 ~d\alpha~d\boldsymbol{\omega}~d\boldsymbol{\eta}<\infty 
\end{align*}

where 
\begin{align*}
\pi(\boldsymbol{\theta},\boldsymbol{\sigma}^2,\boldsymbol{\beta},\sigma_u^2,\alpha,&\boldsymbol{\omega},\boldsymbol{\eta}\mid \boldsymbol{\mathcal{D}})
\propto(\sigma_u^2)^{-m/2}\times\pi(\alpha)\nonumber\\
&\times\prod_{i=1}^m\left[\left(\frac{1}{\sigma_i^2}\right)^{\frac{n_i}{2}+(\alpha\omega_i+1)+1}\times\frac{(\alpha\omega_i\exp(\boldsymbol{z}_i^\top\boldsymbol{\eta}))^{\alpha\omega_i+1}}{\Gamma(\alpha\omega_i+1)}\times\pi^{\rm BP}(\omega_i)\right]\nonumber\\
&\times\prod_{i=1}^m\exp\left\{-\frac{(y_i-\theta_i)^2/2+\tilde n_iD_i+\alpha\omega_i\exp(\boldsymbol{z}_i^\top\boldsymbol{\eta})}{\sigma_i^2}-\frac{(\theta_i-\boldsymbol{x}_i^\top\boldsymbol{\beta})^2}{2\sigma_u^2}\right\}.
\end{align*}

Let $\varphi_i(\theta_i,\alpha,\omega_i,\boldsymbol{\eta})=(y_i-\theta_i)^2/2+\tilde n_i D_i+\alpha\omega_i\exp(\boldsymbol{z}_i^\top\boldsymbol{\eta})$ for $i=1,...,m$. We integrate (\ref{HB_posterior}) with respect to $\sigma_1^2,...,\sigma_m^2$ to obtain

\begin{align*}
\pi(\boldsymbol{\theta},&\boldsymbol{\beta},\sigma_u^2,\alpha,\boldsymbol{\omega},\boldsymbol{\eta}\mid \boldsymbol{\mathcal{D}})\propto~(\sigma_u^2)^{-m/2}\times\pi(\alpha)\times\prod_{i=1}^m\left[\frac{(\alpha\omega_i\exp(\boldsymbol{z}_i^\top\boldsymbol{\eta}))^{\alpha\omega_i+1}}{\Gamma(\alpha\omega_i+1)}\times\pi^{\rm BP}(\omega_i)\right]\\
&\times\prod_{i=1}^m\left[\Gamma\left(\frac{n_i}{2}+\alpha\omega_i+1\right)\times\varphi_i(\theta_i,\alpha,\omega_i,\boldsymbol{\eta})^{-(n_i/2+\alpha\omega_i+1)}\right]\times\exp\left\{-\frac{(\boldsymbol{\theta-\textbf{X}^\top\boldsymbol{\beta}})^\top(\boldsymbol{\theta}-\textbf{X}^\top\boldsymbol{\beta})}{2\sigma_u^2}\right\}.
\end{align*}

Next, let $\textbf{A}=\textbf{I}_m-\textbf{X}^\top(\textbf{X}\textbf{X}^\top)^{-1}\textbf{X}$. We use the fact that

\begin{align*}
    \int_{\mathbb{R}^p}\exp\left\{-\frac{(\boldsymbol{\theta-\textbf{X}^\top\boldsymbol{\beta}})^\top(\boldsymbol{\theta}-\textbf{X}^\top\boldsymbol{\beta})}{2\sigma_u^2}\right\}~d\boldsymbol{\beta}=( 2\pi\sigma_u^2)^{p/2}| \textbf{X}\textbf{X}^\top|^{-1/2}\exp\left\{-\frac{\boldsymbol{\theta}^\top\textbf{A}\boldsymbol{\theta}}{2\sigma_u^2}\right\}
\end{align*}

to obtain

\begin{align}\label{proof_thm_joint_w/o_beta}
    \pi(\boldsymbol{\theta}&,\sigma_u^2,\alpha,\boldsymbol{\omega},\boldsymbol{\eta}\mid \boldsymbol{\mathcal{D}})\propto~(\sigma_u^2)^{-(m-p-2)/2-1}\exp\left\{-\frac{\boldsymbol{\theta}^\top\textbf{A}\boldsymbol{\theta}}{2\sigma_u^2}\right\}\times\pi(\alpha)\nonumber\\
    &\times\prod_{i=1}^m\left[\frac{\Gamma\left(\frac{n_i}{2}+\alpha\omega_i+1\right)(\alpha\omega_i\exp(\boldsymbol{z}_i^\top\boldsymbol{\eta}))^{\alpha\omega_i+1}}{\Gamma(\alpha\omega_i+1)}\times\pi^{\rm BP}(\omega_i)\times\varphi_i(\theta_i,\alpha,\omega_i,\boldsymbol{\eta})^{-(n_i/2+\alpha\omega_i+1)}\right].
\end{align}

When $m-p-2>0$, we can integrate the equation in (\ref{proof_thm_joint_w/o_beta}) with respect to $\sigma_u^2$ to obtain

\begin{align}\label{proof_thm_joint_w/o_sigma2_u}
    \pi(\boldsymbol{\theta}&,\alpha,\boldsymbol{\omega},\boldsymbol{\eta}\mid \boldsymbol{\mathcal{D}})\propto~(\boldsymbol{\theta}^\top\textbf{A}\boldsymbol{\theta})^{-(m-p-2)/2}\times\pi(\alpha)\nonumber\\
    &\times\prod_{i=1}^m\left[\frac{\Gamma\left(\frac{n_i}{2}+\alpha\omega_i+1\right)(\alpha\omega_i\exp(\boldsymbol{z}_i^\top\boldsymbol{\eta}))^{\alpha\omega_i+1}}{\Gamma(\alpha\omega_i+1)}\times\pi^{\rm BP}(\omega_i)\times\varphi_i(\theta_i,\alpha,\omega_i,\boldsymbol{\eta})^{-(n_i/2+\alpha\omega_i+1)}\right].
\end{align}

\cite{2017_Sukasawa} computed the integral of 
(\ref{proof_thm_joint_w/o_sigma2_u}) 
with respect to $\boldsymbol{\theta}$ by decomposing it into the domains of $\Omega=\{\boldsymbol{\theta}\mid\boldsymbol{\theta}^\top\textbf{A}\boldsymbol{\theta}\leq1\}\subset\mathbb{R}^m$ and $\Omega^c$, as follows:

\begin{align}\label{proof_thm_correction_sugasawa}
    \int_{\mathbb{R}^m}(&\boldsymbol{\theta}^\top\textbf{A}\boldsymbol{\theta})^{-(m-p-2)/2}\times\prod_{i=1}^m\varphi_i(\theta_i,\alpha,\omega_i,\boldsymbol{\eta})^{-(n_i/2+\alpha\omega_i+1)}~d\boldsymbol{\theta}\nonumber\\
    \leq &~C_1~\prod_{i=1}^{m-p}\left(\tilde n_i D_i+\alpha\omega_i\exp(\boldsymbol{z}_i^\top\boldsymbol{\eta})\right)^{-(n_i/2+\alpha\omega_i+1)}\nonumber\\
    &\hspace{2cm}\times\prod_{i=m-p+1}^m\int_{-\infty}^\infty\varphi_i(\theta_i,\alpha,\omega_i,\boldsymbol{\eta})^{-(n_i/2+\alpha\omega_i+1)}~d\theta_i&&\text{over the domain $\Omega$}\nonumber\\
    &+C_2~\prod_{i=1}^m\int_{-\infty}^\infty\varphi_i(\theta_i,\alpha,\omega_i,\boldsymbol{\eta})^{-(n_i/2+\alpha\omega_i+1)}~d\theta_i &&\text{over the domain $\Omega^c$}.
\end{align}

We then make a transformation $\mu_i= (y_i-\theta_i)/\sqrt{2\tilde n_i D_i+2\alpha\omega_i\exp(\boldsymbol{z}_i^\top\boldsymbol{\eta})}$ to compute

\begin{align*}
    \int_{-\infty}^\infty& 
    \varphi_i(\theta_i,\alpha,\omega_i,\boldsymbol{\eta})~d\theta_i\\
    &=\sqrt{2}\left(\tilde n_i D_i+\alpha\omega_i\exp(\boldsymbol{z}_i^\top\boldsymbol{\eta})\right)^{-(\tilde n_i+\alpha\omega_i+1)}\int_{-\infty}^\infty(1+\mu_i^2)^{-(n_i/2+\alpha\omega_i+1)}~d\mu_i\\
    &\leq \sqrt{2\pi}\left(\tilde n_i D_i+\alpha\omega_i\exp(\boldsymbol{z}_i^\top\boldsymbol{\eta})\right)^{-(\tilde n_i+\alpha\omega_i+1)}(n_i/2+\alpha\omega_i+1)^{-1/2}
\end{align*}

because, let $B=(n_i/2+\alpha\omega_i+1)$,

\begin{align*}
    \int_{-\infty}^\infty(1+\mu_i^2)^{-B}d\mu_i&=2\int_0^\infty(1+\mu_i^2)^{-B}d\mu_i=\sqrt{\pi}~\frac{\Gamma(B-1/2)}{\Gamma(B)}\leq\sqrt{\pi}~ B^{-1/2}.
\end{align*}

Thus, the integral of $\pi(\boldsymbol{\theta},\alpha,\boldsymbol{\omega},\boldsymbol{\eta}\mid \boldsymbol{\mathcal{D}})$ in (\ref{proof_thm_joint_w/o_sigma2_u}) with respect to $\boldsymbol{\theta}$ is bounded by

\begin{align}\label{proof_thm_joint_w/o_theta}
    \int\pi(\boldsymbol{\theta},\alpha,\boldsymbol{\omega},\boldsymbol{\eta}\mid \boldsymbol{\mathcal{D}})~d\boldsymbol{\theta}<&C_3\times\pi(\alpha)\times\prod_{i=1}^m\left[\frac{\Gamma(n_i/2+\alpha\omega_i+1)}{\Gamma(\alpha\omega_i+1)}\times(n_i/2+\alpha\omega_i+1)^{-1/2}\times\pi^{\rm BP}(\omega_i)\right]\nonumber\\
    &\times\prod_{i=1}^m~(\alpha\omega_i\exp(\boldsymbol{z}_i^\top\boldsymbol{\eta}))^{\alpha\omega_i+1}(\tilde n_iD_i+\alpha\omega_i\exp(\boldsymbol{z}_i^\top\boldsymbol{\eta}))^{-(\tilde n_i+\alpha\omega_i+1)}.
\end{align}

Note that the expression in (\ref{proof_thm_joint_w/o_theta}) is only considered the second term in (\ref{proof_thm_correction_sugasawa}) evaluated over the domain $\Omega^c$ of $\boldsymbol{\theta}$ and the proof of $\int\pi(\boldsymbol{\theta},\alpha,\boldsymbol{\omega},\boldsymbol{\eta}\mid \boldsymbol{\mathcal{D}})~d\boldsymbol{\theta}d\boldsymbol{\eta}d\boldsymbol{\omega}d\alpha<\infty $ below focuses on this domain. A similar proof with small modifications can be applied to show that the aforementioned integral is also finite over the domain $\Omega$ of $\boldsymbol{\theta}$.

Next, we use the bounds of a Gamma function in (\ref{bounds_gamma_func}) to obtain

\begin{align*}
    \frac{\Gamma(n_i/2+\alpha\omega_i+1)}{\Gamma(\alpha\omega_i+1)}&\leq \frac{(n_i/2+\alpha\omega_i+1)^{n_i/2+\alpha\omega_i+1-1/2}e^{-n_i/2}}{(\alpha\omega_i+1)^{\alpha\omega_i+1-1/2}}\\
    &=e^{-n_i/2}(n_i/2+\alpha\omega_i+1)^{n_i/2}\left(1+\frac{n_i/2}{\alpha\omega_i+1}\right)^{\alpha\omega_i+1/2}\\
    &\leq e^{-n_i/2}(n_i/2+\alpha\omega_i+1)^{n_i/2}\left(1+\frac{n_i/2}{\alpha\omega_i+1}\right)^{\alpha\omega_i+1}\\
    &\leq (n_i/2+\alpha\omega_i+1)^{n_i/2}.
\end{align*}

Hence, the expression in (\ref{proof_thm_joint_w/o_theta}) can be bounded further by

\begin{align*}
    \int\pi(\boldsymbol{\theta},\alpha,\boldsymbol{\omega},\boldsymbol{\eta}\mid \boldsymbol{\mathcal{D}})~d\boldsymbol{\theta}<&C_4\times\pi(\alpha)\times\prod_{i=1}^m [(n_i/2+\alpha\omega_i+1)^{\tilde n_i}\times\pi^{\rm BP}(\omega_i)]\nonumber\\
    &\times\prod_{i=1}^m~(\alpha\omega_i\exp(\boldsymbol{z}_i^\top\boldsymbol{\eta}))^{\alpha\omega_i+1}(\tilde n_iD_i+\alpha\omega_i\exp(\boldsymbol{z}_i^\top\boldsymbol{\eta}))^{-(\tilde n_i+\alpha\omega_i+1)}.
\end{align*}

Let $\mathbb{R}_+=\{x\in\mathbb{R}\mid x>0\}$. To complete the proof, we need to show that

\begin{align*}
    I:=&\int_{\mathbb{R}_+^{m+1}}\pi(\alpha)\times\prod_{i=1}^m [(n_i/2+\alpha\omega_i+1)^{\tilde n_i}\times(\alpha\omega_i)^{\alpha\omega_i+1}\times\pi^{\rm BP}(\omega_i)]\\
    &\times\int_{\mathbb{R}^q}~\prod_{i=1}^m~(\exp(\boldsymbol{z}_i^\top\boldsymbol{\eta}))^{\alpha\omega_i+1}(\tilde n_iD_i+\alpha\omega_i\exp(\boldsymbol{z}_i^\top\boldsymbol{\eta}))^{-(\tilde n_i+\alpha\omega_i+1)}~d\boldsymbol{\eta}~d\boldsymbol{\omega}d\alpha<\infty.
\end{align*}

To do so, we follow the approach in \cite{2017_Sukasawa}. Without loss of generality, we assume that the signs of $z_{ik}$ are all positive for $i=1,...,m$ and $k=1,...,q$. Since $\tilde n_iD_i$ and $\alpha\omega_i\exp(\boldsymbol{z}_i^\top\boldsymbol{\eta})$ are positive, we obtain that

\begin{align}\label{proof_thm_joint_eta}
I\leq\int_{\mathbb{R}_+^{m+1}}\pi(\alpha)\times&\prod_{i=1}^m [(n_i/2+\alpha\omega_i+1)^{\tilde n_i}\times(\alpha\omega_i)^{\alpha\omega_i+1}\times\pi^{\rm BP}(\omega_i)]\nonumber\\&\times\int_{\mathbb{R}^q}~\prod_{k=1}^q \exp(\eta_k)^{B_{1k}}\left\{C_*+b_*\prod_{k=1}^q\exp(\eta_k)^{B_{1k}+B_{2k}}\right\}^{-1}~d\boldsymbol{\eta}~d\boldsymbol{\omega}d\alpha,
\end{align}

where $B_{1k}=\sum_{i=1}^m(\alpha\omega_i+1)z_{ik}$, $B_{2k}=\sum_{i=1}^m\tilde n_iz_{ik}$, $b_*=\prod_{i=1}^m(\alpha\omega_i)^{\tilde n_i+\alpha\omega_i+1}$, and $C_*=\prod_{i=1}^m(\tilde n_iD_i)^{\tilde n_i+\alpha\omega_i+1}$. Because $B_{1k}>0$ and $B_{2k}>0$, there exists $\lambda>0$ such that $B_{1k}>1/\lambda>0$, for all $k$. Making a transformation $\phi_k=\exp(\eta_k/\lambda)$, the expression in (\ref{proof_thm_joint_eta}) is given by

\begin{align*}
I\leq\int_{\mathbb{R}_+^{m+1}}\pi(\alpha)\times&\prod_{i=1}^m [(n_i/2+\alpha\omega_i+1)^{\tilde n_i}\times(\alpha\omega_i)^{\alpha\omega_i+1}\times\pi^{\rm BP}(\omega_i)]\nonumber\\&\times\lambda^q\int_{\mathbb{R}_+^q}~\prod_{k=1}^q \phi_k^{\lambda B_{1k}-1}\left\{C_*+b_*\prod_{k=1}^q \phi_k^{\lambda B_{1k}+\lambda B_{2k}}\right\}^{-1}~d\boldsymbol{\phi}~d\boldsymbol{\omega}d\alpha,
\end{align*}

where $\boldsymbol{\phi}=(\phi_1,...,\phi_q)^\top$. We decompose the integral $I$ into the domains $\{\phi_k\leq1\}$ or $\{\phi_k\geq1\}$ for $k=1,...,q$, and $\{\alpha\omega_i\leq n_i/2+1\}$ or $\{\alpha\omega_i\geq n_i/2+1\}$ for $i=1,...,m$. We assume that $0\leq\phi_1,...,\phi_r\leq1$ and $\phi_{r+1},...,\phi_q\geq1$ for fixed $r=1,...,q$. Define a function

\begin{align*}
    g(\phi_1,...,\phi_r,\alpha,\boldsymbol{\omega}):=\lambda^q&\times 
    \pi(\alpha)\times\prod_{i=1}^m [(n_i/2+\alpha\omega_i+1)^{\tilde n_i}\times(\alpha\omega_i)^{\alpha\omega_i+1}\times\pi^{\rm BP}(\omega_i)]\\
    &\times\int_{{[1,\infty)}^{q-r}}~\prod_{k=1}^q \phi_k^{\lambda B_{1k}-1}\left\{C_*+b_*\prod_{k=1}^q \phi_k^{\lambda B_{1k}+\lambda B_{2k}}\right\}^{-1}~d\phi_{r+1}\cdots d\phi_q.
\end{align*}

Hence, let $\mathcal{D}=\{(\alpha,\boldsymbol{\omega})\in \mathbb{R}_+^{m+1}:\alpha\omega_i\leq n_i/2+1,~i=1,...,m \}$, it is sufficient to show that

\begin{align*}
    I_1:=&~\int_{\mathcal{D}}\int_{(0,1]^r}~g(\phi_1,...,\phi_r,\alpha,\boldsymbol{\omega})~d\phi_1\cdots d\phi_r~d\boldsymbol{\omega}d\alpha<\infty
\end{align*}

and

\begin{align*}
    I_2:=&~\int_{\mathcal{D}^c}\int_{(0,1]^r}~g(\phi_1,...,\phi_r,\alpha,\boldsymbol{\omega})~d\phi_1\cdots d\phi_r~d\boldsymbol{\omega}d\alpha<\infty.
\end{align*}

Note that the function $g$ is 0 if at least one of $\phi_1,...,\phi_r,\alpha,\omega_1,....,\omega_m$ is 0. Otherwise,

\begin{align}\label{proof_thm_bound_eta_sugasawa}
    g(\phi_1,.&..,\phi_r,\alpha,\boldsymbol{\omega})\nonumber\\
    &\leq\lambda^q\times 
    \pi(\alpha)\times\prod_{i=1}^m [(n_i/2+\alpha\omega_i+1)^{\tilde n_i}\times(\alpha\omega_i)^{\alpha\omega_i+1}\times\pi^{\rm BP}(\omega_i)]\nonumber\\
    &\hspace{.5cm}\times\int_{{[1,\infty)}^{q-r}}~\prod_{k=1}^q \phi_k^{\lambda B_{1k}-1}\left\{b_*\prod_{k=1}^q \phi_k^{\lambda B_{1k}+\lambda B_{2k}}\right\}^{-1}~d\phi_{r+1}\cdots d\phi_q\nonumber\\
    &=\lambda^q\times\pi(\alpha)\times\prod_{i=1}^m [(n_i/2+\alpha\omega_i+1)^{\tilde n_i}\times(\alpha\omega_i)^{\alpha\omega_i+1}\times\pi^{\rm BP}(\omega_i)]\times b_*^{-1}D_*^{-1}\nonumber\\
    &\hspace{.5cm}\times \prod_{k=1}^r \phi_k^{\lambda B_{1k}-1}\times\prod_{j=r+1}^q\int_1^\infty~\phi_j^{-\lambda B_{2k}-1}~d\phi_j\nonumber\\
    &\leq C_5~D_*^{-1} \pi(\alpha)\times\prod_{i=1}^m\frac{(n_i/2+\alpha\omega_i+1)^{\tilde n_i}}{(\alpha\omega_i)^{\tilde n_i}}\times~\pi^{\rm BP}(\omega_i)\times\prod_{k=1}^r \phi_k^{\lambda B_{1k}-1},
\end{align}

where $D_*=\prod_{k=1}^r\phi_k^{\lambda B_{1k}+\lambda B_{2k}}$. First, we evaluate $I_1$. We make a change of variable $\xi_i=\alpha\omega_i$, so over the domain of $I_1$, we have $\xi_i\in[0,n_i/2+1]$ for $i=1,...,m$. Also, note that using the BP prior,  $\pi^{\rm BP}(\omega_i)\propto\omega_i^{a-1}(1+\omega_i)^{-(a+b)}\leq\omega_i^{a-1}$ for $\omega_i>0$. Hence, from (\ref{proof_thm_bound_eta_sugasawa}), for $0<\phi_1,...,\phi_r\leq1$ and $0<\xi_i\leq n_i/2+1$, we obtain $g(\phi_1,...,\phi_r,\alpha,\boldsymbol{\omega})\leq C_6~D_{**}^{-1}\alpha^{-am}\pi(\alpha)\prod_{i=1}^m\xi_i^{a-1}\prod_{k=1}^r\phi_k^{\lambda B_{1k}-1}$ where $D_{**}=\prod_{i=1}^m\xi_i^{\tilde n_i}D_*$. If we assume that $\int_0^\infty \alpha^{-am}\pi(\alpha)d\alpha$ is finite, the function $g$ is bounded over $[0,1]^r\times[0,n_1/2+1]\times\cdots\times[0,n_m/2+1]$. Therefore, the integral $I_1$ is finite. Next, we evaluate $I_2$. For $\alpha\omega_i\geq n_i/2+1$, the term $(1+(n_i/2+1)/\alpha\omega_i)^{\tilde n_i}<\infty$ in (\ref{proof_thm_bound_eta_sugasawa}) is bounded for $i=1,...,m$. Therefore, if the prior $\pi(\alpha)$ is proper, the integral $I_2$ is finite.

\end{proof}

\clearpage
\section{MCMC algorithms for the existing models}
\label{MCMC_schemes}

\begin{algorithm}
    \caption{MCMC algorithm for the YC model proposed by \cite{2006_You&Chapman}}
    \begin{algorithmic}[1]
    \small
    \State \textbf{Parameters}: $\boldsymbol{\mathcal{V}}=(\boldsymbol{\theta},\boldsymbol{\beta},\boldsymbol{\sigma}^2,\sigma_u^2)$.
    \State \textbf{Input}: initial values $\boldsymbol{\mathcal{V}}^{(0)}$ and observed data $\{n_i,y_i,D_i,\boldsymbol{x}_i\}_{i=1}^m$.
    \State \textbf{Output}: posterior samples $\boldsymbol{\mathcal{V}}^{(s)}$ for $s=1,...,S$.
    \State $a_i\leftarrow0.0001$ and $b_i\leftarrow0.0001$.
    \For {iteration $s=1,2,\ldots,S$}
    \State Sample $\theta_i^{(s)}\sim N\left((1-\gamma_i)y_i+\gamma_i\boldsymbol{x}_i^\top\boldsymbol\beta,\gamma_i\sigma_u^{2}\right)$, $\gamma_i=\frac{\sigma_i^{2}}{\sigma_i^{2}+\sigma_u^{2}}$ independently for $i=1,...,m$.
    \State Sample $\sigma_i^{2(s)}\sim {\rm IG}\left(a_i+\frac{n_i}{2},~b_i+\frac{1}{2}[(y_i-\theta_i)^2+(n_i-1)D_i]\right)$  independently for $i=1,...,m$.
	\State Sample $\boldsymbol{\beta}^{(s)}\sim  N_p\left(\left(\sum_{i=1}^m\boldsymbol{x}_i\boldsymbol{x}_i^\top\right)^{-1}\left(\sum_{i=1}^m\boldsymbol{x}_i\theta_i\right),~\sigma_u^{2}\left(\sum_{i=1}^m\boldsymbol{x}_i\boldsymbol{x}_i^\top\right)^{-1}\right).$
    \State Sample $\sigma_u^{2(s)}\sim {\rm IG}\left(\frac{m}{2}-1,~\frac{1}{2}\sum_{i=1}^m(\theta_i-\boldsymbol{x}_i^\top\boldsymbol{\beta})^2\right)$.
    \EndFor
	\end{algorithmic} 
    \label{alg:YC}
\end{algorithm}

\begin{algorithm}
    \caption{MCMC algorithm for the STK1 model proposed by \cite{2017_Sukasawa}}
    \begin{algorithmic}[1]
    \small
    \State \textbf{Parameters}: $\boldsymbol{\mathcal{V}}=(\boldsymbol{\theta},\boldsymbol{\beta},\boldsymbol{\sigma}^2,\sigma_u^2,\gamma)$.
    \State \textbf{Input}: initial values $\boldsymbol{\mathcal{V}}^{(0)}$ and observed data $\{n_i,y_i,D_i,\boldsymbol{x}_i\}_{i=1}^m$.
    \State \textbf{Output}: posterior samples $\boldsymbol{\mathcal{V}}^{(s)}$ for $s=1,...,S$.
    \State $a_i\leftarrow2$ and $b_i\leftarrow1/n_i$.
    \For {iteration $s=1,2,\ldots,S$}
    \State Sample $\theta_i^{(s)}\sim N\left((1-\gamma_i)y_i+\gamma_i\boldsymbol{x}_i^\top\boldsymbol\beta,\gamma_i\sigma_u^{2}\right)$, $\gamma_i=\frac{\sigma_i^{2}}{\sigma_i^{2}+\sigma_u^{2}}$ independently for $i=1,...,m$.
    \State Sample $\sigma_i^{2(s)}\sim {\rm IG}\left(a_i+\frac{n_i}{2},~b_i\gamma+\frac{1}{2}[(y_i-\theta_i)^2+(n_i-1)D_i]\right)$  independently for $i=1,...,m$.
    \State Sample $\boldsymbol{\beta}^{(s)}\sim  N_p\left(\left(\sum_{i=1}^m\boldsymbol{x}_i\boldsymbol{x}_i^\top\right)^{-1}\left(\sum_{i=1}^m\boldsymbol{x}_i\theta_i\right),~\sigma_u^{2}\left(\sum_{i=1}^m\boldsymbol{x}_i\boldsymbol{x}_i^\top\right)^{-1}\right).$
    \State Sample $\sigma_u^{2(s)}\sim {\rm IG}\left(\frac{m}{2}-1,~\frac{1}{2}\sum_{i=1}^m(\theta_i-\boldsymbol{x}_i^\top\boldsymbol{\beta})^2\right)$.
    \State Sample $\gamma^{(s)}\sim {\rm Ga}\left(1+\sum_{i=1}^m a_i,~\sum_{i=1}^m\frac{b_i}{\sigma_i^{2}}\right)$.
    \EndFor
	\end{algorithmic} 
    \label{alg:STK1}
\end{algorithm}

\begin{algorithm}
    \caption{MCMC algorithm for the STK2 model proposed by \cite{2017_Sukasawa}}
    \begin{algorithmic}[1]
    \small
    \State \textbf{Parameters}: $\boldsymbol{\mathcal{V}}=(\boldsymbol{\theta},\boldsymbol{\beta},\boldsymbol{\sigma}^2,\sigma_u^2,\gamma,\boldsymbol{\eta})$.
    \State \textbf{Input}: initial values $\boldsymbol{\mathcal{V}}^{(0)}$, observed data $\{n_i,y_i,D_i,\boldsymbol{x}_i,\boldsymbol{z}_i\}_{i=1}^m$, and $c^{(1)}>0$.
    \State \textbf{Output}: posterior samples $\boldsymbol{\mathcal{V}}^{(s)}$ for $s=1,...,S$.
    \State $a_i\leftarrow2$ and $b_i\leftarrow1/n_i$. 
    \For {iteration $s=1,2,\ldots,S$}
    \State Sample $\theta_i^{(s)}\sim N\left((1-\gamma_i)y_i+\gamma_i\boldsymbol{x}_i^\top\boldsymbol\beta,\gamma_i\sigma_u^{2}\right)$, $\gamma_i=\frac{\sigma_i^{2}}{\sigma_i^{2}+\sigma_u^{2}}$  independently for $i=1,...,m$.
    \State Sample $\sigma_i^{2(s)}\sim {\rm IG}\left(a_i+\frac{n_i}{2},~b_i\gamma+\frac{1}{2}[(y_i-\theta_i)^2+(n_i-1)D_i]\right)$  independently for $i=1,...,m$.
    \State Sample $\boldsymbol{\beta}^{(s)}\sim  N_p\left(\left(\sum_{i=1}^m\boldsymbol{x}_i\boldsymbol{x}_i^\top\right)^{-1}\left(\sum_{i=1}^m\boldsymbol{x}_i\theta_i\right),~\sigma_u^{2}\left(\sum_{i=1}^m\boldsymbol{x}_i\boldsymbol{x}_i^\top\right)^{-1}\right).$
    \State Sample $\sigma_u^{2(s)}\sim {\rm IG}\left(\frac{m}{2}-1,~\frac{1}{2}\sum_{i=1}^m(\theta_i-\boldsymbol{x}_i^\top\boldsymbol{\beta})^2\right)$.
    \State Sample $\gamma^{(s)}\sim {\rm Ga}\left(1+\sum_{i=1}^m a_i,~\sum_{i=1}^m\frac{b_i}{\sigma_i^{2}}\right)$.
    \State Sample $\boldsymbol{\eta}^{(s)}$ using an adaptive Metropolis Hastings algorithm:
    
        \begin{enumerate}
        \item Sample the proposal $\boldsymbol{\eta}^*\sim N_q(\boldsymbol{\eta}^{(s-1)},c^{(s)}\textbf{I}_q)$ where


        $$
        c^{(s)}=
        \begin{cases}
        c^{(s-1)} *1.1 & \text{if the acceptance rate is $>0.6$}\\
        c^{(s-1)} /1.1 & \text{if the acceptance rate is $<0.4$}
        \end{cases}
        $$


        The variance $c^{(s)}$ is updated every $\ell$ iteration.
        
        \item Compute the acceptance probability 
        \begin{align*}
        {\rm AP}(\boldsymbol{\eta})={\rm min}\left\{1,~\frac{\pi(\boldsymbol{\eta}^*\mid\boldsymbol{\sigma}^2,\gamma)}{\pi(\boldsymbol{\eta}^{(s-1)}\mid\boldsymbol{\sigma}^2,\gamma)}
        \right\},
        \end{align*}
        where $\pi(\boldsymbol{\eta}\mid\boldsymbol{\sigma}^2,\gamma)\propto \prod_{i=1}^m\exp\left(a_i\boldsymbol{z}_i^\top\boldsymbol{\eta}\right)\exp\left(-b_i\gamma\exp(\boldsymbol{z}_i^\top\boldsymbol{\eta})/\sigma_i^2\right)$ is the full conditional distribution of $\boldsymbol{\eta}$.
        \item Sample $u\sim{\rm Uniform}(0,1)$. If $u<{\rm AP}(\boldsymbol{\eta})$, $\boldsymbol{\eta}^{(s)}=\boldsymbol{\eta}^*$. Otherwise,  $\boldsymbol{\eta}^{(s)}=\boldsymbol{\eta}^{(s-1)}$.
        \end{enumerate}
    \EndFor
	\end{algorithmic} 
    \label{alg:STK2}
\end{algorithm}

\clearpage

\section{Estimating  $\nu_i$ using the Delta method in the PEAI--AHS application}

\label{application_sup}

Consider an indicator $y_{ij}\sim{\rm Ber}(\theta_i)$ and let $y_i=n_i^{-1}\sum_{j=1}^{n_i}y_{ij}$ be a direct estimator for the proportion $\theta_i$. Then 
\begin{align*}
    \E(y_i)=\E(y_{ij})=\theta_i\quad{\rm and}\quad\sigma_i^2={\rm Var}(y_i)=\frac{1}{n_i}{\rm Var}(y_{ij})=\frac{\theta_i(1-\theta_i)}{n_i}.
\end{align*}
Thus, the Direct Variance Estimator (DVE) is given by 
\begin{align*}
    D_i=\widehat{\rm Var}(y_i)=\frac{\hat{\theta_i}(1-\hat{\theta_i})}{n_i}=\frac{y_i(1-y_i)}{n_i}.
\end{align*}
Moreover, by the Central Limit Theorem, we have 
\begin{align*}
    y_i\sim N(\theta_i,\theta_i(1-\theta_i)/n_i)\quad{\rm as}\quad n_i\to\infty.
\end{align*}
By the Delta method with a function $g(\theta_i)=\theta_i(1-\theta_i)/n_i$, we obtain
\begin{align*}
    {\rm Var}(D_i)\approx [g'(\theta_i)]^2{\rm Var}(y_i)=\left(\frac{1-2\theta_i}{n_i}\right)^2\frac{\theta_i(1-\theta_i)}{n_i}.
\end{align*}
Since $\nu_iD_i/\sigma_i^2\sim\chi_{\nu_i}^2$, we have ${\rm Var}(D_i)=(\sigma_i^2/\nu_i)^22\nu_i=2\sigma_i^4/\nu_i$. Hence,
\begin{align*}
    \frac{2\sigma_i^4}{\nu_i}=\left(\frac{1-2\theta_i}{n_i}\right)^2\frac{\theta_i(1-\theta_i)}{n_i}\quad\Longrightarrow\quad\nu_i=\frac{2\sigma_i^4n_i^3}{(1-2\theta_i)^2\theta_i(1-\theta_i)}.
\end{align*}
Since $\sigma_i^2={\rm Var}(y_i)=\theta_i(1-\theta_i)/n_i$, we have
\begin{align*}
    \nu_i=\frac{2[\theta_i(1-\theta_i)/n_i]^2n_i^3}{(1-2\theta_i)^2\theta_i(1-\theta_i)}=\frac{2n_i\theta_i(1-\theta_i)}{(1-2\theta_i)^2}.
\end{align*}
Therefore, the estimated degrees of freedom for  the PEAI–AHS application are given by
\begin{align*}
    \hat{\nu}_i=\frac{2n_iy_i(1-y_i)}{(1-2y_i)^2}.
\end{align*}


\end{document}